\documentclass[11pt, reqno]{amsart}

\usepackage{graphicx,amsmath,amssymb,epsfig,harvard}
\usepackage{colortbl}
\long\def\comment#1{}
\long\def\comment#1{}
\newtheorem{theorem}{Theorem}
\newtheorem{algorithm}{Algorithm}[section]

\newtheorem{corollary}{Corollary}
\newtheorem{lemma}{Lemma}

\theoremstyle{definition}

\newtheorem{remark}{Comment}[section]
\numberwithin{remark}{section}

\newcommand{\citen}{\citeasnoun}

\newcommand{\be}{\begin{eqnarray}}
\newcommand{\ee}{\end{eqnarray}}

\newcommand{\ba}{\begin{array}}
\newcommand{\ea}{\end{array}}
\newcommand{\bs}{\begin{align}\begin{split}\nonumber}
\newcommand{\bsnumber}{\begin{align}\begin{split}}
\newcommand{\es}{\end{split}\end{align}}

\renewcommand{\(}{\left(}
\renewcommand{\)}{\right)}
\renewcommand{\[}{\left[}
\renewcommand{\]}{\right]}
\renewcommand{\hat}{\widehat}

\newcommand{\Gn}{\mathbb{G}_n}
\newcommand{\Gnk}{\mathbb{G}_{n_k}}

\newcommand{\Pn}{\mathbb{P}_n}

\newcommand{\Ep}{{\mathrm{E}}}
\newcommand{\barEp}{\bar \Ep}
\newcommand{\En}{{\mathbb{E}_n}}
\newcommand{\Ena}{{\mathbb{E}_{n_a}}}
\newcommand{\Enb}{{\mathbb{E}_{n_b}}}
\newcommand{\Enk}{{\mathbb{E}_{n_k}}}
\newcommand{\LASSO}{\sqrt{\textrm{Lasso}}}

\renewcommand{\Pr}{{\mathrm{P}}}

\def\RR{ {\mathbb{R}}}

\def\z{{d}}
\def\x{{x}}
\def\supp{{\rm support}}

\newcommand{\ceil}[1]{\left\lceil #1 \right\rceil}
\newcommand{\semin}[1]{\phi_{{\rm min}}(#1)}
\newcommand{\semax}[1]{\phi_{{\rm max}}(#1)}
\renewcommand{\hat}{\widehat}
\renewcommand{\leq}{\leqslant}
\renewcommand{\geq}{\geqslant}
\newcommand{\sign}{ {\rm sign}}

\begin{document}

\title[Sparse Models and Methods for Optimal Instruments]
{Sparse Models and Methods for Optimal Instruments with an Application to Eminent Domain}
\author[Belloni \ Chen \ Chernozhukov \ Hansen]{A. Belloni \and D. Chen \and V. Chernozhukov \and C. Hansen}

\date{First version:  June 2009,  this version \today. This is a revision of an October 2010 ArXiv/CEMMAP paper with the same title.  Preliminary results of this paper were first presented at Chernozhukov's invited Cowles Foundation lecture at the Northern American meetings
of the Econometric society in June of 2009.  We thank seminar participants at Brown, Columbia, Hebrew University, Tel Aviv University, Harvard-MIT, the Dutch Econometric Study Group, Fuqua School of Business, NYU, and New Economic School,
for helpful comments. We also thank Denis Chetverikov, JB Doyle, and Joonhwan Lee for thorough reading
of this paper and very useful comments.}

\begin{abstract}

We develop results for the use of Lasso and Post-Lasso methods to form first-stage predictions and estimate optimal instruments in linear instrumental variables (IV) models with many instruments, $p$.  Our results apply even when $p$ is much larger than the sample size, $n$. We show that the IV estimator based on using Lasso or Post-Lasso in the first stage is root-n consistent and asymptotically normal when the first-stage is approximately sparse; i.e. when the conditional expectation of the endogenous variables given the instruments can be well-approximated by a relatively small set of variables whose identities may be unknown.  We also show the estimator is semi-parametrically efficient when the structural error is homoscedastic.   Notably our results allow for imperfect model selection, and do not rely upon the unrealistic "beta-min" conditions that are widely used to establish validity of inference following model selection.  In simulation experiments, the Lasso-based IV estimator with a data-driven penalty performs well compared to recently advocated many-instrument-robust procedures. In an empirical example dealing with the effect of judicial eminent domain decisions on economic outcomes, the Lasso-based IV estimator outperforms an intuitive benchmark.

Optimal instruments are conditional expectations.  In developing the IV results, we establish a series of new results for Lasso and Post-Lasso estimators of nonparametric conditional expectation functions which are of independent theoretical and practical interest.
We construct a modification of Lasso designed to deal with non-Gaussian, heteroscedastic disturbances which uses a data-weighted $\ell_1$-penalty function. By innovatively using moderate deviation theory for self-normalized sums, we provide convergence rates for the resulting Lasso and Post-Lasso estimators that are as sharp as the corresponding rates in the homoscedastic Gaussian case under the condition that $\log p = o(n^{1/3})$.  We also provide a data-driven method for choosing the penalty level that must be specified in obtaining Lasso and Post-Lasso estimates and establish its asymptotic validity under non-Gaussian, heteroscedastic disturbances. \\

\end{abstract}

\maketitle

\section{Introduction}

Instrumental variables (IV) techniques are widely used in applied economic research.  While these methods provide a useful tool for identifying structural effects of interest, their application often results in imprecise inference.  One way to improve the precision of instrumental variables estimators is to use many instruments or to try to approximate the optimal instruments as in \citen{amemiya:optimalIV}, \citen{chamberlain}, and \citen{newey:optimaliv}.  Estimation of optimal instruments will generally be done nonparametrically and thus implicitly makes use of many constructed instruments such as polynomials.  The promised improvement in efficiency is appealing, but  IV estimators based on many instruments may have poor properties.  See, for example, \citen{bekker}, \citen{chao:swanson:05}, \citen{hhn:weakiv}, and \citen{neweyetal:hetmanyiv} which propose solutions for this problem based on ``many-instrument'' asymptotics.\footnote{It is important to note that the precise definition of ``many-instrument" is $p \propto n$ with $p< n$ where $p$ is the number
of instruments and $n$ is the sample size.  The current paper allows for this case and also for ``very many-instrument"
asymptotics where $p \gg n$.}

In this paper, we contribute to the literature on IV estimation with many instruments by considering the use of Lasso and Post-Lasso for estimating the first-stage regression of endogenous variables on the instruments.   Lasso is a widely used method
that acts both as an estimator of regression functions and as a model selection device. Lasso solves for regression coefficients by minimizing the sum of the usual least squares objective function and a penalty for model size through the sum of the absolute values of the coefficients.
The resulting Lasso estimator selects instruments and estimates the first-stage regression coefficients via a shrinkage procedure. The Post-Lasso estimator discards the Lasso coefficient estimates and uses the data-dependent set of instruments selected by Lasso
to refit the first stage regression via OLS to alleviate Lasso's shrinkage bias.
For theoretical and simulation evidence regarding Lasso's  performance,  see Bai and Ng \citeyear{BaiNg2008,BaiNg2009b}, \citen{BickelRitovTsybakov2009}, Bunea, Tsybakov, and Wegkamp \citeyear{BuneaTsybakovWegkamp2006,BuneaTsybakovWegkamp2007b,BuneaTsybakovWegkamp2007},  \citen{CandesTao2007}, \citen{horowitz:lasso}, \citen{knight:shrinkage}, \citen{Koltchinskii2009}, \citen{Lounici2008}, \citen{LouniciPontilTsybakovvandeGeer2009}, \citen{MY2007}, \citen{RosenbaumTsybakov2008}, \citen{T1996}, \citen{vdGeer}, \citen{Wainright2006},   \citen{ZhangHuang2006}, \citen{BC-PostLASSO}, and \citen{BuhlmannGeer2011} among many others.  See \citen{BC-PostLASSO} for analogous results on Post-Lasso.

Using Lasso-based methods to form first-stage predictions in IV estimation provides a practical approach to obtaining the efficiency gains from using optimal instruments while dampening the problems associated with many instruments.  We show that Lasso-based procedures produce first-stage predictions that provide good approximations to the optimal instruments even when the number of available instruments is much larger than the sample size when the first-stage is approximately sparse -- that is, when there exists a relatively small set of important instruments whose identities are unknown that well-approximate the conditional expectation of the endogenous variables given the instruments.
Under approximate sparsity, estimating the first-stage relationship using Lasso-based procedures produces IV estimators that are root-n consistent and asymptotically normal.  The IV estimator with Lasso-based first stage also achieves the semi-parametric efficiency bound under the additional condition that structural errors are homoscedastic.   Our results allow imperfect model selection and do not impose ``beta-min'' conditions that restrict the minimum allowable magnitude of the coefficients on relevant regressors.  We also provide a consistent asymptotic variance estimator.
  Thus, our results generalize the IV procedure of \citen{newey:optimaliv} and \citen{hahn:SeriesRate} based on conventional series approximation of the optimal instruments.  Our results also generalize \citen{BickelRitovTsybakov2009} by providing inference and confidence sets for the second-stage IV estimator based on  Lasso or Post-Lasso estimates of the first-stage predictions.  To our knowledge, our result is the first to verify root-n consistency and asymptotic normality of an estimator for a low-dimensional structural parameter
  in a high-dimensional setting without imposing the very restrictive ``beta-min" condition.\footnote{The ``beta-min" condition requires the relevant coefficients in the regression to be separated from zero by a factor that exceeds the potential estimation error.  This condition implies the identities of the relevant regressors may be perfectly determined.  There is a large body of theoretical work that uses such a condition and thus implicitly assumes that the resulting post-model selection estimator is the same as the oracle estimator that knows the identities of the relevant regressors. See \citen{BuhlmannGeer2011} for the discussion of the ``beta-min" condition and the theoretical role it plays in obtaining ``oracle" results.}
   Our results also remain valid in the presence of heteroscedasticity and thus provide a useful complement to existing approaches in the many instrument literature which often rely on homoscedasticity and may be inconsistent in the presence of heteroscedasticity; see \citen{neweyetal:hetmanyiv} for a notable exception that allows for heteroscedasticity and gives additional discussion.

 Instrument selection procedures complement existing/traditional methods that are meant to be robust to many-instruments but are not a universal
 solution to the many instruments problem.  The good performance of instrument selection procedures relies on approximate sparsity.  Unlike traditional IV methods, instrument selection procedures do not require the identity of these ``important'' variables to be known \textit{a priori} as the identity of these instruments will be estimated from the data.   This flexibility comes with the cost that instrument selection will tend not to work well when the first-stage is not approximately sparse.  When approximate sparsity breaks down, instrument selection procedures may select too few or no instruments or may select too many instruments.  Two scenarios where this failure is likely to occur are the weak-instrument case; e.g. \citen{ss:weakiv}, \citen{ams:optimal}, \citen{andrews:stock}, \citen{moreira:weakivlr}, \citen{kleib:weakiv}, and \citen{kleib:weakgmm}; and the many-weak-instrument case; e.g. \citen{bekker}, \citen{chao:swanson:05}, \citen{hhn:weakiv}, and \citen{neweyetal:hetmanyiv}.  We consider two modifications of our basic procedure aimed at alleviating these concerns.  In Section 4, we present a sup-score testing procedure that is related to \citen{anderson:rubin} and \citen{ss:weakiv} but is better suited to cases with very many instruments; and we consider a split sample IV estimator in Section 5 which combines instrument selection via Lasso with the sample-splitting method of \citen{AngristKruegerSplitSample1995}.  While these two procedures are steps toward addressing weak identification concerns with very many instruments, further exploration of the interplay between weak-instrument or many-weak-instrument methods and variable selection would be an interesting avenue for additional research.

Our paper also contributes to the growing literature on Lasso-based methods by providing results for Lasso-based estimators of nonparametric conditional expectations.  We consider a modified Lasso estimator with penalty weights designed to deal with non-Gaussianity
and heteroscedastic errors.  This new construction allows us to innovatively use the results of moderate deviation theory for self-normalized sums of \citen{jing:etal} to provide convergence rates for Lasso and Post-Lasso.  The derived convergence rates are as sharp as in the homoscedastic Gaussian case under the weak condition that the log of the number of regressors  $p$ is small relative to $n^{1/3}$, i.e. $\log p = o(n^{1/3})$.  Our construction generalizes the standard Lasso estimator of \citen{T1996} and allows us to generalize the Lasso results of \citen{BickelRitovTsybakov2009} and Post-Lasso results of \citen{BC-PostLASSO} both of which assume homoscedasticity and Gaussianity.  The construction as well as theoretical results are important for applied economic analysis where researchers are concerned about heteroscedasticity and non-Gaussianity in their data.  We also provide a data-driven method for choosing the penalty that must be specified to obtain Lasso and Post-Lasso estimates, and we establish its asymptotic validity allowing for non-Gaussian, heteroscedastic disturbances.  Ours is the first paper to provide such a data-driven penalty which was previously not available even in the Gaussian case.\footnote{One exception
is the work of \citen{BCW-SqLASSO} which considers square-root-Lasso estimators and shows that their use allows for pivotal penalty choices.  Those results strongly rely on homoscedasticity.} These results are of independent interest in a variety of theoretical and applied settings.

We illustrate the performance of Lasso-based IV through simulation experiments.  In these experiments,
we find that a feasible Lasso-based procedure that uses our data-driven penalty performs well across a range of
simulation designs where sparsity is a reasonable approximation.  In terms of estimation risk, it outperforms the estimator of \citen{fuller}
(FULL),\footnote{Note that this procedure is only applicable when the number of instruments $p$ is less than the sample size $n$. As mentioned earlier, procedures developed in this paper allow for $p$ to be much larger $n$.}
which is robust to many instruments (e.g. \citename{hhn:weakiv}, \citeyear*{hhn:weakiv}), except in a design where sparsity breaks down and the sample size is large relative to the number of instruments.
In terms of size of 5\% level tests, the Lasso-based IV estimator performs comparably to or better than FULL in
all cases we consider.  Overall, the simulation results are in line with the theory and favorable to the proposed Lasso-based IV procedures.

Finally, we demonstrate the potential gains of the Lasso-based procedure in an application where there are many available instruments among which there is not a clear \textit{a priori} way to decide which instruments to use.  We look at the effect of judicial decisions at the federal circuit court level regarding the government's exercise of eminent domain on house prices and state-level GDP as in \citen{chen:yeh:takings}.  We follow the identification strategy of \citen{chen:yeh:takings} who use the random assignment of judges to three judge panels that are then assigned to eminent domain cases to justify using the demographic characteristics of the judges on the realized panels as instruments for their decision.  This strategy produces a situation in which there are many potential instruments in that all possible sets of characteristics of the three judge panel are valid instruments.  We find that the Lasso-based estimates using the data-dependent penalty produce much larger first-stage Wald-statistics and generally have smaller estimated second stage standard errors than estimates obtained using the baseline instruments of \citen{chen:yeh:takings}.

\textbf{Relationship to econometric literature on variable selection and shrinkage.} The idea of instrument selection goes back to \citen{KloekMennes1960} and \citen{Amemiya1966} who searched
among principal components to approximate the optimal instruments. Related ideas appear in dynamic factor models as in \citen{BaiNg2010}, \citen{KapetaniosMarcellino2010}, and  \citen{KapetaniosKhalafMarcellino2011}.  Factor analysis differs from our approach though principal components, factors, ridge fits, and other functions of the instruments could be considered among the set of potential instruments to select from.\footnote{Approximate sparsity should be understood to be relative to a given structure defined by the set of instruments considered.  Allowing for principle components or ridge fits among the potential regressors considerably expands the applicability of the approximately sparse framework.}

There are several other papers that explore the use of modern variable selection methods in econometrics, including some papers that apply these procedures to IV estimation. \citen{BaiNg2009a} consider an approach to instrument selection that is closely related to ours based on boosting. The latter method is distinct from Lasso, cf. \citen{BuhlmannBoosting2006}, but it also does not rely on knowing the identity of the most  important instruments. They show through simulation examples that instrument selection via boosting works well in the designs they consider but do not provide formal results.
\citen{BaiNg2009a} also expressly mention the idea of using the Lasso method for instrument selection, though they focus their analysis on the boosting method. Our paper complements their analysis by providing a formal set of conditions under which Lasso variable selection will provide good first-stage predictions and providing theoretical estimation and inference results for the resulting IV estimator.  One of our theoretical results for the IV estimator is also sufficiently general to cover the use of any other first-stage variable selection procedure, including boosting, that satisfies a set of provided rate conditions.   \citen{caner:LGMM} considers estimation by penalizing the GMM criterion function by the $\ell_{\gamma}$-norm of the coefficients for $0 < \gamma < 1$.  The analysis of \citen{caner:LGMM} assumes that the number of parameters $p$ is fixed in relation to the sample size, and so it is complementary to our approach where we allow $p \to \infty$ as $n \to \infty$.
 Other uses of Lasso in econometrics include \citen{BaiNg2008},  \citen{BCH-PLM}, \citen{Brodie:etal2009}, \citen{DeMiguel2009},  \citen{horowitz:lasso}, \citen{knight:shrinkage}, and others.  An introductory treatment of this topic is given in \citen{BC-Intro}, and \citen{BCC-Review} provides a review of Lasso targeted at economic applications.

Our paper is also related to other shrinkage-based approaches to dealing with many instruments.  \citen{chamberlain:imbens:reiv} considers IV estimation with many instruments using a shrinkage estimator based on putting a random coefficients structure over the first-stage coefficients in a homoscedastic setting.  In a related approach, \citen{okui:manyiv} considers the use of ridge regression for estimating the first-stage regression in a homoscedastic framework where the instruments may be ordered in terms of relevance.  \citen{okui:manyiv} derives the asymptotic distribution of the resulting IV estimator and provides a method for choosing the ridge regression smoothing parameter that minimizes the higher-order asymptotic mean-squared-error (MSE) of the IV estimator.  These two approaches are related to the approach we pursue in this paper in that both use shrinkage in estimating the first-stage but differ in the shrinkage methods they use. Their results are also only supplied in the context of homoscedastic models.  \citen{donald:newey} consider a variable selection procedure that minimizes higher-order asymptotic MSE which relies on \textit{a priori} knowledge that allows one to order the instruments in terms of instrument strength.  Our use of Lasso as a variable selection technique does not require any \textit{a priori} knowledge about the identity of the most relevant instruments and so provides a useful complement to \citen{donald:newey} and \citen{okui:manyiv}.  \citen{carrasco:regularizedIV} provides an interesting approach to IV estimation with many instruments based on directly regularizing the inverse that appears in the definition of the 2SLS estimator; see also \citen{carrasco:LIML}.  \citen{carrasco:regularizedIV} considers three regularization schemes, including Tikhohov regularization which corresponds to ridge regression, and shows that the regularized estimators achieve the semi-parametric efficiency bound under some conditions. \citen{carrasco:regularizedIV}'s approach implicitly uses $\ell_2$-norm penalization and hence differs from and complements our approach.  A valuable feature of \citen{carrasco:regularizedIV} is the provision of a data-dependent method for choosing the regularization parameter based on minimizing higher-order asymptotic MSE following \citen{donald:newey} and \citen{okui:manyiv}.  Finally, in work that is more recent that the present paper, \citen{GautierTsybakovHDIV} consider the important case where the structural equation in an instrumental variables model is itself very high-dimensional and propose a new estimation method related to the Dantzig selector and the square-root-Lasso. They also provide an interesting inference method which differs from the one we consider.

\textbf{Notation.} In what follows, we work with triangular array data $\{\(z_{i,n}, i=1,...,n\), n=1,2,3,...\}$
defined on some common probability space $(\Omega, \mathcal{A}, \mathrm{P})$. Each  $z_{i,n}= (y_{i,n}', x_{i,n}', d_{i,n}')'$
is a vector, with components defined below in what follows, and these vectors are i.n.i.d. -- independent across $i$, but not necessarily identically distributed. The law $\Pr_n$ of $\{z_{i,n}, i=1,...,n\}$ can change with $n$, though we do not make explicit use of $\Pr_n$.  Thus, all parameters that characterize
the distribution of  $\{z_{i,n}, i=1,...,n\}$ are
implicitly indexed by the sample size $n$, but we omit the index $n$ in what follows to simplify
notation.  We use triangular array asymptotics to better capture some finite-sample phenomena and to retain the robustness
of conclusions to perturbations of the data-generating process. We also use the following empirical process notation, $\En[f] := \En[f(z_i)] := \sum_{i=1}^n f(z_i)/n$,  and $\Gn(f) := \sum_{i=1}^n ( f(z_i)
- \Ep[f(z_i)] )/\sqrt{n}.$
Since we want to deal with i.n.i.d. data, we also introduce the average expectation operator:
$
\barEp[f] := \Ep \En[f] =  \Ep \En[f(z_i)] = \sum_{i=1}^n \Ep[f(z_i)]/n.
$
The $\ell_2$-norm is denoted by
$\|\cdot\|_2$, and the $\ell_0$-norm, $\|\cdot\|_0$, denotes the number of non-zero components of a vector.  We use $\|\cdot\|_{\infty}$ to denote the maximal element of a vector.  The empirical $L^2(\Pn)$ norm
of a random variable $W_i$ is defined as
$
\|W_i\|_{2,n} := \sqrt{\En[ W_i^2]}.
$
When the empirical $L^2(\Pn)$ norm is applied to regressors $f_1,\ldots,f_p$ and a vector $\delta\in\RR^p$, $\|f_i'\delta\|_{2,n}$, it is called the prediction norm. Given a vector $\delta \in \RR^p$ and a set of
indices $T \subset \{1,\ldots,p\}$, we denote by $\delta_T$ the vector in which $\delta_{Tj} = \delta_j$ if $j\in T$, $\delta_{Tj}=0$ if $j \notin T$. We also denote $T^c:=\{1,2,\ldots,p\}\setminus T$.
 We use the notation $(a)_+ = \max\{a,0\}$, $a \vee b = \max\{ a, b\}$ and $a \wedge b = \min\{ a , b \}$. We also use the notation $a \lesssim b$ to denote $a \leqslant c b$ for some constant $c>0$ that does not depend on $n$; and
$a\lesssim_{\mathrm{P}} b$ to denote $a=O_{\mathrm{P}}(b)$. For an event $E$, we say that $E$ wp $\to$ 1 when $E$ occurs with probability approaching one as $n$ grows.
We say $X_n =_d Y_n + o_{\mathrm{P}}(1)$ to mean that $X_n$ has the same distribution as $Y_n$ up to a term $o_{\mathrm{P}}(1)$ that vanishes in probability.

\section{Sparse Models and Methods for Optimal Instrumental Variables}

In this section of the paper, we present the model and provide an overview of the main results. Sections 3 and 4 provide a technical presentation that includes a set of sufficient regularity conditions, discusses their plausibility, and establishes the main formal results of the paper.

\subsection{The IV Model and Statement of The Problem} The model is $y_i = \z_i'\alpha_0 + \epsilon_i$ where $\alpha_0$ denotes the true value of a vector-valued parameter $\alpha$. $y_i$ is the response variable, and $d_i$ is a finite
$k_d$-vector of variables whose first $k_e$ elements contain endogenous variables.  The disturbance $\epsilon_i$ obeys for all $i$ (and $n$):
$$
\Ep[\epsilon_i|\x_i] = 0,
$$
where $\x_i$ is a $k_{\x}$-vector of instrumental variables.

As a motivation,
suppose that the structural disturbance is conditionally homoscedastic, namely, for all $i$,
$
\Ep[\epsilon_i^2|\x_i] = \sigma^2.
$
Given a $k_d$-vector of instruments  $A(\x_i)$, the standard IV estimator of $\alpha_0$ is given by
$
\hat \alpha = (\En[A(\x_i)d_i'])^{-1} \En [A(\x_i)y_i],
$
where $\{(\x_i, d_i, y_i), i=1,...,n\}$ is an i.i.d. sample from the IV model above. For a given $A(\x_i)$,
$
\sqrt{n}(\hat \alpha - \alpha_0) =_d N(0, Q^{-1}_0 \Omega_0 {Q_0^{-1}}' ) + o_{\mathrm{P}}(1)$, where $Q_0 = \barEp [A(\x_i)d_i']$ and
 $\Omega_0 =  \sigma^2 \barEp [A(\x_i) A(\x_i)']$  under standard conditions.  Setting
$
A(\x_i) = D(\x_i) = \Ep[d_i|\x_i]$
minimizes the asymptotic variance which becomes $$\Lambda^* = \sigma^2 \{\barEp [D(\x_i) D(\x_i)'] \}^{-1},$$ the semi-parametric efficiency bound for estimating $\alpha_0$; see \citen{amemiya:optimalIV}, \citen{chamberlain},  and \citen{newey:optimaliv}.  In practice, the optimal instrument $D(x_i)$ is an unknown function and has to be estimated.  In what follows,
we investigate the use of sparse methods -- namely
Lasso and Post-Lasso -- for use in estimating the optimal instruments.  The resulting IV estimator is asymptotically as efficient as the infeasible optimal IV estimator above.

Note that if $d_i$ contains exogenous components $w_i$, then
$d_i = (d_{i1},..., d_{ik_e}, w_i')'$ where the first $k_e$ variables
are endogenous.  Since the rest of the components $w_i$ are exogenous, they appear
in $\x_i = (w_i', \tilde x_i')'$.  It follows that
$
D_i := D(x_i) := \Ep[d_i|\x_i]=  (\Ep[d_{1}|\x_i],..., \Ep[d_{k_e}|\x_i], w_i')';
$
i.e. the estimator of $w_i$ is simply $w_i$.  Therefore, we discuss
estimation of  the conditional expectation functions:
$$
D_{il} := D_{l}(\x_i) := \Ep[d_{{l}}|\x_i], \ {l} =1,...,k_e.
$$
In what follows, we focus on the strong instruments case which translates into the assumption that $Q = \barEp[D(\x_i)D(\x_i)']$ has eigenvalues bounded away from zero and from above. We also present an inference procedure that remains valid in the absence of strong instruments which is related to \citen{anderson:rubin} and \citen{ss:weakiv} but allows for $p \gg n$.

\subsection{Sparse Models for  Optimal Instruments and Other Conditional Expectations}
Suppose there is a very large list of instruments,
$$
f_i := (f_{i1},...,f_{ip})' := ( f_1(\x_i),...,f_p(\x_i))',
$$
to be used in estimation of conditional expectations $D_{{l}}(x_i), \ {l} =1,...,k_e$, where
the number of instruments $p$ is possibly much larger than the sample size $n$.

For example, high-dimensional instruments $f_i$ could arise as any combination of the following two cases.  First, the list of available instruments may simply be large, in which case $f_i=x_i$ as in e.g. \citen{amemiya:optimalIV} and \citen{bekker}.  Second, the list $f_i$ could consist of a large number of series terms with respect to some elementary regressor vector $x_i$; e.g., $f_i$ could be composed of B-splines, dummies, polynomials, and various interactions as in  \citen{newey:optimaliv} or \citen{hahn:SeriesRate} among others.  We term the first example the many instrument case and the second example the many series instrument case and note that our formulation does not require us to distinguish between the two cases.
We mainly use the term ``series instruments"  and  contrast
  our results with those in the seminal work of \citen{newey:optimaliv} and \citen{hahn:SeriesRate}, though our results
  are not limited to canonical series regressors as in \citen{newey:optimaliv} and \citen{hahn:SeriesRate}. The
  most important feature of our approach is that by allowing $p$ to be much larger than the sample size,
   we are able to consider many more series instruments than in \citen{newey:optimaliv} and \citen{hahn:SeriesRate} to approximate the optimal instruments.

The key assumption that allows effective use of this large set of instruments is sparsity.
To fix ideas, consider the case where $D_{{l}}(\x_i)$ is
 a function  of only $s \ll n$ instruments:
\begin{equation}\label{eq: sparse instrument0} 
\begin{array}{rcl}
&& D_{{l}}(\x_i) = f_i' \beta_{{l0}},   \ \ {l}=1,...,k_e, \\
&&\max_{1\leq{l}\leq k_e} \|\beta_{{l0}}\|_0 = \max_{1\leq{l}\leq k_e} \sum_{j=1}^{p}1\{ \beta_{{l}0 j} \neq 0 \} \leq  s \ll n.
\end{array}\end{equation}
This simple sparsity model generalizes the classic parametric model of optimal instruments of \citen{amemiya:optimalIV} by letting
the identities of the relevant instruments
$
T_{{l}} = \text{support}(\beta_{l0})=\{ j \in \{1,\ldots,p\} \ : \ |\beta_{{l0} j}| > 0\}
$
be unknown.

The model given by (\ref{eq: sparse instrument0}) is unrealistic in that it presumes exact sparsity.  We make no formal use of this model,
but instead use a much more general approximately sparse or nonparametric model:

\textbf{Condition AS.}(\textbf{Approximately Sparse Optimal Instrument}).   \textit{Each optimal instrument function $D_{{l}}(\x_i)$ is well-approximated by a function of unknown $s\geq 1$ instruments:}
\begin{equation}\label{eq: sparse instrument}
\begin{array}{rcl}
&& D_{{l}}(\x_i) = f_i' \beta_{{l} 0} + a_{{l}}(x_i),   \ \ \ \ \ {l}=1,...,k_e, \ \ k_e \ \mbox{fixed,}
\\
&& \max_{1\leq{l}\leq k_e} \|\beta_{{l} 0}\|_0  \leq  s = o(n), \ \ \ \max_{1\leq{l}\leq k_e} [\En a_{{l}}(x_i)^2]^{1/2} \leq c_{s} \lesssim_{\mathrm{P}} \sqrt{s/n}.
\end{array}
\end{equation}

Condition AS is the key assumption.  It requires that there are at most $s$ terms for each endogenous variable that are able to approximate the conditional expectation function $D_{{l}}(\x_i)$ up to approximation error $a_{{l}}(x_i)$ chosen
to be no larger than the conjectured size $\sqrt{s/n}$ of the estimation error of the infeasible estimator that knows the identity of these important variables, the ``oracle estimator.''  In other words, the number $s$ is
defined so that the approximation error is of the same order as the estimation error,
$\sqrt{s/n}$, of the oracle estimator.   Importantly, the assumption allows the identity
$$ T_{{l}} = \text{support}(\beta_{l0})$$
to be unknown and to differ for $l=1,\ldots,k_e$.

For a detailed motivation and discussion of this assumption, we refer the reader to \citen{BCC-Review}. Condition AS
generalizes  the conventional series approximation of optimal instruments in Newey \citeyear{newey:optimaliv,newey:series} and \citen{hahn:SeriesRate} by letting the identities
of the most important $s$ series terms $T_{{l}}$ be unknown. The rate
$\sqrt{s/n}$ generalizes the rate obtained with the optimal number $s$ of series terms in \citen{newey:optimaliv} for estimating conditional expectation by not relying on knowledge of what $s$ series terms to include. Knowing
the identities of the most important series terms is unrealistic
in many examples. The most important series terms
need not be the first $s$ terms, and the optimal number
of series terms to consider is also unknown. Moreover, an optimal approximation
could come from the combination of completely different bases e.g
by using both polynomials and B-splines.

Lasso and Post-Lasso use the data to estimate the set of the most relevant series terms in a manner that allows the resulting IV estimator
to achieve good performance  if a key growth condition,
$$
 \frac{s^2 \log^2 (p\vee n)}{n} \to 0,
$$
holds along with other more technical conditions.  The growth condition requires the optimal
instruments to be sufficiently smooth  so that a small (relative to $n$) number of series terms can be used
to approximate them well.  The use of a small set of instruments ensures that the impact of first-stage estimation on the IV estimator is asymptotically negligible. We can weaken this condition to $s \log (p\vee  n) = o(n)$ by using the sample-splitting idea from the many instruments
literature.

\subsection{Lasso-Based Estimation Methods for Optimal Instruments and Other Conditional Expectation Functions}

Let us write the first-stage regression equations as
\begin{equation}
\label{eq: k reg equations}
d_{i{l}} = D_{l}(\x_i) + v_{i{l}}, \ \  \Ep[v_{i{l}}|\x_i] =0, \ \ {l} = 1,..., k_e.
\end{equation}
Given the sample $\{(\x_i, d_{i{l}}, {l}=1,..., k_e), i=1,...,n\}$, we consider
estimators of the optimal instrument $D_{il}:= D_{{l}}(\x_i)$ that take the form
$$
\widehat D_{il}:=\widehat D_{l}(\x_i) = f_i'\hat \beta_{l}, \ \ {l}=1,...,k_e,
$$
where $\hat \beta_{{l}}$ is the Lasso or Post-Lasso estimator obtained by using $d_{il}$ as the dependent variable
and $f_i$ as regressors.

Consider the usual least squares criterion function:
$$
\hat Q_{{l}} (\beta):= \En [(d_{i{l}} - f_i'\beta)^2].
$$
The Lasso estimator is defined as a solution of the following optimization program:
\begin{equation}\label{Def:LASSOmain}
\widehat \beta_{{l} L} \in \arg \min_{\beta \in \mathbb{R}^p} \hat Q_{{l}} (\beta) + \frac{\lambda}{n} \|\hat \Upsilon_l \beta \|_{1}
\end{equation}
where  $\lambda$ is the penalty level and  $\hat \Upsilon_l = \text{diag}(\hat \gamma_{l1},..., \hat \gamma_{lp})$ is a diagonal matrix specifying  penalty loadings.

Our analysis will first employ the following ``ideal" penalty loadings:
$$
\hat \Upsilon_{{l}}^0 = \text{diag}(\hat \gamma^0_{{l} 1},..., \hat \gamma^0_{{l} p}),  \ \hat \gamma^0_{lj} = \sqrt{ \En [ f^2_{ij} v^2_{il}]}, \ j =1,...,p.
$$
The ideal  option is not feasible but leads to rather sharp theoretical bounds on estimation risk. This option is not feasible since $v_{il}$ is not observed.  In practice, we estimate the
ideal loadings by first using conservative penalty loadings and
then plugging-in the resulting estimated residuals in place of $v_{il}$
to obtain the refined loadings. This procedure could be iterated
via Algorithm A.1 stated in the appendix.

The idea behind the ideal penalty loading is to introduce self-normalization of the first-order condition of the Lasso problem by using data-dependent penalty loadings. This self-normalization allows us to apply moderate deviation theory of \citen{jing:etal} for self-normalized sums to bound deviations of the maximal element of the score vector
$$S_{l} =  2 \En [(\widehat \Upsilon^{0}_l)^{-1} f_i v_{il}]$$
which provides a representation of the estimation noise in the problem.
Specifically, the use of self-normalized moderate deviation theory allows us to establish that
\begin{equation}\label{eq: score bound}
\Pr \left( \sqrt{n} \max_{1 \leq {l} \leq k_e } \|S_{l}\|_{\infty} \leq  2\Phi^{-1}(1- \gamma/(2k_ep) ) \right) \geq 1- \gamma + o(1),
\end{equation}
from which we obtain sharp convergence results for the Lasso estimator under non-Gaussianity and heteroscedasticity.
Without using these loadings, we may not be able to achieve the same sharp rate of convergence.  It is important to emphasize that our construction of the penalty loadings for Lasso is new and differs from the canonical penalty loadings proposed in \citen{T1996} and \citen{BickelRitovTsybakov2009}.
Finally, to insure the good performance of the Lasso estimator, one needs to select the penalty level
   $\lambda/n$ to dominate the noise for all $k_e$ regression problems simultaneously; i.e. the penalty level should satisfy
\begin{equation}\label{the principle}
\Pr \left( \lambda/n \geq  c \max_{1 \leq {l} \leq k_e } \|S_{l}\|_{\infty} \right) \to 1,
 \end{equation}
for some constant $c>1$.
The bound (\ref{eq: score bound}) suggests that this can be achieved by selecting
\begin{equation}\label{the principle lambda}
\lambda = c 2 \sqrt{n} \Phi^{-1}(1- \gamma/(2k_ep) ), \ \ \mbox{with} \ \gamma\to 0, \ \ \log(1/\gamma) \lesssim \log (p\vee n),
\end{equation}
which implements (\ref{the principle}). Our current recommendation is to set the confidence level $\gamma = 0.1/\log(p\vee n)$ and the constant $c = 1.1$.\footnote{We note that there is not much room to change $c$. Theoretically, we require $c>1$, and finite-sample experiments show that increasing $c$ away from $c=1$ worsens the performance. Hence a value slightly above unity, namely $c=1.1$, is our current recommendation. The simulation evidence suggests that setting $c$ to any value near $1$, including $c=1$, does not impact the result noticeably.}

The Post-Lasso estimator is defined as the ordinary least square regression applied to the model $\widehat I_l \supseteq \widehat T_l$ where $\widehat T_l$ is the model selected by Lasso:
$$\widehat T_{{l}} = \supp( \hat \beta_{{l} L} ) = \{ j \in \{1,\ldots,p\} \ : \ |\hat\beta_{{l} L j }| > 0\},  \ \  {l} =1,..., k_e.$$
The set $\widehat I_l$ can contain additional variables not selected by Lasso, but we require the number of such variables to be similar to or smaller than the number selected by Lasso. The Post-Lasso estimator $\hat\beta_{{l} PL}$  is  \begin{equation}\label{Def:TwoStep} \widehat \beta_{{l} PL} \in \arg\min_{\beta \in \mathbb{R}^p:
\beta_{\widehat I^c_{l}} = 0}   \hat Q_l(\beta), \  {l} =1,..., k_e.
\end{equation}
In words, this estimator is ordinary least squares (OLS) using only the instruments/regressors whose coefficients were estimated to be non-zero by Lasso and any additional variables the researcher feels are important despite having Lasso coefficient estimates of zero.

Lasso and Post-Lasso are motivated by the desire to predict the target function well without overfitting.  Clearly, the OLS estimator is not consistent for estimating the target function when $p > n$. Some approaches based on BIC-penalization of model size are consistent but computationally infeasible.  The Lasso estimator of \citen{T1996} resolves  these difficulties
by penalizing model size through the sum of absolute parameter values. The Lasso estimator is computationally attractive because it minimizes a convex function. Moreover, under suitable conditions, this estimator achieves near-optimal rates in estimating
the regression function  $D_{{l}}(x_i)$.  The estimator achieves these rates by adapting to the unknown smoothness or sparsity of  $D_{l}(x_i)$.  Nonetheless, the estimator has an important drawback: The regularization by the $\ell_1$-norm employed in (\ref{Def:LASSOmain}) naturally lets the Lasso estimator avoid overfitting the data, but it also shrinks the estimated coefficients towards zero causing a potentially significant bias. The Post-Lasso estimator is meant to remove some of this shrinkage bias. If model selection by Lasso works perfectly -- that is,
if it selects exactly the ``relevant" instruments -- then the resulting Post-Lasso estimator
is simply the standard OLS estimator using only the relevant variables.  In cases where perfect selection does not occur, Post-Lasso estimates of coefficients will still tend to be less biased than Lasso. We prove the Post-Lasso estimator achieves the same rate of convergence as Lasso, which is a near-optimal rate, despite imperfect model selection by Lasso.

The introduction of self-normalization via the penalty loadings allows us to contribute to the broad Lasso literature cited in the introduction by showing that under possibly heteroscedastic and non-Gaussian  errors the Lasso and Post-Lasso estimators obey the following near-oracle performance bounds:
\begin{eqnarray}\label{key side A0}
&& \max_{1\leq {l}\leq k_e} \|\widehat D_{i{l}} -D_{i{l}}\|_{2,n} \lesssim_{\mathrm{P}} \sqrt{\frac{s \log (p\vee n)}{n}} \ \ \ \mbox{and} \  \ \max_{1\leq {l}\leq k_e} \| \widehat \beta_{{l}} - \beta_{{l0}}\|_{1}  \lesssim_{\mathrm{P}} \sqrt{\frac{s^2 \log (p\vee n)}{n}}. \ \ \ \ \ \ 
\end{eqnarray}
The performance bounds in (\ref{key side A0}) 
are called near-oracle because they coincide up to a $\sqrt{\log p}$  factor
with the bounds achievable when the the ideal series terms $T_l$ for each of the $k_e$ regressions equations in (\ref{eq: sparse instrument}) are known. Our results
extend those of \citen{BickelRitovTsybakov2009} for Lasso with Gaussian errors and those of \citen{BC-PostLASSO}
for Post-Lasso with Gaussian errors.  Notably, these bounds are as sharp as the results
for the Gaussian case under the weak condition $\log p = o(n^{1/3})$.  They are  also the first results
in the literature that allow for data-driven choice of the penalty level.

It is also useful to contrast the rates given in (\ref{key side A0}) with the rates available for nonparametrically estimating conditional expectations in the series literature; see, for example, \citeasnoun{newey:series}.  Obtaining rates of convergence for series estimators relies on approximate sparsity just as our results do. Approximate sparsity in the series context is typically motivated by smoothness assumptions, but approximate sparsity is more general than typical smoothness assumptions.\footnote{See, e.g., \citeasnoun{BCC-Review} and \citeasnoun{BCH-PLM} for detailed discussion of approximate sparsity.}  The standard series approach postulates that the first $K$ series terms are the most important for approximating the target regression function $D_{il}$. The Lasso approach postulates that $s$ terms from a large number $p$ of terms are important but does not require knowledge of the identity of these terms or the number of terms, $s$, needed to approximate the target function well-enough that approximation errors are small relative to estimation error.  Lasso methods estimate both the optimal number of series terms $s$ as well as the identities of these terms and thus automatically adapt to the unknown sparsity (or smoothness) of the true optimal instrument (conditional expectation). This behavior differs sharply from standard series procedures that do not adapt to the unknown sparsity of the target function unless the number of series terms is chosen by a model selection method.  Lasso-based methods may also provide enhanced approximation of the optimal instrument by allowing selection of the most important terms from a among a set of very many series terms with total number of terms $p \gg K$ that can be much larger than the sample size.\footnote{We can allow for $p \gg n$ for series formed with  orthonormal bases with bounded components, such as trigonometric bases, but further restrictions on the number of terms apply if bounds on components of the series are allowed to increase with the sample size. For example, if we work with B-spline series terms, we can only consider $p = o(n)$ terms.  }  For example, a standard series approach based on $K$ terms will perform poorely when the terms $m+1$, $m+2$,...,$m+j$ are the most important for approximating the optimal instrument for any $K < m$.  On the other hand, lasso-based methods will find the important terms as long as $p > m+j$ which is much less stringent than what is required in usual series approaches since $p$ can be very large. This point can also be made using the array asymptotics where the model changes with $n$ in such a way that the important series terms are always missed by the first $K \rightarrow \infty$ terms.  Of course, the additional flexibility allowed for by Lasso-based methods comes with a price, namely slowing the rate of convergence by $\sqrt{\log p}$ relative to the usual series rates.

\subsection{The Instrumental Variable Estimator based on Lasso and Post-Lasso constructed Optimal Instrument}
Given Condition AS, we take advantage of the approximate sparsity
by using Lasso and Post-Lasso methods to construct estimates of $D_{{l}}(x_i)$ of the form
$$
\widehat D_{{l}}(x_i) = f_i'\widehat \beta_{{l}},  \ \ {l} =1,...,k_e,
$$
and then set
$$
\widehat D_i =  ( \widehat D_{1}(x_i), ..., \widehat D_{k_{e}}(x_i), w_i')'.
$$
The resulting IV estimator takes the form
$$
\widehat \alpha =  \En [\widehat D_id_i']^{-1}  \En [\widehat D_i y_i].
$$

The main result of this paper is to show that, despite the possibility of $p$ being very large, Lasso and
Post-Lasso can select a set of instruments to produce estimates
of the optimal instruments $\widehat D_i$ such that the resulting IV estimator
achieves the efficiency bound asymptotically:
$$
\sqrt{n}(\hat \alpha - \alpha_0) =_d  N(0, \Lambda^*) + o_{\mathrm{P}}(1).
$$
The estimator matches the performance of the classical/standard series-based IV estimator of \citen{newey:optimaliv} and has additional advantages
mentioned in the previous subsection.  We also show that the IV estimator with Lasso-based optimal instruments continues to be root-$n$ consistent and asymptotically normal in the presence of heteroscedasticity:
\begin{equation}\label{robust normality result}
\sqrt{n}(\hat \alpha - \alpha_0) =_d N(0, Q^{-1}\Omega Q^{-1}) + o_{\mathrm{P}}(1),
\end{equation}
where $\Omega :=  \barEp[\epsilon_i^2 D(\x_i) D(\x_i)']$ and $Q:= \barEp[D(\x_i) D(\x_i)']$. A consistent estimator for the asymptotic variance is
\begin{equation}\label{ass variance}
\widehat Q^{-1}\widehat \Omega \widehat Q^{-1}, \ \ \hat \Omega :=\En [\hat \epsilon_i^2 \widehat D(\x_i) \widehat D(\x_i)'], \ \
\hat Q: = \En [\widehat D(\x_i) \widehat D(\x_i)'],
\end{equation}
where $ \hat \epsilon_i := y_i-d_i'\hat\alpha$,  $i=1,\ldots,n$. Using (\ref{ass variance}) we can perform robust inference.

We note that our result (\ref{robust normality result}) for the IV estimator do not rely on the Lasso and Lasso-based procedure specifically.  We provide the properties of the IV estimator for any generic sparsity-based procedure that
achieves the near-oracle performance  bounds (\ref{key side A0}). 

We conclude by stressing that our result (\ref{robust normality result})
does not rely on perfect model selection. Perfect model selection only occurs
in extremely limited circumstances that are unlikely to occur in practice.  We show that model selection mistakes do not affect the  asymptotic distribution of the IV estimator $\widehat \alpha$ under mild regularity conditions. The intuition is that the model selection mistakes are sufficiently small to allow the Lasso or Post-Lasso to estimate the first stage predictions with a sufficient, near-oracle accuracy, which translates to the result above.  Using analysis like that given in \citeasnoun{BCH-PLM}, the result (\ref{robust normality result})  can be shown to hold  over models with strong optimal instruments which are uniformly approximately sparse.  We also offer an inference test procedure in Section 4.2 that remains valid in the absence of a strong optimal instrument, is robust to many weak instruments, and can be used even if  $p \gg n$.  This procedure could also be shown to be uniformly valid over a large class of models.

\section{Results on Lasso and Post-Lasso Estimation of Conditional Expectation
Functions under Heteroscedastic, Non-Gaussian Errors}

In this section, we present our main results on Lasso and Post-Lasso estimators
of conditional expectation functions under non-classical assumptions and
data-driven penalty choices.  The problem we analyze in this section
has many applications outside the IV framework of the present paper.

\subsection{Regularity Conditions for Estimating Conditional Expectations}

The key condition concerns the behavior of the empirical Gram matrix
 $\En[f_if_i']$. This matrix is necessarily singular when $p > n$, so in principle it is not well-behaved.
However, we only need good behavior of certain moduli of continuity of the Gram matrix.    The first modulus
of continuity is called the restricted eigenvalue and is needed for Lasso. The second modulus is called the sparse eigenvalue and is needed for Post-Lasso.

In order to define the restricted eigenvalue, first define the restricted set:
$$\Delta_{C,T} = \{\delta \in \mathbb{R}^p: \|\delta_{T^c}\|_{1} \leq C  \| \delta_{T}\|_{1},  \delta \neq 0\}.$$
The restricted eigenvalue of a Gram matrix $M = \En[f_if_i']$ takes the form:
 \begin{eqnarray}\label{RE}
    \kappa^2_{C}(M) : = \min_{\delta \in \Delta_{C,T}, |T| \leq s} s\frac{\delta' M \delta }{\|\delta_T\|^2_{1} }.
\end{eqnarray}
This restricted eigenvalue can depend on $n$, but we suppress the dependence in our notation.

 In making simplified asymptotic statements involving the Lasso estimator, we will invoke the following condition:

\textbf{Condition RE.} \textit{ For any $C>0$, there exists a finite constant $\kappa> 0$, which does not depend on $n$ but may depend on $C$, such that the restricted eigenvalue obeys
 $\kappa_{C}(\En [f_i f_i']) \geq \kappa$ with probability approaching one  as $n \to \infty$.}

The restricted eigenvalue (\ref{RE}) is a variant of the  restricted eigenvalues introduced in \citen{BickelRitovTsybakov2009} to analyze the properties of Lasso in the classical Gaussian regression model.  Even though the minimal eigenvalue of the empirical Gram matrix $\En[f_if_i']$ is zero whenever $p \geq n$, \citen{BickelRitovTsybakov2009} show that its restricted eigenvalues can  be bounded away from zero. Lemmas \ref{Lemma:GaussianDesign} and \ref{Lemma:BoundedDesign} below contain sufficient conditions for this. Many other
sufficient conditions are available from the literature; see \citen{BickelRitovTsybakov2009}.  Consequently, we take restricted eigenvalues as primitive quantities and Condition RE as a primitive condition.

\begin{remark}[On Restricted Eigenvalues]
In order to gain intuition about restricted eigenvalues, assume the
exactly sparse model, in which there is no approximation error. In this model, the term $\delta$ stands for a generic deviation between an estimator and the true parameter vector $\beta_0$.  Thus, the restricted eigenvalue represents a modulus of continuity between a penalty-related term and the prediction norm, which allows us to derive the rate of convergence.  Indeed, the restricted eigenvalue bounds the minimum change in the prediction norm induced by a deviation $\delta$ within the restricted set $\Delta_{C,T}$ relative to the norm of $\delta_T$, the deviation on the true support.  Given a specific choice of the penalty level, the deviation of the estimator belongs to the restricted set, making the restricted eigenvalue relevant for deriving rates of convergence. 
\end{remark}

In order to define the sparse eigenvalues, let us define the $m$-sparse subset of a unit sphere as
$$
\Delta(m) = \{ \delta \in \mathbb{R}^p: \|\delta\|_0\leq m, \|\delta\|_2 = 1\},
$$
and also define the minimal and maximal $m$-sparse eigenvalue of the Gram matrix $M= \En\[f_if_i'\]$ as
$$
\semin{m}(M) = \min_{\delta \in \Delta(m)}  \delta'M\delta \ \ \mbox{and} \ \
\semax{m}(M) = \max_{\delta \in \Delta(m)}  \delta'M\delta.
$$

To simplify asymptotic statements for Post-Lasso, we use the following condition:

\textbf{Condition SE.} \textit{ For any $C>0$, there exist constants
$0< \kappa' <  \kappa'' < \infty$, which do not depend on $n$ but may depend on $C$, such
that with probability approaching one, as $n \to \infty$,
$\kappa' \leq \semin{Cs}(\En[f_if_i']) \leq \semax{Cs}(\En[f_if_i']) \leq \kappa''$.
}

Condition SE requires only that certain ``small" $m \times m$ submatrices of the large $p \times p$
empirical Gram matrix are well-behaved, which is a reasonable assumption and will be sufficient for the results that follow.
Condition SE implies Condition RE by the argument given in \citen{BickelRitovTsybakov2009}. The following lemmas show that Conditions RE and SE are plausible for both many-instrument and many series-instrument settings. We refer to \citen{BC-PostLASSO} for proofs; the first lemma builds upon results in \citen{ZhangHuang2006} and the second builds upon results in \citen{RudelsonVershynin2008}. The lemmas could also be derived from \citen{RudelsonZhou2011}.

\begin{lemma}[Plausibility of RE and SE under Many Gaussian Instruments]\label{Lemma:GaussianDesign}   Suppose $f_i$, $i = 1,\ldots,n$, are i.i.d. zero-mean Gaussian random vectors.  Further suppose that the population Gram matrix $\Ep[f_i f_i']$ has  $s\log n$-sparse eigenvalues  bounded from above and away from zero uniformly in $n$. Then if $s\log n = o(n/\log p)$, Conditions RE and SE hold.
\end{lemma}

\begin{lemma}[Plausibility of RE and SE under Many Series Instruments]\label{Lemma:BoundedDesign}  Suppose $f_i$, $i=1,\ldots,n$, are i.i.d. bounded zero-mean random vectors with $\| f_i\|_\infty \leq K_B$ a.s.  Further suppose that the population Gram  matrix $\Ep[f_i f_i']$ has $s\log n$-sparse eigenvalues bounded from above and away from zero uniformly in $n$. Then if $K^2_Bs\log^2 (n)  \log^2 (s\log n)  \log (p\vee n) =  o(n)$, Conditions RE and SE hold.
\end{lemma}

In the context of i.i.d. sampling, a standard assumption in econometric research is that the population Gram matrix $\Ep[f_i f_i']$ has eigenvalues bounded from above and below, see e.g. \citen{newey:series}.
The lemmas above allow for this and more general behavior, requiring only that the sparse eigenvalues of the population Gram matrix $\Ep[f_i f_i']$ are bounded from below and from above. The latter is important for allowing functions $f_i$ to be formed as a combination of elements from different bases, e.g. a combination of B-splines with polynomials.  The lemmas above further show that the good behavior of the population sparse eigenvalues translates into good behavior of empirical sparse eigenvalues under some restrictions on the growth of $s$ in relation to the sample size $n$.  For example, if $p$ grows polynomially with $n$ and the components of technical regressors are uniformly bounded, Lemma \ref{Lemma:BoundedDesign} holds provided $s = o(n/\log^5 n)$.

We also impose the following moment conditions on the reduced form errors $v_{il}$ and regressors $f_i$,
where we let $\tilde d_{il} := d_{il} - \barEp[d_{il}]$.

\textbf{Condition RF.}  \textit{
(i) $   \max_{l \leq k_e, j \leq p} \ \barEp[\tilde d^2_{il}] +\barEp[ |f^2_{ij} \tilde d^2_{il}|] + 1/\barEp[f_{ij}^2 v_{il}^2]  \lesssim 1$, (ii) $   {\displaystyle \max_{l \leq k_e, j \leq p} } \barEp[|f_{ij}^3 v_{il}^3|] \lesssim K_n$, (iii) $K_n^{2}\log^3 (p\vee n) = o(n) $ and $\ s \log (p\vee n) = o(n)
$, (iv)  $  \max_{i \leq n, j \leq p}  f_{ij}^2 [s \log (p\vee n)]/n \to_{\Pr} 0$ and $ \max_{l \leq k_e, j \leq p}   |(\En - \barEp)[f_{ij}^2 v^2_{il}] |  + | (\En - \barEp)[f^2_{ij}\tilde d^2_{il}]| \to_{\Pr} 0.$}

We emphasize that the conditions given above are only one possible set of
sufficient conditions, which are presented in a manner that reduces
the complexity of the exposition.

The following lemma shows that the population and empirical moment conditions appearing
 in Condition RF are plausible for both many-instrument and many series-instrument settings. Note that we say that a  random variable $g_{i}$ has uniformly bounded conditional moments of order $K$
 if for some positive constants  $ 0< B_1 < B_2< \infty$:
$$ B_1 \leq \Ep\Big[ |g_{i}|^k \Big |\x_i \Big ] \leq B_2 \text{ with probability } 1, \ \mbox{for} \ k = 1,\ldots, K, \ i=1, \ldots, n.$$


\begin{lemma}[Plausibility of RF]\label{lemma: RF}
 1. If the moments $\barEp[\tilde d_{il}^8]$ and $\barEp[v_{il}^8]$ are bounded uniformly in $1 \leq l \leq k_e$ and in $n$, the regressors $f_i$ obey
$\max_{1 \leq j \leq p}\En[f^8_{ij}] \lesssim_{\mathrm{P}} 1$  and $\max_{1 \leq i \leq n, 1 \leq j \leq p}f_{ij}^2 \frac{s \log (p \vee n)}{n}$ $\to_{\mathrm{P}}$ $0
$, Conditions RF(i)-(iii) imply Condition RF (iv).  2. Suppose that $\{(f_i, \tilde d_i, v_i),i=1,...n\}$ are i.i.d. vectors, and that $\tilde d_{il}$ and $v_{il}$  have uniformly bounded conditional moments of order 4 uniformly in $l=1,\ldots,k_e$. (1) If the regressors $f_i$ are Gaussian as in Lemma \ref{Lemma:GaussianDesign}, Condition RF(iii) holds, and $ s \log^2 (p\vee n)/n \to 0$ then Conditions RF(i),(ii) and (iv) hold.  (2) If the regressors $f_i$ have bounded entries as in Lemma \ref{Lemma:BoundedDesign}, then Conditions RF(i),(ii) and (iv) hold under Condition RF(iii).
\end{lemma}

\subsection{Results on  Lasso and Post-Lasso for Estimating Conditional Expectations}

We consider Lasso and Post-Lasso estimators defined in equations (\ref{Def:LASSOmain}) and (\ref{Def:TwoStep})  in the system of $k_e$ nonparametric regression equations (\ref{eq: k reg equations})
with  non-Gaussian and  heteroscedastic errors.    These results extend the previous results
of \citen{BickelRitovTsybakov2009} for Lasso and of \citen{BC-PostLASSO} for Post-Lasso with classical i.i.d. errors.
In addition, we account for the fact that we are simultaneously estimating $k_e$ regressions and account for the dependence of our results on $k_e$.

The following theorem presents the properties of Lasso.  Let us call asymptotically valid
 any penalty loadings $\hat \Upsilon_l $ that obey a.s.
\begin{equation}\label{Def:AsympValidPenaltyLoading}
\ell \hat \Upsilon^0_l \leq \hat \Upsilon_l \leq u \hat \Upsilon^0_l,
\end{equation}
with $ 0< \ell \leq 1 \leq u$ such that $\ell \to_{\mathrm{P}} 1$ and $u  \to_{\mathrm{P}} u'$ with $u' \geq 1$.  The penalty loadings constructed by Algorithm A.1 satisfy this condition.

\begin{theorem}[Rates for Lasso under Non-Gaussian and Heteroscedastic Errors]\label{Thm:RateLASSO} Suppose that in the regression model (\ref{eq: k reg equations}) Conditions AS and RF hold. Suppose
the penalty level is specified as in (\ref{the principle lambda}), and  consider any asymptotically valid penalty loadings $\hat \Upsilon $, for example, penalty loadings constructed by Algorithm A.1 stated in Appendix A. Then, the Lasso estimator $\widehat \beta_{l}= \widehat \beta_{lL}$
and the Lasso fit $\hat D_{il} = f_i' \hat \beta_{lL}$, $l= 1,...,k_e$, satisfy
$$ \max_{1 \leq l \leq k_e }\|\hat D_{il}  - D_{il}\|_{2,n}  \lesssim_{\mathrm{P}}  \frac{1}{\kappa_{\bar C}}\sqrt{\frac{s \log (k_ep/\gamma)}{ n }} \ \ \mbox{and} \ \    \max_{1 \leq l \leq k_e } \ \ \|\hat{\beta}_{l} - {\beta_{l}}_0 \|_1   \lesssim_{\mathrm{P}} \frac{1}{(\kappa_{2\bar C})^2}\sqrt{\frac{s^2 \log (k_ep/\gamma)}{ n }},
$$
where $\bar C = {\displaystyle \max_{1\leq l\leq k_e}}\{\|\hat \Upsilon_l^0\|_\infty\|(\hat\Upsilon_l^{0})^{-1}\|_\infty\}(uc+1)/(\ell c - 1)$ and $\kappa_{\bar C}=\kappa_{\bar C}(\En[f_if_i'])$. \end{theorem}

The theorem provides a rate result for the Lasso estimator constructed specifically to deal with non-Gaussian errors and heteroscedasticity.
The rate result generalizes, and is as sharp as, the rate results of \citen{BickelRitovTsybakov2009} obtained for the homoscedastic Gaussian case.
 This generalization is important for real applications where non-Gaussianity and heteroscedasticity are
 ubiquitous.  Note that the obtained rate is near-optimal in the sense
 that if we happened to know the model $T_\ell$, i.e. if we knew the identities of the most important variables, we would only improve the rate by the $\log p$ factor.  The theorem also shows that the data-driven penalty loadings defined in Algorithm A.1 are asymptotically valid.

The following theorem presents the properties of Post-Lasso which requires a mild assumption on the number of additional variables in the set $\widehat I_l$, $l=1,\ldots,k_e$. We assume that the size of these sets are not substantially larger than the model selected by Lasso, namely, a.s.
\begin{equation}\label{Def:AmeliorationSets}
|\widehat I_l \setminus  \widehat T_l| \lesssim 1 \vee |\widehat T_l|, \ \ l=1,\ldots,k_e.
\end{equation}

\begin{theorem}[Rates for Post-Lasso under Non-Gaussian and Heteroscedastic Errors]\label{Thm:RatesPostLASSO}
  Suppose that in the regression model (\ref{eq: k reg equations}) Conditions AS and RF hold. Suppose
the penalty level for the Lasso estimator is specified as in (\ref{the principle lambda}), that Lasso's
penalty loadings $\hat \Upsilon $ are asymptotically valid, and the sets of additional variables obey (\ref{Def:AmeliorationSets}).  Then, the Post-Lasso estimator $\widehat \beta_{l}= \widehat \beta_{lPL}$
and the Post-Lasso fit $\hat D_{il} = f_i' \hat \beta_{lPL}$, $l= 1,...,k_e$, satisfy
$$ \max_{1 \leq l \leq k_e }\|\hat D_{il}  - D_{il}\|_{2,n} \lesssim_{\mathrm{P}} \frac{\mu}{\kappa_{\bar C}} \sqrt{\frac{s \log (k_ep/\gamma)}{ n }} \ \ \mbox{and} \ \
   \max_{1 \leq l \leq k_e } \ \ \|\hat{\beta}_{l} - {\beta_{l}}_0 \|_1   \lesssim_{\mathrm{P}} \frac{\mu^2}{(\kappa_{\bar C})^2}\sqrt{\frac{s^2 \log (k_ep/\gamma)}{ n }},
$$
where $\mu^2 = \min_{k} \{\semax{k}(\En[f_if_i'])/\semin{k+s}(\En[f_if_i']): k > 18 \bar C^2s\semax{k}(\En[f_if_i'])/(\kappa_{\bar C})^2\}$ for $\bar{C}$ defined in Theorem \ref{Thm:RateLASSO}. \end{theorem}

The theorem provides a rate result for the Post-Lasso estimator with non-Gaussian errors and heteroscedasticity. The rate result generalizes the results of \citen{BC-PostLASSO} obtained for the homoscedastic Gaussian case.
The Post-Lasso achieves the same near-optimal rate of convergence as Lasso.  As stressed in the introductory sections, our analysis  allows Lasso to make model selection mistakes which is expected generically. We show that these model selection mistakes are small enough to allow the Post-Lasso estimator to perform as well as Lasso.\footnote{Under further conditions stated in proofs, Post-Lasso can sometimes achieve a faster rate of convergence. In special cases where perfect model selection is possible, Post-Lasso becomes the so-called oracle estimator and can completely remove the $\log p$ factor.}

Rates of convergence in different norms can also be of interest in other applications. In particular, the $\ell_2$-rate of convergence can be derived
from the  rate of convergence in the prediction norm and Condition SE using a sparsity result for Lasso established in Appendix \ref{AppendixThm2}. Below we specialize the previous theorems to the important case that Condition SE holds. 

\begin{corollary}[Rates for Lasso and Post-Lasso under SE]\label{CorollaryL2}
Under the conditions of Theorem \ref{Thm:RatesPostLASSO} and Condition SE, the Lasso and Post-Lasso estimators satisfy
{\small $$ \max_{1 \leq l \leq k_e }\|\hat D_{il}  - D_{il}\|_{2,n}   \lesssim_{\mathrm{P}}  \sqrt{\frac{s \log (p\vee n)}{ n }}, $$ $$\ \
  \max_{1 \leq l \leq k_e } \|\hat{\beta}_{l} - {\beta_{l}}_0 \|_2   \lesssim_{\mathrm{P}} \sqrt{\frac{s \log (p\vee n)}{ n }}, \ \
  \max_{1 \leq l \leq k_e }  \|\hat{\beta}_{l} - {\beta_{l}}_0 \|_1   \lesssim_{\mathrm{P}} \sqrt{\frac{s^2 \log (p\vee n)}{ n }}.
$$}
\end{corollary}
The rates of convergence in the prediction norm and $\ell_2$-norm are faster than the rate of convergence in the $\ell_1$-norm which is typical of  high dimensional settings.

\section{Main Results on IV Estimation}  In this section we present our main inferential results
on instrumental variable estimators.

\subsection{The IV estimator with Lasso-based  instruments}
We impose the following moment conditions on the instruments, the structural errors,
and regressors.

\textbf{Condition SM.}  \textit{ (i) The eigenvalues of $Q=\barEp[D(\x_i)D(\x_i)']$ are bounded uniformly from above and away from zero, uniformly in $n$. The conditional variance
$\Ep[\epsilon_i^2|\x_i]$ is bounded uniformly from above and away from zero, uniformly in $i$ and $n$. Given this assumption,
without loss of generality, we normalize the instruments so that
 $\barEp[f_{ij}^2\epsilon_i^2] =1$ for each $1 \leq j \leq p$  and for all $n$.
(ii)  For some $q>2$ and $q_{\epsilon} > 2$, uniformly in $n$,
$$
\max_{1\leq j \leq p}\barEp[ |f_{ij}\epsilon_i|^3] + \barEp[\|D_i\|^q_2 |\epsilon_i|^{2q}] + \barEp[ \| D_i\|^q_2] + \barEp[|\epsilon_i|^{q_{\epsilon}}] + \barEp[ \|d_i\|^q_2]  \lesssim 1.
$$
(iii) In addition to $\log^3 p = o(n)$, the following growth conditions hold:
$$(a) \ \frac{s \log (p\vee n) }{n} n^{2/q_{\epsilon}} \to 0 \ \ (b) \ \frac{s^2 \log^2 (p\vee n)}{n} \to 0, \ \ (c) \max_{1\leq j \leq p} \En[f^2_{ij} \epsilon^2_i] \lesssim_P 1.
$$
}

\begin{remark}(On Condition SM)
Condition SM(i) places restrictions on the variation of the structural errors $(\epsilon)$ and the optimal instruments $(D(x))$.  The first condition about the variation in the optimal instrument guarantees that identification is strong; that is, it ensures that the conditional expectation of the endogenous variables given the instruments is a nontrivial function of the instruments.  This assumption rules out non-identification in which case $D(x)$ does not depend on $x$ and weak-identification in which case $D(x)$ would be local to a constant function.  We present an inference procedure that remains valid without this condition in Section 4.2.
The remaining restriction in Condition SM(i) requires that structural errors are boundedly heteroscedastic.  Given this
we make a normalization assumption on the instruments. This entails no loss of generality
since this is equivalent to suitably rescaling the parameter space for coefficients
$\beta_{l0}, \ l =1,..., k_e$, via an isomorphic transformation.
We use this normalization to simplify notation in the proofs but do not use it in the construction of the estimators.  Condition SM(ii) imposes  some mild moment assumptions. Condition SM(iii) strengthens the growth requirement
 $s \log p/n \to 0$ needed for estimating conditional expectations.
However, the restrictiveness of  Condition SM(iii)(a) rapidly decreases
 as the number of bounded moments of the structural error increases.
 Condition SM(iii)(b) indirectly requires the optimal
instruments in Condition AS to be smooth enough that the number of unknown series terms $s$ needed
to approximate them well is not too large.  This condition ensures that the impact of the instrument
estimation on the IV estimator is asymptotically negligible.  This condition can be relaxed using
the sample-splitting method. \end{remark}

The following lemma shows that the moment assumptions in Condition SM (iii) are plausible for both many-instrument and many series-instrument settings.

\begin{lemma}[Plausibility of SM(iii)]\label{lemma: SM}
Suppose that the structural disturbance $\epsilon_i$ has
uniformly bounded conditional moments of order 4 uniformly in $n$ and that $s^2\log^2(p\vee n) = o(n)$. Then Condition SM(iii) holds if  (1) the regressors $f_i$ are Gaussian
as in Lemma 1 or (2) the regressors $f_i$ are arbitrary i.i.d. vectors with bounded entries as in Lemma 2.
\end{lemma}

The first result describes the properties of the IV estimator
with the optimal IV constructed using Lasso or Post-Lasso in the
setting of the standard model. The result
also provides a consistent estimator for the asymptotic variance of this estimator
under heteroscedasticity.

\begin{theorem}[Inference with Optimal IV Estimated by Lasso or Post-Lasso]\label{theorem 1}
Suppose that data $(y_i, \x_i, d_i)$ are i.n.i.d. and obey the linear IV model described in Section 2.
Suppose also that Conditions AS, RF, SM, (\ref{the principle lambda}) and (\ref{Def:AsympValidPenaltyLoading}) hold. To construct the estimate of the optimal instrument, suppose that Condition RE holds in the case of using Lasso or that Condition SE and (\ref{Def:AmeliorationSets}) hold in the case of using Post-Lasso.
Then the IV estimator $\hat \alpha$, based on either Lasso
or Post-Lasso estimates of the optimal instrument, is root-$n$ consistent and asymptotically normal:
$$
 (Q^{-1}\Omega Q^{-1})^{-1/2}\sqrt{n}(\hat \alpha - \alpha_0) \to_d N(0,I),
$$
for $\Omega :=  \barEp[\epsilon_i^2 D(\x_i) D(\x_i)']$ and $Q:= \barEp[D(\x_i) D(\x_i)']$.
Moreover, the result above continues to hold with
  $\Omega$ replaced by $\hat \Omega := \En [\hat \epsilon_i^2 \widehat D(\x_i) \widehat D(\x_i)'] $  for  $\hat \epsilon_i =
y_i - d_i'\hat \alpha$, and $Q$ replaced by $\hat Q: = \En [\widehat D(\x_i) \widehat D(\x_i)']$. In the case that the structural error $\epsilon_i$ is homoscedastic conditional on $x_i$, that is, $E[\epsilon_i^2|\x_i] = \sigma^2$ a.s. for all $i=1,...,n$, the IV estimator $\widehat \alpha$ based on either Lasso or Post-Lasso estimates of the optimal instrument is root-$n$ consistent, asymptotically normal, and achieves the efficiency bound:
$
(\Lambda^*)^{-1/2}\sqrt{n}(\hat \alpha - \alpha_0) \to_d N(0, I)
$
where $\Lambda^* := \sigma^2 Q^{-1}$. The result above continues to hold  with $\Lambda^*$ replaced by  $\hat \Lambda^* :=\hat \sigma^2 \hat Q^{-1}$, where $\hat Q:= \En [\widehat D(\x_i) \widehat D(\x_i)']$ and $\hat\sigma^2 := \En[(y_i-d_i'\widehat\alpha)^2]$.
\end{theorem}

In the setting with homoscedastic structural errors
the estimator achieves the efficiency bound asymptotically.
In the case of heteroscedastic structural errors, the estimator does not achieve the efficiency bound,
but we can expect it to be close to achieving the bound if heteroscedasticity
is mild.

%

The final result of this section extends the previous result to any IV-estimator
with a generic sparse estimator of the optimal instruments.

\begin{theorem}[Inference with IV Constructed by a Generic Sparsity-Based Procedure]\label{theorem 3}
Suppose that conditions AS and SM hold, and suppose that the fitted values of the optimal instrument, $\widehat D_{il} = f_i'\widehat \beta_l$, are constructed
using any estimator $\widehat \beta_l$ such that
\begin{eqnarray}\label{key side}
&& \max_{1\leq {l}\leq k_e} \|\widehat D_{i{l}} -D_{i{l}}\|_{2,n} \lesssim_{\mathrm{P}} \sqrt{\frac{s \log (p\vee n)}{n}} \ \ \mbox{and} \ \ \max_{1\leq {l}\leq k_e} \| \widehat \beta_{{l}} - \beta_{{l0}}\|_{1}  \lesssim_{\mathrm{P}} \sqrt{\frac{s^2 \log (p\vee n)}{n}}. \ \ \ \ \ \ \ \ \ \ 
\end{eqnarray}
Then the conclusions reached in Theorem \ref{theorem 1} continue to apply. \end{theorem}

 This result shows that the previous two theorems apply for any first-stage estimator that attains near-oracle performance given in
(\ref{key side}). 
Examples of other sparse estimators covered by this theorem are Dantzig and Gauss-Dantzig (\citename{CandesTao2007}, \citeyear*{CandesTao2007}), $\LASSO$ and post-$\LASSO$ (\citename{BCW-SqLASSO}, \citeyear*{BCW-SqLASSO2} and \citeyear*{BCW-SqLASSO}), thresholded Lasso and Post-thresholded Lasso (\citename{BC-PostLASSO}, \citeyear*{BC-PostLASSO}), group Lasso and Post-group Lasso (\citename{horowitz:lasso}, \citeyear*{horowitz:lasso}; and \citename{LouniciPontilTsybakovvandeGeer2009}, \citeyear*{LouniciPontilTsybakovvandeGeer2009}), adaptive versions of the above (\citename{horowitz:lasso}, \citeyear*{horowitz:lasso}), and boosting (\citename{BuhlmannBoosting2006}, \citeyear*{BuhlmannBoosting2006}).  Verification
of the near-oracle performance (\ref{key side}) 
can be done on a case by case basis
using the best conditions in the literature.\footnote{Post-$\ell_1$-penalized
procedures have only been analyzed for the case of Lasso and $\LASSO$; see \citen{BC-PostLASSO} and \citen{BCW-SqLASSO2}. We expect
that similar results carry over to other procedures listed above.}
Our results extend to Lasso-type estimators under alternative forms of regularity conditions
that fall outside the framework of Conditions RE and Conditions RF; all that is required is the near-oracle performance
of the kind (\ref{key side}). 

\subsection{Inference when instruments are weak}   When instruments are weak individually, Lasso may end up selecting no instruments or may produce unreliable estimates of the optimal instruments. To cover this case, we propose a method for inference based on inverting pointwise tests performed using a sup-score statistic defined below.  The procedure is similar in spirit to \citen{anderson:rubin} and \citen{ss:weakiv} but uses a different statistics that is better suited to cases with very many instruments. In order to describe the approach, we rewrite the main structural equation as:
\begin{eqnarray}
& &  y_{i}  = d_{ei} '\alpha_{1} + w_i'\alpha_2 + \epsilon_i, \ \ \Ep[\epsilon_i|x_i] = 0, \label{Def: WIVmodel}
\end{eqnarray}
where $y_{i}$ is the response variable, $d_{ei}$ is a vector of endogenous variables, $w_i$ is a $k_w$-vector of control variables, $x_i = (z_i',w_i')'$ is a vector of elementary instrumental variables, and $\epsilon_i$ is a disturbance such that $\epsilon_1,..., \epsilon_n$ are i.n.i.d. conditional on $X=[x_1',...,x_n']$.  We partition $d_i = ({d_{ei}}', w_i')'$. The parameter of interest is $\alpha_1 \in \mathcal{A}_1 \subset \mathbb{R}^{k_e}$.  We use $f_{i} =  P(x_i)$, a vector which includes $w_i$, as technical instruments.     In this subsection, we treat $X$ as fixed; i.e. we condition on $X$.

 We  would like to  use a high-dimensional vector $f_i$ of technical instruments for inference on $\alpha_1$ that is robust to weak identification. In order to formulate a practical sup-score statistic, it is useful to partial-out the effect of $w_i$ on the key variables.
For an $n$-vector $\{u_i, i=1,...,n\}$, define $\tilde u_{i} = u_{i} - w_{i}'\En[w_i w_i']^{-1} \En[w_i u_i]$, i.e. the residuals left after regressing this vector on $\{w_i, i=1,...,n\}$. Hence $\tilde y_i$, $\tilde d_{ei}$, and $\tilde f_{ij}$ are residuals
obtained by partialling out controls $w_i$.  Also, let
\begin{equation}\label{def: define tilde f}
\tilde f_i = (\tilde f_{i1},...,\tilde f_{ip})'.
\end{equation} In this formulation, we
omit elements of $w_i$ from $\tilde f_{ij}$ since they are eliminated by partialling out. We then normalize these technical instruments so that
\begin{equation}\label{def: normalize tilde f}
\En[\tilde f_{ij}^2] =1, \ \ j =1,...,p.
\end{equation}

  The sup-score statistic for testing
the hypothesis $\alpha_1 = a $ takes the form
\begin{equation}\label{def:sup-score}
\Lambda_a =  \max_{1 \leq j \leq p} |n \En [(\tilde y_{i} - \tilde d_{ei}'a) \tilde f_{ij}]|/\sqrt{\En[
(\tilde y_{i} - \tilde d_{ei}'a)^2 \tilde f^2_{ij} ]}.
\end{equation}
As before, we apply self-normalized moderate deviation theory for self-normalized sums to obtain
$$
\Pr ( \Lambda_{\alpha_1} \leq  c \sqrt{n}\Phi^{-1}(1- \gamma/2p) ) \geq 1- \gamma + o(1).
$$
Therefore, we can employ  $\Lambda(1- \gamma):= c \sqrt{n}\Phi^{-1}(1- \gamma/2p)$
 for $c > 1$ as a critical value for testing $\alpha_1 = a$ using $\Lambda_a$ as the test-statistic. The asymptotic $(1- \gamma)$ -- confidence region for $\alpha_1$ is then given by
$
\mathcal{C}:=\{ a \in \mathcal{A}_1:  \Lambda_a \leq \Lambda(1- \gamma)\}.
$

The construction of confidence regions above can be given the following \emph{Inverse Lasso} interpretation.
Let
$$
\hat \beta_a  \in \arg\min_{\beta \in \mathbb{R}^p} \En[ (\tilde y_{i} -  \tilde d_{ei}'a) - \tilde f_{ij}'\beta]^2 + \frac{\lambda}{n}
\sum_{j=1}^p \gamma_{aj} | \beta_j |, \ \ \gamma_{aj} = \sqrt{\En[ (\tilde y_{i} - \tilde d_{ei}'a)^2 \tilde f^2_{ij} ]}.
$$
If $\lambda = 2\Lambda(1- \gamma)$, then $ \mathcal{C}$ is equivalent to the region $\{ a \in \mathbb{R}^{k_e}:  \hat \beta_a = 0\}$. In words, this confidence region collects all potential values of the structural parameter where the Lasso regression of the potential structural disturbance on the instruments yields zero coefficients on the instruments.  This idea is akin to the Inverse Quantile Regression and Inverse Least Squares ideas in \citename{ch:iqrWeakId} \citeyear{ch:iqrWeakId,ch:WeakId}.

Below, we state the main regularity condition for the validity of inference using the sup-score statistic as well as the formal inference result.

\textbf{Condition SM2.} \emph{Suppose that for each $n$  the linear model (\ref{Def: WIVmodel}) holds with $\alpha_1 \in \mathcal{A}_1 \subset \mathbb{R}^{k_e}$ such that $\epsilon_1,..., \epsilon_n$ are i.n.i.d., $X$ is fixed, and $\tilde f_1,... ,\tilde f_n$ are $p$-vectors of technical instruments defined in (\ref{def: define tilde f}) and (\ref{def: normalize tilde f}). Suppose that
(i) the dimension of $w_i$ is $k_w$ and $\|w_i\|_2 \leq \zeta_w$ such that  $\sqrt{k_w} \zeta_w/\sqrt{n} \to 0$, (ii)  the eigenvalues of $\En[w_i w_i']$ are bounded away from zero and eigenvalues of $\barEp[\epsilon_i^2w_i w_i']$ are bounded away from above, uniformly in $n$, (iii) $\max_{1 \leq j \leq p }\barEp[|\epsilon_i|^3 |\tilde f_{ij}|^3] ^{1/3}/ \barEp[\epsilon_i^2 \tilde f^2_{ij}] ^{1/2} \leq K_n$ , and (iv)  $K^2_n \log (p\vee n) = o(n^{1/3})$. }

\begin{theorem}[Valid Inference based on the Sup-Score Statistic]\label{Thm:WIV} Let $\gamma \in (0,1)$ be fixed or, more generally,   such that $\log (1/\gamma) \lesssim \log (p\vee n)$.
Under Condition SM2,
(1) in large samples, the constructed confidence set $\mathcal{C}$ contains the true value $\alpha_1$ with at least the prescribed probability,  namely $\Pr( \alpha_1 \in \mathcal{C}) \geq 1- \gamma -o(1).$
  (2) Moreover,  the confidence set $\mathcal{C}$ necessarily excludes a sequence of parameter value $a$, namely $\Pr ( a \in \mathcal{C}) \to 0$, if
  $$
\max_{1 \leq j \leq p}  \frac{\sqrt{n/\log (p/\gamma)} \ |\En[(a-\alpha_1)' \tilde d_{ei}\tilde f_{ij}]|}{ c\sqrt{\En[\epsilon^2_i \tilde f_{ij}^2]} + \sqrt{
\En[ \{(a-\alpha_1)' \tilde d_{ei}\}^2 \tilde f^2_{ij}] }} \to_{\Pr} \infty. $$
\end{theorem}

The theorem shows that the confidence region $\mathcal{C}$  constructed above is valid in large samples and that the probability of including a false point $a$ in  $\mathcal{C}$ tends to zero as long as $a$ is sufficiently distant from $\alpha_1$ and instruments
are not too weak.  In particular, if there is a strong instrument, the confidence regions will eventually
exclude points $a$ that are further than $\sqrt{\log (p\vee n)/n}$ away from $\alpha_1$. Moreover,
if there are instruments whose correlation with the endogenous variable is of greater order
than $\sqrt{\log (p\vee n)/n}$, then the confidence regions will asymptotically be bounded.

\section{Further Inference and Estimation Results for the IV Model}
In this section we provide further estimation and inference results. We develop
an overidentification test which compares the IV-Lasso based estimates to estimates obtained
using a baseline set of instruments.  We also combine the IV selection using Lasso with a sample-splitting
technique from the many instruments literature which allows us to relax the growth requirement on the number
of relevant instruments.

\subsection{ A Specification Test for Validity of Instrumental Variables}
Here we develop a Hausman-style specification test for the validity of the instrumental variables.
 Let $A_i =A(x_i)$ be a baseline
set of instruments, with $\dim(A_i) \geq \dim(\alpha)=k_{\alpha}$ bounded. Let  $\tilde \alpha$ be the baseline instrumental variable estimator based on these instruments:
$$
\tilde \alpha = (\En[d_i A_i']\En[A_i A_i']^{-1}\En[A_i d_i'])^{-1}\En[d_i A_i'] \En[A_i A_i']^{-1} \En[A_i y_i].
$$
If the instrumental variable exclusion restriction is valid, then the unscaled difference between this estimator and the IV estimator $\widehat \alpha$ proposed in the previous sections should be small. If the exclusion restriction is not valid, the difference between $\tilde \alpha$ and  $\widehat \alpha$ should be large.
Therefore, we can reject the null hypothesis of instrument validity if the difference is large.

We formalize the test as follows.  Suppose we care about
$R'\alpha$ for some $k\times k_d$ matrix $R$ of $\text{rank}(R) = k$. For instance,
we might care only about the first $k$ components of $\alpha$, in which case $R=[I_k \ 0]$ is a
 $k \times k_d$ matrix that selects the first $k$ coefficients of $\alpha$.
Define the estimand for $\tilde \alpha$ as
$$
\alpha =  (\barEp[d_i A_i']\barEp[A_i A_i']^{-1}\barEp[A_i d_i'])^{-1}\barEp[d_i A_i'] \barEp[A_i A_i']^{-1} \barEp[A_i y_i],
$$
and define the estimand of $\hat \alpha$ as
$$\alpha_a = \barEp[ D(\x_i) D(\x_i)']^{-1} \barEp[ D(\x_i) y_i].$$
The null hypothesis
$H_0$ is $R(\alpha - \alpha_a)= 0$ and the alternative $H_a$ is $R(\alpha - \alpha_a) \neq 0$.
We can form a test statistic
$$
J= \sqrt{n} (\tilde \alpha - \hat \alpha)'R' (R\widehat \Sigma R')^{-1} \sqrt{n} R(\tilde \alpha - \hat \alpha)
$$
 for a matrix $\widehat \Sigma$
defined below and reject $H_0$ if  $J >  c_\gamma$
where $c_\gamma$ is the $(1-\gamma)$-quantile of chi-square random variable with $k$ degrees of freedom.  The justification for this test is provided by the following theorem which builds upon the previous results coupled with conventional results
for the baseline instrumental variable estimator.\footnote{The proof of this result is provided in a supplementary appendix.}

\begin{theorem}[Specification Test]\label{theorem: spec}  (1) Suppose the conditions of Theorem \ref{theorem 1} hold,
that $\barEp[\|A_i\|_2^{q}]$ is bounded uniformly in $n$ for $q >4$, and the
eigenvalues of $$\Sigma:= \barEp[\epsilon^2_i( M A_i-Q^{-1} D(x_i))( M A_i-Q^{-1} D(x_i))' ]$$ are bounded
from above and below, uniformly in $n$, where
$$M = (\barEp[d_i A_i']\barEp[A_i A_i']^{-1}\barEp[A_i d_i'])^{-1}\barEp[d_i A_i'] \barEp[A_i A_i']^{-1}.$$
 Then
$
\sqrt{n} \widehat \Sigma^{-1/2} (\tilde \alpha - \hat \alpha)' \to_d  N(0,I) \text { and } J \to_d \chi^2(k),$
where $$\widehat \Sigma= \En[\hat \epsilon^2_i( \hat M^{-1} A_i-\hat Q^{-1} \hat D(x_i))( \hat M^{-1} A_i-\hat Q^{-1} \hat D(x_i))' ],$$ for  $\hat \epsilon_i =
y_i - d_i'\hat \alpha$, $ \widehat Q = \En [\widehat D(\x_i) \widehat D(\x_i)']$, and $$\hat M = (\En[d_i A_i']\En[A_i A_i']^{-1}\En[A_i d_i'])^{-1}\En[d_i  A_i'] \En[A_i A_i']^{-1}.$$
(2)  Suppose the conditions of Theorem \ref{theorem 1} hold with the exception that $\Ep[A_i \epsilon_i]=0$ for all $i=1,...,n$ and $n$, but $\|\barEp[D(x_i) \epsilon_i]\|_2$ is bounded away from zero. Then
$J \to_{\Pr} \infty$.
\end{theorem}

\subsection{Split-sample IV estimator}

The rate condition $s^2\log^2 (p\vee n) = o(n)$ can be substantive and cannot be substantially weakened for the full-sample IV estimator considered above.  However, we can replace this condition with the weaker condition that
$$ s \log (p\vee n) = o(n)$$
by employing a sample splitting method from the many instruments literature (\citename{AngristKruegerSplitSample1995}, \citeyear*{AngristKruegerSplitSample1995}).  Specifically, we consider dividing the sample randomly into (approximately) equal parts $a$ and $b$, with sizes $n_a = \lceil n/2 \rceil$ and  $n_b = n - n_a$. We use superscripts $a$ and $b$ for variables in the first and second subsample respectively.  The index $i$ will enumerate observations in both samples, with ranges for the index given by $1 \leq i \leq n_a$ for sample $a$ and $1 \leq i \leq n_b$ for sample $b$.  We can use each of the subsamples to fit the first stage via Lasso or Post-Lasso to obtain the first stage estimates $\hat \beta^k_l, k=a, b,$ and $l=1,\ldots,k_e$. Then setting $\hat D_{il}^a = {f_i^a}'\widehat \beta^b_l, 1 \leq i \leq n_a $,  $\hat D_{il}^b = {f_i^b}'\widehat \beta^a_l, 1 \leq i \leq n_b$,  $\widehat D^k_i  = ( {\widehat{D}_{i1}^{k}},\ldots, \widehat D_{ik_e}^k, {w^k_i}')', k=a, b $,
we form the IV estimates in the two subsamples:
$$
\widehat \alpha_a = \Ena [\hat D_i^a  {d_i^a}' ]^{-1} \Ena [\hat D_i^a y^a_{i}] \ \ \ \mbox{and} \ \ \  \widehat \alpha_b = \Enb [\widehat D_i^b  {d_i^b}' ]^{-1} \Enb [{\widehat D_i^b} y^b_{i}].
$$
Then we combine the estimates into one
 \begin{equation}\label{Def:SplitIVcombined}
\widehat \alpha_{ab} = ( n_a  \Ena [\hat D_i^a  \hat D_i^{a}{}']  +  n_b \Enb [\hat D_i^b  \hat D_i^{b}{}']   )^{-1} (n_a  \Ena [\hat D_i^a  \hat D_i^{a}{}'] \widehat \alpha_a +   n_b \Enb [\hat D_i^b  \hat D_i^{b}{}'] \widehat \alpha_b   ).
 \end{equation}

The following result shows that the split-sample IV estimator $\widehat \alpha_{ab}$ has the same large sample properties
 as the estimator $\widehat \alpha$ of the previous section but requires a weaker growth condition.

\begin{theorem}[Inference with a Split-Sample IV Based on Lasso or Post-Lasso]\label{Thm:InferenceSplitSampleIV}
Suppose that data $(y_i, \x_i, d_i)$ are i.n.i.d. and obey the linear IV model described in Section 2. Suppose
also that Conditions AS, RF, SM, (\ref{the principle lambda}), (\ref{Def:AsympValidPenaltyLoading}) and (\ref{Def:AmeliorationSets}) hold, except that instead of growth condition $s^2\log^2 (p\vee n) = o(n)$ we now have a weaker growth condition  $ s \log (p\vee n) = o(n)$. Suppose also that Condition SE hold for $M^k=\Enk[f^k_i f^k_i{}']$ for $k=a,b$. Let $\widehat D^k_{il}= f^k_i{}'\hat \beta^{k^c}_l$ where $\hat \beta^{k^c}_l$ is the Lasso or Post-Lasso estimator  applied to the subsample $ \{(d^{k^c}_{li},f^{k^c}_i) : 1\leq i \leq n_{k^c}\}$ for  $k =a,b$, and $k^c=\{a,b\}\setminus k$.  Then the split-sample IV estimator based on equation (\ref{Def:SplitIVcombined}) is $\sqrt{n}$-consistent and asymptotically normal,  as $n \to \infty$
$$
 (Q^{-1}\Omega Q^{-1})^{-1/2}\sqrt{n}(\widehat\alpha_{ab} - \alpha_0) \to_d N(0,I),
$$
for $\Omega :=  \barEp[\epsilon_i^2 D(\x_i) D(\x_i)']$ and $Q:= \barEp[D(\x_i) D(\x_i)']$.
Moreover, the result above continues to hold with
  $\Omega$ replaced by $\hat \Omega := \En [\hat \epsilon_i^2 \widehat D(\x_i) \widehat D(\x_i)'] $  for  $\hat \epsilon_i =
y_i - d_i'\hat \alpha_{ab}$, and $Q$ replaced by $\hat Q: = \En [\widehat D(\x_i) \widehat D(\x_i)']$.
\end{theorem}

\section{Simulation Experiment}

The previous sections' results suggest that using Lasso for fitting first-stage regressions should result in IV estimators with good estimation and inference properties.  In this section, we provide simulation evidence regarding these properties in a situation where there are many possible instruments.  We also compare the performance of the developed Lasso-based estimators to many-instrument robust estimators that are available in the literature.

Our simulations are based on a simple instrumental variables model data generating process (DGP):
\begin{align*}
\begin{array}{ll}
y_i &= \beta d_i + e_i \\
d_i &= z_i'\Pi + v_i
\end{array}  \ \ \ \ \ \
(e_i,v_i) &\sim N\(0,\(\begin{array}{cc} \sigma^2_e & \sigma_{ev} \\ \sigma_{ev} & \sigma^2_{v}\end{array}\)\) \ \text{i.i.d.}
\end{align*}
where $\beta=1$ is the parameter of interest, and $z_i = (z_{i1},z_{i2},...,z_{i100})' \sim N(0,\Sigma_Z)$ is a 100 x 1 vector with $E[z_{ih}^2] = \sigma^2_z$ and $\textrm{Corr}(z_{ih},z_{ij}) = .5^{|j-h|}$.  In all simulations, we set $\sigma^2_e = 1$ and $\sigma^2_z = 1$.  We also set $\textrm{Corr}(e,v) = 0.6.$

For the other parameters, we consider various settings.  We provide results for sample sizes, $n$, of 100 and 250.  We set $\sigma^2_v$ so that the unconditional variance of the endogenous variable equals one; i.e. $\sigma^2_v = 1 - \Pi'\Sigma_Z\Pi$.  We use three different settings for the pattern of the first-stage coefficients, $\Pi$.  In the first, we set $\Pi = C\widetilde\Pi = C (1,.7,.7^2,...,.7^{98},.7^{99})'$; we term this the ``exponential'' design.  In the second and third case, we set $\Pi = C\widetilde\Pi = C(\iota_s,0_{n-s})'$ where $\iota_s$ is a $1 \times s$ vector of ones and $0_{n-s}$ is a $1 \times n-s$ vector of zeros.  We term this the ``cut-off'' design and consider two different values of $s$, $s = 5$ and $s = 50$.  In the exponential design, the model is not literally sparse although the majority of explanatory power is contained in the first few instruments.  While the model is exactly sparse in the cut-off design, we expect Lasso to perform poorly with $s = 50$ since treating $\frac{s^2\log^2 p}{n}$ as vanishingly small seems like a poor approximation given the sample sizes considered.  We consider different values of the constant $C$ that are chosen to generate target values for the concentration parameter, $\mu^2 = \frac{n\Pi'\Sigma_Z\Pi}{\sigma^2_v}$, which plays a key role in the behavior of IV estimators; see, e.g. \citen{stock:survey} or \citen{hhn:weakiv}.\footnote{The concentration parameter is closely related to the first-stage Wald statistic and first-stage F-statistic for testing that the coefficients on the instruments are equal to 0.  Under homoscedasticity, the first-stage Wald statistic is $W = \widehat\Pi' (Z'Z) \widehat\Pi/\widehat\sigma^2_v$ and the first-stage F-statistic is $W/\dim(Z)$.}  Specifically, we choose $C$ to solve $\mu^2 = \frac{nC^2 \widetilde\Pi \Sigma_Z \widetilde\Pi}{1-C^2\widetilde\Pi \Sigma_Z \widetilde\Pi}$ for $\mu^2 = 30$ and for $\mu^2 = 180$.  These values of the concentration parameter were chosen by using estimated values from the empirical example reported below as a benchmark.\footnote{In the empirical example, first-stage Wald statistics based on the selected instruments range from between 44 and 243.  In the cases with constant coefficients, our concentration parameter choices correspond naturally to ``infeasible F-statistics'' defined as $\mu^2/s$ of 6 and 36 with $s = 5$ and .6 and 3.6 with $s = 50$.  In an online appendix, we provide additional simulation results.  The results reported in the current section are sufficient to capture the key patterns.}

For each setting of the simulation parameter values, we report results from seven different procedures.  A simple possibility when presented with many instrumental variables (with $p < n$) is to just estimate the model using 2SLS and all of the available instruments.  It is well-known that this will result in poor-finite sample properties unless there are many more observations than instruments; see, for example, \citen{bekker}.  The estimator proposed in \citen{fuller} (FULL) is robust to many instruments (with $p < n$) as long as the presence of many instruments is accounted for when constructing standard errors for the estimators; see \citen{hhn:weakiv} for example.\footnote{FULL requires a user-specified parameter.  We set this parameter equal to one which produces a higher-order unbiased estimator.  See \citen{hhk:weakmse} for additional discussion.  LIML is another commonly proposed estimator which is robust to many instruments.  In our designs, its performance was generally similar to that of FULL, and we report only FULL for brevity.}  We report results for these estimators in rows labeled 2SLS(100) and FULL(100) respectively.\footnote{With $n = 100$, estimates are based on a randomly selected 99 instruments.}  For our variable selection procedures, we use Lasso to select among the instruments using the refined data-dependent penalty loadings given in (\ref{choice of loadings2}) described in Appendix A and consider two post-model selection estimation procedures.  The first, Post-Lasso, runs 2SLS using the instruments selected by Lasso; and the second, Post-Lasso-F, runs FULL using the instruments selected by Lasso.  In cases where no instruments are selected by Lasso, we use the point-estimate obtained by running 2SLS with the single instrument with the highest within sample correlation to the endogenous variable as the point estimate for Post-Lasso and Post-Lasso-F.  In these cases, we use the sup-Score test for performing inference.\footnote{Inference based on the asymptotic approximation when Lasso selects instruments and based on the sup-Score test when Lasso fails to select instruments is our preferred procedure.}  We report inference results based on the weak-identification robust sup-score testing procedure in rows labeled ``sup-Score''.

The other two procedure ``Post-Lasso (Ridge)'' and ``Post-Lasso-F (Ridge)'' use a combination of Ridge regression, Lasso, and sample-splitting.  For these procedures, we randomly split the sample into two equal-sized parts.  Call these sub-samples ``sample A'' and ``sample B.''  We then use leave-one-out cross-validation with only the data in sample A to select a ridge penalty parameter and then estimate a set of ridge coefficients using this penalty and the data from sample A.  We then use the data from sample B with these coefficients estimated using only data from sample A to form first-stage fitted values for sample B.  Then, we take the full-set of instruments augmented with the estimated fitted value just described and perform Lasso variable selection using only the data from sample B.  We use the selected variables to run either 2SLS or Fuller in sample B to obtain estimates of $\beta$ (and associated standard errors), say $\widehat\beta_{B,2SLS}$ ($s_{B,2SLS}$) and $\widehat\beta_{B,Fuller}$ ($s_{B,Fuller}$).  We then repeat this exercise switching sample A and B to obtain estimates of $\beta$ (and associated standard errors) from sample A, say $\widehat\beta_{A,2SLS}$ ($s_{A,2SLS}$) and $\widehat\beta_{A,Fuller}$ ($s_{A,Fuller}$).  Post-Lasso (Ridge) is then $w_{A,2SLS} \widehat\beta_{A,2SLS} + (1-w_{A,2SLS}) \widehat\beta_{B,2SLS}$ for $w_{A,2SLS} = \frac{s^2_{B,2SLS}}{s^2_{A,2SLS}+s^2_{B,2SLS}}$, and Post-Lasso-F (Ridge) is defined similarly.  If instruments are selected in one subsample but not in the other, we put weight one on the estimator from the subsample where instruments were selected.  If no instruments are selected in either subsample, we use the single-instrument with the highest correlation to obtain the point estimate and use the sup-score test for performing inference.

For each estimator, we report median bias (Med. Bias), median absolute deviation (MAD), and rejection frequencies for 5\% level tests (rp(.05)).  For computing rejection frequencies, we estimate conventional, homoscedastic 2SLS standard errors for 2SLS(100) and Post-Lasso and the many instrument robust standard errors of \citen{hhn:weakiv} which rely on homoscedasticity for FULL(100) and Post-Lasso-F.  We report the number of cases in which Lasso selected no instruments in the column labeled N(0).

We summarize the simulation results in Table 1.  It is apparent that the Lasso procedures are dominant when $n = 100.$  In this case, the Lasso-based procedures outperform 2SLS(100) and FULL(100) on all dimensions considered.  When the concentration parameter is 30 or $s = 50$, the instruments are relatively weak, and Lasso accordingly selects no instruments in many cases.  In these cases, inference switches to the robust sup-score procedure which controls size.  With a concentration parameter of 180, the instruments are relatively more informative and sparsity provides a good approximation in the exponential design and $s = 5$ cut-off design.  In these cases, Lasso selects instruments in the majority of replications and the procedure has good risk and inference properties relative to the other procedures considered.  In the $n = 100$ case, the simple Lasso procedures also clearly dominate Lasso augmented with Ridge as this procedure often results in no instruments being selected and relatively low power; see Figure 1.   We also see that the sup-score procedure controls size across the designs considered.

In the $n = 250$ case, the conventional many-instrument asymptotic sequence which has $p$ proportional to $n$ but $p/n < 1$ provides a reasonable approximation to the DGP, and one would expect FULL to perform well.  In this case, 2SLS(100) is clearly dominated by the other procedures.  However, there is no obvious ranking between FULL(100) and the Lasso-based procedures.  With $s = 50$, sparsity is a poor approximation in that there is signal in the combination of the 50 relevant instruments but no small set of instruments has much explanatory power.  In this setting, FULL(100) has lower estimation risk than the Lasso procedure which is not effectively able to capture the diffuse signal though both inference procedures have size close to the prescribed level.  Lasso augemented with the Ridge fit also does relatively well in this setting, being roughly on par with FULL(100).  In the exponential and cut-off with $s = 5$ designs, sparsity is a much better approximation.  In these cases, the simple Lasso-based estimators have smaller risk than FULL(100) or Lasso with Ridge and produce tests that have size close to the nominal 5\% level.  Finally, we see that the sup-score procedure continues to control size with $n = 250$.

Given that the sup-score procedure uniformly controls size across the designs considered but is actually substantially undersized, it is worth presenting additional results regarding power.  We plot size-adjusted power curves for the sup-score test, Post-Lasso-F, Post-Lasso-F (Ridge), and FULL(100) across the different designs in the $\mu^2 = 180$ cases in Figure 1.  We focus on $\mu^2 = 180$ since we expect it is when identification is relatively strong that differences in power curves will be most pronounced.  From these curves, it is apparent that the robustness of the sup-score test comes with a substantial loss of power in cases where identification is strong.  Exploring other procedures that are robust to weak identification, allow for $p \gg n$, and do not suffer from such power losses may be interesting for future research.

\subsection{Conclusions from Simulation Experiments}

The evidence from the simulations is supportive of the derived theory and favorable to Lasso-based IV methods.  The Lasso-IV estimators clearly dominate on all metrics considered when $p = n$ and $s \ll n$.  The Lasso-based IV estimators generally have relatively small median bias and estimator risk and do well in terms of testing properties, though they do not dominate FULL in these dimensions across all designs with $p < n$.  The simulation results verify that FULL becomes more appealing as the sparsity assumption breaks down.  This breakdown of sparsity is likely in situations with weak instruments, be they many or few, where none of the first-stage coefficients are well-separated from zero relative to sampling variation.  Overall, the simulation results show that simple Lasso-based procedures can usefully complement other many-instrument methods.

\section{The Impact of Eminent Domain on Economic Outcomes}

As an example of the potential application of Lasso to select instruments, we consider IV estimation of the effects of federal appellate court decisions regarding eminent domain on a variety of economic outcomes.\footnote{See \citen{chen:yeh:takings} for a detailed discussion of the economics of takings law (or eminent domain), relevant institutional features of the legal system, and a careful discussion of endogeneity concerns and the instrumental variables strategy in this context.}
To try to uncover the relationship between takings law and economic outcomes, we estimate structural models of the form
$$
y_{ct} = \alpha_c + \alpha_t + \gamma_c t + \beta \ Takings \ Law_{ct} + W_{ct}'\delta + \epsilon_{ct}
$$
where $y_{ct}$ is an economic outcome for circuit $c$ at time $t$, \textit{Takings Law}$_{ct}$ represents the number of pro-plaintiff appellate takings decisions in circuit $c$ and year $t$; $W_{ct}$ are judicial pool characteristics,\footnote{The judicial pool characteristics are the probability of a panel being assigned with the characteristics used to construct the instruments.  There are 30, 33, 32, and 30 controls available for FHFA house prices, non-metro house prices, Case-Shiller house prices, and GDP respectively.} a dummy for whether there were no cases in that circuit-year, and the number of takings appellate decisions; and $\alpha_c$, $\alpha_t$, and $\gamma_c t$ are respectively circuit-specific effects, time-specific effects, and circuit-specific time trends.  An appellate court decision is coded as pro-plaintiff if the court ruled that a taking was unlawful, thus overturning the government's seizure of the property in favor of the private owner.  We construe pro-plaintiff decisions to indicate a regime that is more protective of individual property rights.  The parameter of interest, $\beta$, thus represents the effect of an additional decision upholding individual property rights on an economic outcome.

We provide results using four different economic outcomes: the log of three home-price-indices and log(GDP).  The three different home-price-indices we consider are the quarterly, weighted, repeat-sales FHFA/OFHEO house price index that tracks single-family house prices at the state level for metro (FHFA) and non-metro (Non-Metro) areas  and the Case-Shiller home price index (Case-Shiller) by month for 20 metropolitan areas based on repeat-sales residential housing prices. We also use state level GDP from the Bureau of Economic Analysis to form log(GDP).  For simplicity and since all of the controls, instruments, and the endogenous variable vary only at the circuit-year level, we use the within-circuit-year average of each of these variables as the dependent variables in our models.  Due to the different coverage and time series lengths available for each of these series, the sample sizes and sets of available controls differ somewhat across the outcomes.  These differences lead to different first-stages across the outcomes as well.
The total sample sizes are 312 for FHFA and GDP which have identical first-stages.  For Non-Metro and Case-Shiller, the sample sizes are 110 and 183 respectively.

The analysis of the effects of takings law is complicated by the possible endogeneity between governmental takings and takings law decisions and economic variables.  To address the potential endogeneity of takings law, we employ an instrumental variables strategy based on the identification argument of \citen{chen:sethi} and \citen{chen:yeh:takings} that relies on the random assignment of judges to federal appellate panels.  Since judges are randomly assigned to three judge panels to decide appellate cases, the exact identity of the judges and, more importantly, their demographics are randomly assigned conditional on the distribution of characteristics of federal circuit court judges in a given circuit-year.  Thus, once the distribution of characteristics is controlled for, the realized characteristics of the randomly assigned three judge panel should be unrelated to other factors besides judicial decisions that may be related to economic outcomes.

There are many potential characteristics of three judge panels that may be used as instruments.  While the basic identification argument suggests any set of characteristics of the three judge panel will be uncorrelated with the structural unobservable, there will clearly be some instruments which are more worthwhile than others in obtaining precise second-stage estimates.  For simplicity, we consider only the following demographics: gender, race, religion, political affiliation, whether the judge's bachelor was obtained in-state, whether the bachelor is from a public university, whether the JD was obtained from a public university, and whether the judge was elevated from a district court along with various interactions.  In total, we have 138, 143, 147, and 138 potential instruments for FHFA prices, non-metro prices, Case-Shiller, and GDP respectively that we select among using Lasso.\footnote{Given the sample sizes and numbers of variables, estimators using all the instruments without shrinkage are only defined in the GDP and FHFA data.  For these outcomes, the \citen{fuller} point estimate (standard error) is -.0020 (3.123) for FHFA and .0120 (.1758) for GDP.}

Table 3 contains estimation results for $\beta$.  We report OLS estimates and results based on three different sets of instruments.  The first set of instruments, used in the rows labeled 2SLS, are the instruments adopted in \citen{chen:yeh:takings}.\footnote{\citen{chen:yeh:takings} used two variables motivated on intuitive grounds, whether a panel was assigned an appointee who did not report a religious affiliation and whether a panel was assigned an appointee who earned their first law degree from a public university, as instruments.}  We consider this the baseline.    The second set of instruments are those selected through Lasso using the refined data-driven penalty.\footnote{Lasso selects the number of panels with at least one appointee whose law degree is from a public university (Public) cubed for GDP and FHFA.  In the Case-Shiller data, Lasso selects Public and Public squared.  For non-metro prices, Lasso selects Public interacted with the number of panels with at least one member who reports belonging to a mainline protestant religion, Public interacted with the number of panels with at least one appointee whose BA was obtained in-state (In-State), In-State interacted with the number of panels with at least one non-white appointee, and the interaction of the number of panels with at least one Democrat appointee with the number of panels with at least one Jewish appointee.}  The number of instruments selected by Lasso is reported in the row ``S''.  We use the Post-Lasso 2SLS estimator and report these results in the rows labeled ``Post-Lasso''.  The third set of instruments is simply the union of the first two instrument sets.  Results for this set of instruments are in the rows labeled ``Post-Lasso+''.  In this case, ``S'' is the total number of instruments used.  In all cases, we use heteroscedasticity consistent standard error estimators.  Finally, we report the value of the test statistic discussed in Section 4.3.1 comparing estimates using the first and second sets of instruments in the row labeled ``Spec. Test''.

The most interesting results from the standpoint of the present paper are found by comparing first-stage Wald-statistics and estimated standard errors across the instrument sets.  The Lasso instruments are clearly much better first-stage predictors as measured by the first-stage Wald-statistic compared to the \citen{chen:yeh:takings} benchmark.  Given the degrees of freedom, this increase obviously corresponds to Lasso-based IV providing a stronger first-stage relationship for FHFA prices, GDP, and the Case-Shiller prices. In the non-metro case, the p-value from the Wald test with the baseline instruments of \citen{chen:yeh:takings} is larger than that of the Lasso-selected instruments.  This improved first-stage prediction is associated with the resulting 2SLS estimator having smaller estimated standard errors than the benchmark case for non-metro prices, Case-Shiller prices, and GDP.  The reduction in standard errors is sizable for both non-metro and Case-Shiller.  Tthe standard error estimate is somewhat larger in the FHFA case despite the improvement in first-stage prediction.  Given that the Post-Lasso first-stage produces a larger first-stage Wald-statistic while choosing fewer instruments than the benchmark suggests that we might prefer the Post-Lasso results in any case.  We also see that the test statistics for testing the difference between the estimate using the \citen{chen:yeh:takings} instruments and the Post-Lasso estimate is equal to zero are uniformly small.  Given the small differences between estimates using the first two sets of instruments, it is unsurprising that the results using the union of the two instrument sets are similar to those already discussed.

The results are also economically interesting.  The point estimates for the effect of an additional pro-plaintiff decision, a decision in favor of individual property holders, are positive, suggesting these decisions are associated with increases in property prices and GDP.  These point estimates are all small, and it is hard to draw any conclusion about the likely effect on GDP or the FHFA index given their estimated standard errors.  On the other hand, confidence intervals for non-metro and Case-Shiller constructed at usual confidence levels exclude zero.  Overall, the results do suggest that the causal effect of decisions reinforcing individual property rights is an increase in the value of holding property, at least in the short term.
The results are also consistent with the developed asymptotic theory in that the 2SLS point-estimates based on the benchmark instruments are similar to the estimates based on the Lasso-selected instruments while Lasso produces a stronger first-stage relationship and the Post-Lasso estimates are more precise in three of the four cases.  The example suggests that there is the potential for Lasso to be fruitfully employed to choose instruments in economic applications.

\appendix

\section{Implementation Algorithms} It is useful to organize the precise
implementation details into the following algorithm.  We
establish the asymptotic validity of this algorithm in the subsequent sections.
Feasible options for setting the penalty level and the loadings for $j=1,\ldots,p$, and $l=1,\ldots,k_e$ are
\begin{equation}\label{choice of loadings2}\begin{array}{llll}
&\text{initial}  &   \hat \gamma_{lj}  =   \sqrt{\En[f^2_{ij} (d_{i{l}}- \bar d_{l} )^2]},  &  \lambda =  2c\sqrt{n} \Phi^{-1}(1- \gamma/(2k_ep) ), \\
& \text{refined } &   \hat \gamma_{lj}  = \sqrt{\En [f^2_{ij} \widehat v^2_{i{l}}]},  & \lambda = 2c\sqrt{n} \Phi^{-1}(1- \gamma/(2k_ep) ),
\end{array}\end{equation}
where  $c>1$ is a constant, $\gamma \in (0,1)$, $\bar d_{l}:= \En[d_{il}]$ and $\hat v_{il}$ is an estimate of $v_{il}$.  Let $K \geq 1$ denote a bounded
number of iterations. We used $c=1.1$, $\gamma = 0.1/\log(p\vee n)$, and $K = 15$ in the simulations. In what follows Lasso/Post-Lasso estimator indicates that the practitioner can apply either the Lasso or Post-Lasso estimator. Our preferred approach uses Post-Lasso at every stage.

\begin{algorithm}[Lasso/Post-Lasso Estimators]\label{alg1} 
(1) For each $l=1,\ldots,k_e$, specify penalty loadings according to the initial option in (\ref{choice of loadings2}).  Use these penalty loadings in computing the Lasso/Post-Lasso estimator $\hat \beta_{l}$ via equations (\ref{Def:LASSOmain}) or (\ref{Def:TwoStep}). Then compute residuals $\hat v_{il}= d_{li}- f_i'\hat \beta_{l}$, $i=1,...,n$.
(2) For each $l=1,\ldots,k_e$, update the penalty loadings according to the refined option in (\ref{choice of loadings2}) and update the Lasso/Post-Lasso estimator $\hat \beta_{l}$.  Then compute a new set of residuals using the updated Lasso/Post-Lasso coefficients $\hat v_{il}= d_{li}- f_i'\hat \beta_{l}$, $i=1,...,n$. (3) Repeat the previous step $K$ times.
\end{algorithm}

If  the Algorithm \ref{alg1} selected no instruments other than intercepts, or, more generally if $\En[\hat D_{il} \hat D_{il} ']$ is near-singular, proceed to Algorithm \ref{alg3}; otherwise,
we recommend the following algorithm.

\begin{algorithm}[IV Inference Using Estimates of Optimal Instrument]\label{alg2} Compute the estimates of the optimal instrument, $\hat D_{il} = f_i' \hat\beta_{l}$, for $i=1,...,n$ and each $l=1,...,k_e$, where $\hat \beta_l$ is computed by Algorithm \ref{alg1}.
 Compute the IV estimator $\widehat \alpha =  \En [\widehat D_id_i']^{-1}  \En [\widehat D_i y_i]$. (2) Compute estimates of the asymptotic variance matrix $\hat Q^{-1}\hat \Omega\hat Q^{-1}$ where $\hat \Omega := \En [\hat \epsilon_i^2 \widehat D_i \widehat D_i'] $  for  $\hat \epsilon_i =
y_i - d_i'\hat \alpha$, and $\hat Q: = \En [\widehat D_i \widehat D_i']$. (3) Proceed to perform conventional  inference using the normality
result (\ref{robust normality result}).
\end{algorithm}

The following algorithm is only invoked if the weak instruments problem has been diagnosed, e.g., using the methods of \citen{StockYogo2005}.
In the algorithm below, $\mathcal{A}_1$ is the parameter space, and $\mathcal{G}_1 \subset \mathcal{A}_1$ is a grid of potential values for $\alpha_1$.  Choose the confidence level $1-\gamma$ of the interval, and set $\Lambda(1-\gamma)=c\sqrt{n}\Phi^{-1}(1-\gamma/2p)$.

\begin{algorithm}[IV Inference Robust to Weak Identification]\label{alg3}
(1) Set $\mathcal{C}=\emptyset$. (2) For each $a\in  \mathcal{G}_1$ compute $\Lambda_{a}$ as in (\ref{def:sup-score}). If $\Lambda_{ a} \leq \Lambda(1-\gamma)$ add $a$ to $\mathcal{C}$.
(3) Report $\mathcal{C}$.
\end{algorithm}

\section{Tools}

The following useful lemma is a consequence of moderate deviations theorems for self-normalized sums in  \citen{jing:etal} and \citeasnoun{delapena}.

We shall be using the following result -- Theorem 7.4 in \citen{delapena}. Let $X_1,...,X_n$ be independent, zero-mean variables, and
$
S_n = \sum_{i=1}^n X_i,  \ \ V^2_n = \sum_{i=1}^n X^2_i.
$
For $0 < \mu \leq 1$  set
$
B^2_n = \sum_{i=1}^n \Ep X_i^2, \ \ L_{n, \mu} = \sum_{i=1}^n \Ep|X_i|^{2 + \mu}, \ \ d_{n, \mu} = B_n/L_{n,\mu}^{1/(2 + \mu)}.
$
Then uniformly in $ 0 \leq x \leq d_{n, \mu}$,
 \begin{eqnarray*}
\frac{\Pr(S_n/V_n  \geq x) }{ \bar \Phi(x)} = 1 + O(1) \left( \frac{1+x}{d_{n,\mu}}\right)^{2 + \mu}, \ \ \frac{\Pr(S_n/V_n  \leq -x) }{\Phi(-x)} = 1 + O(1) \left( \frac{1+x}{d_{n,\mu}}\right)^{2 + \mu},
  \end{eqnarray*}
where the terms $O(1)$ are bounded in absolute value by a universal constant $A$,  $\bar \Phi := 1 - \Phi$, and $\Phi$ is the cumulative distribution function of a standard Gaussian random variable.

\begin{lemma}[Moderate Deviation Inequality for Maximum of a Vector]\label{Lemma:SNMD} Suppose that
$$ \mathcal{S}_{j} =  \frac{\sum_{i=1}^n U_{ij}}{\sqrt{ \sum_{i=1}^n U^2_{ij}}},$$
where $U_{ij}$ are independent variables across $i$ with mean zero.  We have that
$$
\Pr \left( \max_{1 \leq j\leq p }|\mathcal{S}_{j}|  >  \Phi^{-1}(1- \gamma/2p)  \right) \leq \gamma \(1 + \frac{A}{\ell^3_n}\),
$$
where $A$ is an absolute constant, provided that for $\ell_n>0$
$$
0 \leq \Phi^{-1}(1- \gamma/(2p))  \leq \frac{n^{1/6}}{\ell_n} \min_{1\leq j \leq p} M[U_j]-1, \ \  \ M[U_j] := \frac{\left( \frac{1}{n} \sum_{i=1}^n E U_{ij}^2\right)^{1/2}}{\left(\frac{1}{n} \sum_{i=1}^n E|U_{ij}^3| \right)^{1/3}}.
$$
\end{lemma}

\begin{proof}[Proof of Lemma \ref{Lemma:SNMD}.]
Step 1. We first note the following simple consequence  of the result
of Theorem 7.4 in \citen{delapena}.  Let $X_{1,n},...,X_{n,n}$ be the triangular array of i.n.i.d, zero-mean random variables. Suppose that
$$n^{1/6} M_n/\ell_n \geq 1,  \  \ M_n := \frac{ (\frac{1}{n}\sum_{i=1}^n \Ep X_{i,n}^2 )^{1/2}}{ (\frac{1}{n}\sum_{i=1}^n\Ep|X_{i,n}|^3)^{1/3} }.$$
Then uniformly on
$0 \leq  x \leq  n^{1/6} M_n/\ell_n-1$, the
quantities
$S_{n,n} = \sum_{i=1}^n X_{i,n}$ and  $V^2_{n,n} = \sum_{i=1}^n X^2_{i,n}$
obey
$$
\left |\frac{\Pr(|S_{n,n}/V_{n,n}|  \geq x) }{ 2 \bar \Phi(x)} - 1 \right |  \leq
\frac{A}{ \ell_n^3}. \\
$$
This corollary follows by the application of the quoted theorem to the case with $\mu =1$.
The calculated  error bound follows from the triangular inequalities and conditions on $\ell_n$ and $M_n$.

Step 2.  It follows that
\begin{eqnarray*}
&& \Pr \left( \max_{1 \leq j\leq p }|\mathcal{S}_{j}|  >  \Phi^{-1}(1- \gamma/2p)  \right)  \leq_{(1)} p  \max_{1 \leq j\leq p } \Pr \left( |\mathcal{S}_{j}|  >  \Phi^{-1}(1- \gamma/2p)  \right) \\
&&  =_{(2)} p   \Pr \left( |\mathcal{S}_{j_n}|  >  \Phi^{-1}(1- \gamma/2p)  \right)   \leq_{(3)} p 2\bar \Phi (  \Phi^{-1}(1- \gamma/2p)  ) \(1 + \frac{A}{\ell^3_n}\) \\
&&  \leq 2p \gamma/(2p) \(1 + \frac{A}{\ell^3_n}\) \leq  \gamma \(1 + \frac{A}{\ell^3_n}\),
\end{eqnarray*}
on the set
$
0 \leq \Phi^{-1}(1- \gamma/(2p))  \leq \frac{n^{1/6}}{\ell_n} M_{j_n}-1,
$
where  inequality (1) follows by the union bound,  equality (2) is the maximum taken
over finite set, so the maximum is attained at some $j_n \in \{1,...,p\}$, and the last
inequality follows by the application of Step 1, by setting $X_{i,n}= U_{i,j_n}$. \end{proof}


\section{Proof of Theorem 1}\label{AppendixThm1}

The proof of Theorem \ref{Thm:RateLASSO} has four steps.   The most important steps are the Steps 1-3.
One half  of Step 1 for bounding $\|\cdot\|_{2,n}$-rate follows the strategy of \citen{BickelRitovTsybakov2009}, but accommodates
data-driven penalty loadings.  The other half of Step 1 for bounding  the $\|\cdot\|_1$-rate is new for the nonparametric case. Step 2 innovatively uses the moderate deviation theory for self-normalized sums which allows us to
obtain sharp results for non-Gaussian and heteroscedastic errors as well as handle data-driven penalty loadings. Step 3 relates the ideal penalty loadings and the feasible penalty loadings.  Step 4 puts the results together to reach the conclusions.

Step 1. For $C>0$ and each $l=1,\ldots, k_e$, consider the following weighted restricted eigenvalue
$$ \kappa_{C}^l = \min_{\delta \in \mathbb{R}^p : \  \|\widehat \Upsilon^0_{l} {\delta}_{T^c_l}\|_{1} \leq C\|\widehat \Upsilon^0_{l} {\delta}_{T_l}\|_1, \|\delta\|_2 \neq 0}\frac{\sqrt{s}\|f_i'{\delta}\|_{2,n}}{\|\widehat \Upsilon^0_{l} {\delta}_{T_l}\|_1}.$$
This quantity controls the modulus of continuity between the prediction norm $\|f_i'{\delta}\|_{2,n}$ and the $\ell_1$-norm $\|\delta\|_1$ within a restricted region that depends on $l=1,\ldots,k_e$. Note that if $\displaystyle a = \min_{1\leq l\leq k_e}\min_{1\leq j\leq p} \widehat \Upsilon_{lj}^0 \leq \max_{1\leq l\leq k_e}\|\widehat \Upsilon_{l}^0\|_\infty = b,$ for every $C>0$, because $\{\delta \in \RR^p: \|\widehat \Upsilon^0_{l} {\delta}_{T^c_l}\|_{1} \leq C\|\widehat \Upsilon^0_{l} {\delta}_{T_l}\|_1 \} \subseteq \{\delta \in \RR^p: a\| {\delta}_{T^c_l}\|_{1} \leq bC\| {\delta}_{T_l}\|_1 \}$ and $\|\widehat \Upsilon^0_{l} {\delta}_{T_l}\|_1\leq b \| {\delta}_{T_l}\|_1 $, we have $$\min_{1\leq l\leq k_e} \kappa_C^l \geq (1/b)\kappa_{(bC/a)}(\En[f_if_i'])$$ where the latter is the restricted eigenvalue defined in (\ref{RE}). If $C=c_0=(uc+1)/(\ell c-1)$ we have $\min_{1\leq l\leq k_e} \kappa_{c_0}^l \geq (1/b) \kappa_{\bar C}(\En[f_if_i'])$. By Condition RF and by Step 3 of Appendix \ref{AppendixThm1} below, we have $a$ bounded away from zero and $b$ bounded from above with probability approaching one as $n$ increases. 

\comment{We assume that for $c_0 = (c+1)/(c-1)$, where $c>1$ is defined in (\ref{choice of lambda}), $m\semin{m+s}\geq 4c_0^2s\semax{m}$ for some $m\geq s$. Let $$\kappa_0 = \sqrt{\semin{s+m}}( 1 - c_0\sqrt{s\semax{m}/[ m\semin{s+m} ]})\geq \sqrt{\semin{s+m}}/2.$$
For example, if $\semax{s\log n}$ is uniformly bounded, and $\semin{s(1+\log n)}$ is uniformly bounded away from zero,
we have $\kappa_0 / \sqrt{\semin{s(1+\log n)}} \to 1$.

For instance, if $\Ep[f_if_i']$ has eigenvalues bounded away from zero and from above, $\|f_i\|_\infty < K $, and $s^2\log^5p = o(n)$, it follows that $\semin{s\log(n)}$ is bounded away from zero as $n$ grows.}

The main result of this step is the following lemma:

\begin{lemma}\label{Thm:ReviewModel1} Under Condition AS, if $\lambda/n \geq c \| S_{l}\|_{\infty}$, and $\widehat \Upsilon_l$ satisfies (\ref{Def:AsympValidPenaltyLoading}) with $u \geq 1 \geq \ell > 1/c$  then
$$
\|f_i'(\hat{\beta_{l}}-{\beta_{l}}_0)\|_{2,n}  \leq \(u + \frac{1}{c}\) \frac{\lambda \sqrt{s}}{n \kappa_{c_0}^l}+2c_s,
$$
$$ \|\hat\Upsilon_{l}^0 (\hat{\beta_{l}} - {\beta_{l}}_0 )\|_1 \leq  3c_0 \frac{\sqrt{s}}{\kappa_{2c_0}^l}\left( ( u + [1/c])\frac{\lambda \sqrt{s}}{n \kappa_{c_0}^l}+2c_s \right) +   \frac{3c_0n}{\lambda} c_s^2,$$
where $c_0 = (uc+1)/(\ell c-1)$.
\end{lemma}
\begin{proof}[Proof of Lemma \ref{Thm:ReviewModel1}]
Let ${\delta_{{l}}} := \widehat {\beta_{l}} - {\beta_{l0}}$. By optimality of
$\widehat {\beta_{l}}$ we have \begin{equation}\label{Eq:OptBeta}  \hat Q_l(\hat {\beta_{l}}) - \hat Q_l({\beta_{l0}})
\leq \frac{\lambda}{n}\(\| \hat\Upsilon_{l}  {\beta_{l0}}\|_1 -
\| \hat\Upsilon_{l} \hat{\beta_{l}}\|_1\). \end{equation}
Expanding the quadratic function $\hat Q_l$, and using that $S_l =  2(\widehat \Upsilon^0_{l})^{-1} \En[v_{il}f_i]$, we have
 \begin{equation}\label{Eq:QuadraticExp}
\begin{array}{rcl} \left| \  \hat Q_l(\hat {\beta_{l}}) - \hat Q_l({\beta_{l}}_0) - \|f_i'{\delta_{{l}}}\|^2_{2,n} \right| & = & |2\En[v_{il}f_i'\delta_l] + 2\En[a_{il}f_i'\delta_l] |\\
& \leq & \| S_{l}\|_{\infty} \| \widehat \Upsilon^0_{l}{\delta_{{l}}}\|_{1} + 2c_s\|f_i'{\delta_{{l}}}\|_{2,n}.
  \end{array}
 \end{equation}
So combining (\ref{Eq:OptBeta}) and (\ref{Eq:QuadraticExp}) with $\lambda/n \geq
c\|S_{l}\|_{\infty}$ and the conditions imposed on $\widehat \Upsilon_l$ in the statement of the theorem,
\begin{equation}\label{Eq:Basic} \begin{array}{rcl}\|f_i'{\delta_{{l}}}\|_{2,n}^2  & \leq & \displaystyle \frac{\lambda}{n}\(\| \hat\Upsilon_{l} {\delta_{{l}}}_{T_l}\|_1 - \| \hat\Upsilon_{l} {\delta_{{l}}}_{T^c_l}\|_1\) + \|S_{l}\|_{\infty}\|\widehat \Upsilon_l^0{\delta_{{l}}}\|_1 + 2c_s\|f_i'{\delta_{{l}}}\|_{2,n}\\
& & \\
& \leq & \displaystyle \(u + \frac{1}{c}\) \frac{\lambda}{n} \| \hat\Upsilon^0_{l} {\delta_{{l}}}_{T_l}\|_1 - \(\ell-\frac{1}{c}\)\frac{\lambda}{n}\|  \hat\Upsilon^0_{l} {\delta_{{l}}}_{T^c_l}\|_1 +2c_s\|f_i'{\delta_{{l}}}\|_{2,n}.\\
\end{array}
\end{equation}

To show the first statement of the Lemma we can assume $\|f_i'{\delta_{{l}}}\|_{2,n}
\geq 2c_s$, otherwise we are done. This condition together with
relation (\ref{Eq:Basic}) implies that for $c_0 = (uc + 1)/(\ell c - 1)$ we have
$\|\hat\Upsilon^0_{l} {\delta_{{l}}}_{T^c_l}\|_1 \leq  c_0\|\hat\Upsilon^0_{l}
{\delta_{{l}}}_{T_l}\|_1.$ Therefore, by definition of $\kappa_{c_0}^l$, we have $ \|\hat\Upsilon^0_{l} {\delta_{{l}T_l} }\|_1 \leq \sqrt{s}\|f_i'{\delta_{{l}}}\|_{2,n} / \kappa_{c_0}^l.$ Thus, relation
(\ref{Eq:Basic}) implies
$ \|f_i'{\delta_{{l}}}\|_{2,n}^2  \leq \(u + \frac{1}{c}\)  \frac{\lambda\sqrt{s}}{n \kappa_{c_0}^l}\|f_i'{\delta_{{l}}}\|_{2,n} + 2c_s\|f_i'{\delta_{{l}}}\|_{2,n}$
and the result follows.

 To establish the second  statement of the Lemma, we
consider two cases. First, assume $\|\hat\Upsilon^0_{l}
{\delta_{{l}}}_{T^c_l}\|_1 \leq 2c_0  \|\hat\Upsilon^0_{l} {\delta_{{l}}}_{T_l}\|_1.$ In this
case, by definition of $\kappa_{2c_0}^l$, we have
$$ \|\hat\Upsilon^0_{l} {\delta_{{l}}}\|_1 \leq (1+2c_0) \|\hat\Upsilon^0_{l} {\delta_{{l}}}_T\|_1 \leq (1+2c_0)\sqrt{s}\|f_i'{\delta_{{l}}}\|_{2,n}/\kappa_{2c_0}^l $$
and the result follows by applying the first bound to
$\|f_i'{\delta_{{l}}}\|_{2,n}$. On the other hand, consider the case that
\begin{equation}\label{L1:Case2}\|\hat\Upsilon^0_{l}
{\delta_{{l}}}_{T^c_l}\|_1 > 2c_0 \|\hat\Upsilon^0_{l} {\delta_{{l}}}_{T_l}\|_1\end{equation} which would already imply $\|f_i'{\delta_{{l}}}\|_{2,n} \leq 2c_s$ by (\ref{Eq:Basic}). Moreover,
$$\begin{array}{rcl}
 \|\hat\Upsilon^0_{l} {\delta_{{l}}}_{T^c_l}\|_1 & \leq_{(1)} & c_0  \|\hat\Upsilon^0_{l} {\delta_{{l}}}_{T_l}\|_1 + \frac{c}{\ell c-1} \frac{n}{\lambda}\|f_i'{\delta_{{l}}}\|_{2,n}(2c_s - \|f_i'{\delta_{{l}}}\|_{2,n})\\
 & \leq_{(2)} & c_0   \|\hat\Upsilon^0_{l} {\delta_{{l}}}_{T_l}\|_1 + \frac{c}{\ell c-1}\frac{n}{\lambda}  c_s^2 \leq_{(3)}  \frac{1}{2} \|\hat\Upsilon^0_{l} {\delta_{{l}}}_{T^c_l}\|_1 + \frac{c}{\ell c-1}\frac{n}{\lambda}  c_s^2,\\
\end{array}
$$
where (1) holds by (\ref{Eq:Basic}), (2) holds since
$\|f_i'{\delta_{{l}}}\|_{2,n}(2c_s - \|f_i'{\delta_{{l}}}\|_{2,n}) \leq \max_{x\geq 0} x(2c_s-x) \leq c_s^2$, and (3) follows from (\ref{L1:Case2}). Thus,
$$\|\hat\Upsilon^0_{l} {\delta_{{l}}}\|_1 \leq \(1 + \frac{1}{2c_0}\) \|\hat\Upsilon^0_{l} {\delta_{{l}}}_{T^c_l}\|_1 \leq
  \(1 + \frac{1}{2c_0 }\)\frac{2c}{\ell c-1}\frac{n}{\lambda} c_s^2, $$
and the result follows from noting that $c/(\ell c-1) \leq c_0/u \leq c_0$ and $1 + 1/2c_0 \leq 3/2$.
\end{proof}

Step 2.   In this step we prove a lemma about the quantiles of the maximum of the scores
$ S_{l} =  2 \En [(\widehat \Upsilon^0_{l})^{-1} f_i v_{il}],$
and use it to pin down the level of the penalty.
For  $
\lambda =  c2\sqrt{n}\Phi^{-1}(1-\gamma/(2k_ep)),$ 
we have that as $\gamma \to 0$ and $n \to \infty$,
$
\Pr \left( c \max_{1 \leq {l} \leq k_e } n \|S_{l}\|_{\infty}  >  \lambda  \right) = o(1),
$
provided that for some $b_n \to \infty$
$$
2\Phi^{-1}(1-\gamma/(2k_ep)) \leq \frac{n^{1/6}}{b_n} \min_{1\leq j \leq p, 1 \leq l \leq k_e} M_{jl}, \ \  \ M_{jl} := \frac{\barEp[f_{ij}^2 v_{il}^2]^{1/2}}{\barEp[|f_{ij}|^3 |v_{il}|^3]^{1/3}}.
$$
Note that the last condition is satisfied under our conditions for large $n$ for some $b_n \to \infty$, since $k_e$ is fixed, $\log (1/\gamma) \lesssim \log(p\vee n)$,
$K_n^{2}\log^3 (p\vee n) = o (n)$, and  $\min_{1\leq j \leq p, 1 \leq l \leq k_e} M_{jl} \gtrsim 1/ K_n^{1/3}$.  This result follows from the bounds
on moderate deviations of a maximum of a vector provided in Lemma \ref{Lemma:SNMD}, by $\bar \Phi(t) \leq \phi(t)/t$, $\max_{j\leq p,l\leq k_e} 1/ M_{jl} \lesssim K_n^{1/3}$, and $K_n^{2/3}\log (p\vee n) = o(n^{1/3})$  holding by Condition RF.

%

Step 3. The main result of this step is the following: Define the expected ``ideal" penalty loadings
$
 \Upsilon^0_{l} := \text{diag}\left(  \sqrt{\barEp[f_{i1}^2 v_{il}^2 ]},..., \sqrt{\barEp[f_{ip}^2 v_{il}^2 ]} \right ),
$
where the entries of  $\Upsilon^0_{l}$ are bounded away from zero and from above uniformly in $n$ by Condition RF.
Then the empirical ``ideal" loadings converge to the expected ``ideal" loadings:  $\max_{1 \leq l \leq k_e} \| \widehat \Upsilon^0_{l} - \Upsilon^0_{l}\|_{\infty}  \to_{\mathrm{P}} 0.$ This is assumed
in Condition RF.

Step 4.  Combining the results of all the steps above, given that $\lambda=2c\sqrt{n}\Phi^{-1}(1-\gamma/(2pk_e))\lesssim c\sqrt{n\log(pk_e/\gamma)}$, $k_e$ fixed, and asymptotic valid penalty loadings $\widehat \Upsilon_l$,
and using the bound $c_{s} \lesssim_{\mathrm{P}} \sqrt{s/n}$ from Condition AS, we obtain the conclusion that
$$
\|f_i'(\hat{\beta_{l}}-{\beta_{l0}})\|_{2,n}  \lesssim_{\mathrm{P}}  \frac{1}{\kappa_{c_0}^l}\sqrt{\frac{s \log (k_ep/\gamma)}{ n }} +  \sqrt{\frac{s }{ n }},
$$
which gives, by the triangular inequality and by $\| D_{il} - f_i'{\beta_{l0}}\|_{2,n} \leq c_s \lesssim_{\mathrm{P}} \sqrt{s/n}$ holding by Condition AS,
$$
\|\hat D_{il}  - D_{il}\|_{2,n}  \lesssim_{\mathrm{P}} \frac{1}{\kappa_{c_0}^l} \sqrt{\frac{s \log (k_ep/\gamma)}{ n }}.
$$
The first result follows since $\kappa_{\bar C}\lesssim_P \kappa_{c_0}^l$ by Step 1.

To derive the $\ell_1$-rate we apply the second result in Lemma \ref{Thm:ReviewModel1} as follows
$$\begin{array}{l}  \displaystyle  \|\hat{\beta_{l}} - {\beta_{l0}} \|_1 \displaystyle  \leq \|(\hat\Upsilon_{l}^0)^{-1}\|_{\infty} \|\hat\Upsilon_{l}^0 (\hat{\beta_{l}} - {\beta_{l}}_0 )\|_1 \\
  \displaystyle \lesssim_{\mathrm{P}}  \|(\hat\Upsilon_{l}^0)^{-1}\|_{\infty}\left (
 \frac{\sqrt{s}}{\kappa_{2c_0}^l} \left(\frac{1}{\kappa_{c_0}^l}\sqrt{\frac{s \log (k_ep/\gamma)}{ n }} +  \sqrt{\frac{s }{ n }}\right) +  \frac{1}{\sqrt{\log (p/\gamma)}} \frac{s}{\sqrt{n}} \right)
 \displaystyle
 \lesssim_{\mathrm{P}} \frac{1}{(\kappa_{2c_0}^l)^2}\sqrt{\frac{s^2 \log (k_ep/\gamma)}{ n }}.\end{array} $$
That yields the result since $\kappa_{2\bar C}\lesssim_P \kappa_{2c_0}^l$ by Step 1. \qed

\section{Proof of Theorem 2}\label{AppendixThm2}

The proof proceeds in three steps. The general strategy of Step 1 follows \citen{BC-SparseQR} and \citen{BC-PostLASSO}, but a major difference is the use of moderate deviation theory for self-normalized sums which allows us to
obtain the results for non-Gaussian and heteroscedastic errors as well as handle data-driven penalty loadings. The sparsity proofs are motivated by \citen{BC-PostLASSO} but adjusted for the data-driven penalty loadings that contain self-normalizing factors.

Step 1.  Here we derive a general performance bound for Post-Lasso, that actually contains
more information than the statement of the theorem. This lemma will be invoked in Step 3 below.

Let $F = [f_1;\ldots;f_n]'$ denote a $n$ by $p$ matrix and for a set of indices $S \subset \{1,\ldots,p\}$ we define
$\mathcal{P}_{S} = F[S](F[S]'F[S])^{-1} F[S]'$ denote the projection matrix on the columns associated with the indices in $S$.

\begin{lemma}[Performance of the Post-Lasso]\label{Thm:2StepMain}
Under Conditions AS and RF, let $\widehat T_l$ denote the support selected by $\widehat \beta_{l} = \widehat \beta_{lL}$, $\widehat T_l \subseteq \widehat I_l$, $\widetilde m_l =
|\widehat I_l\setminus T_l|$, and $\widehat \beta_{lPL}$ be
the Post-Lasso estimator based on $\widehat I_l$, $l=1,\ldots, k_e$. Then we have
$$
\begin{array}{l}
\displaystyle \max_{l\leq k_e} \| D_{il} - f_i'\hat \beta_{lPL}\|_{2,n} \lesssim_P   \max_{l\leq k_e}\left\{ \frac{\sqrt{(s+\widetilde m_l) \log(p k_e)}}{\sqrt{n \ \semin{s+\widetilde m_l}}}  + \|(D_l- \mathcal{P}_{\hat I_l}D_l)/\sqrt{n}\|_2\right\},\\
\displaystyle \max_{1\leq l \leq k_e}\| \widehat \Upsilon_l (\widehat\beta_{lPL} -\beta_{l0} ) \|_1 \leq  \max_{1\leq l \leq k_e} \frac{\( \|\widehat \Upsilon_l^0\|_\infty + \| \widehat \Upsilon_l - \widehat \Upsilon_l^0\|_\infty\)\sqrt{\widetilde m_l + s}}{\sqrt{\semin{\widetilde m_l +s}}} \|f_i'( \widehat\beta_{lPL} -\beta_{l0} )\|_{2,n}.\\
\end{array}
$$

If in addition $\lambda/n \geq
c\| S_l\|_{\infty}$, and $\widehat \Upsilon_l$ satisfies (\ref{Def:AsympValidPenaltyLoading}) with $u \geq 1 \geq \ell > 1/c$ in the first stage for Lasso for every $l=1,\ldots,k_e$, then we have
$$ \max_{l\leq k_e} \|(D_l- \mathcal{P}_{\hat I_l}D_l)/\sqrt{n}\|_2 \leq  \max_{l\leq k_e} \(u + \frac{1}{c}\) \frac{\lambda \sqrt{s}}{n \kappa_{c_0}^l}+3c_s.$$

\end{lemma}
\begin{proof}[Proof of Lemma \ref{Thm:2StepMain}]
We have that $D_l - F\hat\beta_{lPL} = ( I - \mathcal{P}_{\hat I_l})D_l - \mathcal{P}_{\hat I_l}v_l$ where $I$ is the identity operator. Therefore for every $l = 1,\ldots, k_e$ we have
\begin{equation}\label{eqPL}
\begin{array}{rl}
\| D_l - F\hat\beta_{lPL}\|_2 & \leq \|D_l - \mathcal{P}_{\hat I_l}D_l - \mathcal{P}_{\hat I_l}v_l\|_2 \\
 & \leq \|(I- \mathcal{P}_{\hat I_l})D_l\|_2 +\|\mathcal{P}_{\hat I_l}v_l\|_2.
\end{array}
\end{equation}
Note that $\|\mathcal{P}_{\hat I_l}v_l\|_2\leq \|(F[\hat I_l]/\sqrt{n})((F[\hat I_l]'F[\hat I_l]/n)^{-1}) \| \ \|F[\hat I_l]'v_l/\sqrt{n}\|_2$ where $\| M \|$ denotes the operator norm of the matrix $M$.
Since $\| M \|=\sqrt{{\rm max \ eig}\{M'M\}}$ we have $$\begin{array}{rl}
\|(F[\hat I_l]/\sqrt{n})((F[\hat I_l]'F[\hat I_l]/n)^{-1}) \| & = \sqrt{{\rm max \ eig}\{(F[\hat I_l]'F[\hat I_l]/n)^{-1}(F[\hat I_l]'F[\hat I_l]/n)(F[\hat I_l]'F[\hat I_l]/n)^{-1}\}}\\
& = \sqrt{{\rm max \ eig}\{(F[\hat I_l]'F[\hat I_l]/n)^{-1}\}}\\
 &\leq \sqrt{1/\semin{s+\widetilde m_l}},\end{array}$$ where $\widetilde m_l = |\hat I_l\setminus T_l|,$  so that the last term in (\ref{eqPL}) satisfies
$$\|\mathcal{P}_{\hat I_l}v_l\|_2 \leq  \sqrt{1/\semin{s+\widetilde m_l}}  \|F[\hat I_l]'v_l/\sqrt{n}\|_2 \leq \sqrt{\frac{s+\widetilde m_l}{\semin{s+\widetilde m_l}}}  \|F'v_l/\sqrt{n}\|_\infty.$$
Under Condition RF, by Lemma \ref{Lemma:SNMD} we have
$$\max_{l=1,\ldots, k_e}  \|F'v_l/\sqrt{n}\|_\infty \lesssim_P \sqrt{\log(pk_e)} \max_{l\leq k_e, j\leq p} \sqrt{\En[f_{ij}^2v_{il}^2]}.$$
Note that Condition RF also implies $\max_{l\leq k_e, j\leq p} \sqrt{\En[f_{ij}^2v_{il}^2]} \lesssim_P 1$ since $\max_{l\leq k_e, j\leq p} |(\En-\barEp)[f_{ij}^2v_{il}^2]|\to_P 0$ and $\max_{l\leq k_e, j\leq p}\barEp[f_{ij}^2v_{il}^2] \leq \max_{l\leq k_e, j\leq p}\barEp[f_{ij}^2\tilde d_{il}^2] \lesssim 1$.

These relations yield the first result.


Letting $\delta_l = \hat\beta_{lPL}-\beta_{l0}$, the statement regarding the $\ell_1$-norm of the theorem follows from
$$\|\widehat \Upsilon_l \delta_l\|_1 \leq  \|\widehat \Upsilon_l\|_\infty\|\delta_l\|_1 \leq \|\widehat \Upsilon_l\|_\infty \sqrt{\|\delta_l\|_0}\|\delta_l\|_2 \leq \|\widehat \Upsilon_l\|_\infty \sqrt{\|\delta_l\|_0}\|f_i'\delta_l\|_{2,n}/\sqrt{\semin{\|\delta_l\|_0}},$$
and noting that $\|\delta_l\|_0 \leq \widetilde m_l + s$ and $\|\widehat \Upsilon_l\|_\infty \leq \|\widehat \Upsilon_l^0\|_\infty + \|\widehat \Upsilon_l - \widehat \Upsilon_l^0\|_\infty$.

The last statement follows from noting that the Lasso solution provides an upper bound to the approximation of the best model based on $\widehat I_l$, since $\widehat T_l \subseteq \widehat I_l$, and the application of Lemma \ref{Thm:ReviewModel1}.\end{proof}

\begin{remark}[Comparison between Lasso and Post-Lasso performance]\label{comment: main}
Under mild conditions on the empirical Gram matrix and on the number of additional variables, Lemma \ref{Thm:Sparsity} below derives sparsity bounds on the model selected by Lasso, which establishes that
$$|\widehat T_l\setminus T_l| = \widehat m_l \lesssim_P s.$$ Under this condition, we have that the rate of Post-Lasso is no worse than Lasso's rate. This occurs despite the fact that Lasso may in general fail to correctly select the oracle model $T_l$ as a subset, that is $T_l \not \subseteq \widehat T_l$. However, if the oracle model has well-separated coefficients and the approximation error does not dominate the estimation error, then the Post-Lasso rate improves upon Lasso's rate. Specifically, this occurs if  Condition AS holds,  $\widehat m_l= o_{\mathrm{P}}(s)$ and $T_l \subseteq \widehat T_l$ wp $\to 1$, or if $T = \widehat T$ wp $\to 1$ as under the conditions of \citen{Wainright2006}. In such cases,  the rates found for Lasso are sharp, and they cannot be faster than $\sqrt{  s  \log p/n}$. Thus, the improvement in the rate of convergence of Post-Lasso over Lasso is strict these cases.
Note that, as shown in the proof of Lemma \ref{Lemma:SparsityLASSO},  a higher penalty level will tend to reduce $\widehat m_l$ but will increase the likelihood
of $T_l \not\subseteq \widehat T_l$. On the other hand, a lower penalty level will decrease the likelihood
of $T_l \not\subseteq \widehat T_l$ (bias) but will tend to increase $\widehat m_l$ (variance). The impact in the estimation of this trade off is captured by the last term of the bound in Lemma \ref{Thm:2StepMain}.
\end{remark}

Step 2. In this step we provide a sparsity bound for Lasso, which is important for establishing various rate results and fundamental to the analysis of Post-Lasso. It relies on the following lemmas.

\begin{lemma}[Empirical pre-sparsity for Lasso]\label{Lemma:SparsityLASSO}
Let $\widehat T_l$ denote the support selected by the Lasso estimator, $\hat m_l = |\widehat T_l \setminus T_l|$, and assume that $\lambda/n \geq c \|S_l\|_\infty$ and $u\geq 1 \geq \ell > 1/c$ as in Lemma \ref{Thm:ReviewModel1}. Then, for $c_0 = (uc+1)/(\ell c - 1)$ we have
$$\sqrt{\hat m_l} \leq \sqrt{\semax{\hat m_l}}\|(\widehat \Upsilon_l^{0})^{-1}\|_\infty c_0\[\frac{2\sqrt{s}}{\kappa_{c_0}^l} + \frac{6nc_s}{\lambda}\].$$
\end{lemma}
\begin{proof}[Proof of Lemma \ref{Lemma:SparsityLASSO}] We have from the optimality conditions that the Lasso estimator $\widehat \beta_l = \widehat \beta_{lL}$ satisfies  $$2\En[ \widehat \Upsilon_{lj}^{-1}f_{ij}(y_i-f_i'\hat\beta_l)] = \sign(\hat\beta_{lj})\lambda/n \ \text{ for each } \ j \in \widehat T_l\setminus T_l.  $$

Therefore, noting that $\|\widehat \Upsilon_l^{-1}\widehat \Upsilon_l^{0}\|_\infty \leq 1/\ell$, we have for $R=(a_{l1},\ldots,a_{ln})'$ and $F$ denoting the $n\times p$ matrix with rows $f_i'$, $i=1,\ldots,n$
{\small \begin{eqnarray*}
   &  &\sqrt{\hat m_l} \lambda    =   2\| (\widehat \Upsilon_l^{-1} F'(Y - F\hat \beta_l))_{\widehat T_l\setminus T_l} \|_2 \\
            &  & \leq   2\| (\widehat \Upsilon_l^{-1}F'(Y - R - F \beta_{l0}))_{\widehat T_l\setminus T_l} \|_2 + 2\| (\widehat \Upsilon_l^{-1}F'R)_{\widehat T_l\setminus T_l} \|_2 + 2\| (\widehat \Upsilon_l^{-1}F'F(\beta_{l0}-\hat \beta_l))_{\widehat T_l\setminus T_l} \|_2 \\
            & &  \leq  \sqrt{\hat m_l}\ n\|\widehat \Upsilon_l^{-1}\widehat \Upsilon_l^{0}\|_\infty\|S_l\|_{\infty} + 2 n\sqrt{\semax{\hat m_l}}\|\widehat \Upsilon_l^{-1}\|_\infty c_s+  2n \sqrt{ \semax{\hat m_l}} \|\widehat \Upsilon_l^{-1}\|_\infty\|f_i'(\hat \beta_l - \beta_{l0}) \|_{2,n},\\
            & & \leq \sqrt{\hat m_l}\ (1/\ell) \ n\|S_l\|_{\infty} + 2 n\sqrt{\semax{\hat m_l}}\frac{\|(\widehat \Upsilon_l^{0})^{-1}\|_\infty}{\ell} c_s+  2n \sqrt{ \semax{\hat m_l}} \frac{\|(\widehat \Upsilon_l^{0})^{-1}\|_\infty}{\ell}\|f_i'(\hat \beta_l- \beta_{l0}) \|_{2,n},
\end{eqnarray*}
}\!where we used that
$$\begin{array}{rcl}
&& \| (F'F(\beta_{l0}-\hat \beta_l))_{\widehat T_l\setminus T_l} \|_2 \\
&& = \sup_{\|\delta\|_0\leq \hat m_l, \|\delta\|_2\leq 1}|
\delta' F'F(\beta_{l0}-\hat \beta_l)| \leq
\sup_{\|\delta\|_0\leq \hat m_l, \|\delta\|_2\leq 1}\| \delta'F'\|_2\|F(\beta_{l0}-\hat \beta_l)\|_2 \\
&& \leq  \sup_{\|\delta\|_0\leq \hat
m_l, \|\delta\|_2\leq 1}\sqrt{| \delta'F'F\delta|}\|F(\beta_{l0}-\hat \beta_l)\|_2 \leq
n\sqrt{\semax{\hat m_l}}\|f_i'(\beta_{l0}-\hat \beta_l)\|_{2,n},
\end{array}$$ and
similarly $$\begin{array}{rcl}
\| (F'R)_{\widehat T_l\setminus T_l} \|_2&
= &\sup_{\|\delta\|_0\leq \hat m_l, \|\delta\|_2\leq 1}|
\delta'F'R| \leq
\sup_{\|\delta\|_0\leq \hat m_l, \|\delta\|_2\leq 1}\| \delta'F'\|_2\|R\|_2 \\
&= & \sup_{\|\delta\|_0\leq \hat
m_l, \|\delta\|_2\leq 1}\sqrt{| \delta'F'F\delta|}\|R\|_2
 \leq  n\sqrt{\semax{\hat m_l}} c_s.\\
\end{array}$$
Since $\lambda/c \geq n\|S_l\|_\infty$, and by Lemma \ref{Thm:ReviewModel1}, $\|f_i'(\widehat\beta_l-\beta_{l0})\|_{2,n} \leq \(u + \frac{1}{c}\) \frac{\lambda \sqrt{s}}{n \kappa_{c_0}^l} + 2c_s$ we have
$$ \sqrt{\hat m_l} \leq \frac{2\sqrt{\semax{\hat m_l}}\frac{\|(\widehat \Upsilon_l^{0})^{-1}\|_\infty}{\ell}\[\(u+\frac{1}{c}\)\frac{\sqrt{s}}{\kappa_{c_0}^l} + \frac{3nc_s}{\lambda}\]}{\(1-\frac{1}{c\ell}\)}.$$

The result follows by noting that $(u+[1/c])/(1-1/[\ell c]) = c_0 \ell$ by definition of $c_0$.
\end{proof}

\begin{lemma}[Sub-linearity of maximal sparse eigenvalues]
\label{Lemma:SparseEigenvalueIMP}Let $M$ be a semi-definite positive matrix. For any integer $k \geq 0$ and constant $\ell \geq 1$ we have $ \semax{\ceil{\ell k}}(M) \leq  \lceil \ell \rceil \semax{k}(M).$
\end{lemma}
\begin{proof}
Denote by $\phi_M(k) = \semax{k}(M)$, and let $\bar \alpha$ achieve $\phi_M(\ell k)$. Moreover let
$\sum_{i=1}^{\lceil \ell \rceil} \alpha_i = \bar \alpha$ such that
$\sum_{i=1}^{\lceil \ell \rceil} \|\alpha_i\|_0 = \|\bar
\alpha\|_0$. We can choose $\alpha_i$'s such that $\|\alpha_i\|_0
\leq k$ since $\lceil \ell \rceil k \geq \ell k$. Since $M$ is
positive semi-definite, for any $i, j$ w $\alpha_i'M\alpha_i +
\alpha_j'M\alpha_j \geq 2\left|\alpha_i'M\alpha_j\right|.$
Therefore
$$
\begin{array}{rcl}
\phi_M(\ell k) & = &  \bar \alpha' M \bar\alpha =  \displaystyle \sum_{i=1}^{\lceil \ell \rceil} \alpha_i'M\alpha_i + \sum_{i=1}^{\lceil \ell \rceil} \sum_{j\neq i} \alpha_i'M\alpha_j \leq \sum_{i=1}^{\lceil \ell \rceil} \left\{ \alpha_i'M\alpha_i + ({\lceil \ell \rceil}-1)\alpha_i'M\alpha_i \right\}  \\
& \leq & \displaystyle {\lceil \ell \rceil} \sum_{i=1}^{\lceil
\ell \rceil} \|\alpha_i\|^2 \phi_M(\|\alpha_i\|_0)  \leq {\lceil \ell \rceil} \max_{i=1,\ldots,{\lceil
\ell \rceil}} \phi_M(\|\alpha_i\|_0) \leq {\lceil \ell \rceil}
\phi_M(k)
\end{array}$$
where we used that $\sum_{i=1}^{\lceil \ell \rceil} \|\alpha_i\|^2=1$. \end{proof}

\begin{lemma}[Sparsity bound for Lasso under data-driven penalty]\label{Thm:Sparsity}
Consider the Lasso estimator $\widehat \beta_l = \widehat \beta_{lL}$ with $\lambda/n \geq c\|S_l\|_\infty$, and let $\widehat m_l = |\widehat T_l\setminus T_l|$. Consider the set $$\mathcal{M}=\left\{ m \in \mathbb{N}:
m > s \ 2\semax{m}\|(\widehat \Upsilon_l^{0})^{-1}\|_\infty^2\[\frac{2c_0}{\kappa_{c_0}^l} + \frac{6c_0nc_s}{\lambda\sqrt{s}}\]^2 \right\}.$$ Then,
$$ \hat m_l \leq s \  \left( \min_{m \in \mathcal{M}}\semax{m\wedge n}\right) \  \|(\widehat \Upsilon_l^{0})^{-1}\|_\infty^2\(\frac{2c_0}{\kappa_{c_0}^l} + \frac{6c_0nc_s}{\lambda\sqrt{s}}\)^2.$$
\end{lemma}

\begin{remark}[Sparsity Bound]
Provided that the regularization event $\lambda/n\geq c\|S_l\|_\infty$ occurs,
Lemma \ref{Thm:Sparsity} bounds the number of components $\widehat m_l$ incorrectly selected by Lasso. Essentially, the bound depends on $s$ and on the ratio between the maximum sparse eigenvalues and the restricted eigenvalues. Thus, the empirical Gram matrix can impact the sparsity bound substantially. However, under Condition SE, the ratio mentioned is bounded from above uniformly in $n$.
As expected the bound improves and the regularization event is more likely to occur if a larger value of the penalty parameter $\lambda$ is used.
\end{remark}

\begin{proof}[Proof of Lemma \ref{Thm:Sparsity}]
Rewriting the conclusion in Lemma \ref{Lemma:SparsityLASSO} we have
 \begin{equation}\label{Eq:Sparsity}\hat m_l \leq s \ \semax{\hat m_l}\|(\widehat \Upsilon_l^{0})^{-1}\|_\infty^2\[\frac{2c_0}{\kappa_{c_0}^l} + \frac{6c_0nc_s}{\lambda\sqrt{s}}\]^2.\end{equation}
Note that $\widehat m_l \leq n$ by optimality conditions.  Consider any $M \in \mathcal{M}$, and suppose $\widehat m_l > M$. Therefore by Lemma \ref{Lemma:SparseEigenvalueIMP} on sublinearity of sparse eigenvalues
$$ \hat m_l \leq s \ \ceil{\frac{\hat m_l}{M}}\semax{M}\|(\widehat \Upsilon_l^{0})^{-1}\|_\infty^2\[\frac{2c_0}{\kappa_{c_0}^l} + \frac{6c_0nc_s}{\lambda\sqrt{s}}\]^2.$$ Thus, since $\ceil{k}\leq 2k$ for any $k\geq 1$  we have
$$ M \leq  s \ 2\semax{M}\|(\widehat \Upsilon_l^{0})^{-1}\|_\infty^2\[\frac{2c_0}{\kappa_{c_0}^l} + \frac{6c_0nc_s}{\lambda\sqrt{s}}\]^2$$
which violates the condition that $M \in \mathcal{M}$. Therefore, we have $\widehat m_l \leq M$.

In turn, applying (\ref{Eq:Sparsity}) once more with $\widehat m_l \leq (M\wedge n)$ we obtain
 $$ \hat m_l \leq  s \ \semax{M\wedge n}\|(\widehat \Upsilon_l^{0})^{-1}\|_\infty^2\[\frac{2c_0}{\kappa_{c_0}^l} + \frac{6c_0nc_s}{\lambda\sqrt{s}}\]^2.$$

The result follows by minimizing the bound over $M \in \mathcal{M}$.
\end{proof}

Step 3. Next we combine the previous steps to establish Theorem \ref{Thm:RatesPostLASSO}. As in Step 3 of Appendix \ref{AppendixThm1}, recall that $\max_{1\leq l\leq k_e}\| \widehat \Upsilon_l^0 - \Upsilon_l^0\|_\infty \to_{\mathrm{P}} 0.$ 

Let $\bar k$ be the integer that achieve the minimum in the definition of $\mu^2$. Since $c_s \lesssim_{\mathrm{P}} \sqrt{s/n}$ leads to $nc_s/[\lambda\sqrt{s}] \to_P 0$, we have that $\bar k \in \mathcal{M}$ with high probability as $n\to \infty$.
Moreover,  as long as $\lambda/n \geq c\max_{1\leq l\leq k_e}\|S_l\|_\infty$, $\ell \to_{\mathrm{P}} 1$ and $c>1$, by Lemma \ref{Thm:Sparsity} we have for every $l=1,\ldots,k_e$ that
\begin{equation}\label{SparsityBound} \widehat m_l \lesssim_{\mathrm{P}} s \mu^2 \semin{\bar k+s}/\kappa_{\bar C}^2\lesssim_{\mathrm{P}} s \mu^2 \semin{\hat m_l+s}/\kappa_{\bar C}^2\end{equation}
since  $\bar k \in \mathcal{M}$ implies $\bar k \geq \hat m_l$.

By the choice of $\lambda=2c\sqrt{n} \Phi^{-1}(1-\gamma/(2pk_e))$ in (\ref{the principle lambda}), since $\gamma \to 0$,
the event $\lambda/n \geq c\max_{1\leq l \leq k_e}\|S_l\|_\infty$ holds with probability approaching $1$.
Therefore, by the first and last results in Lemma \ref{Thm:2StepMain} we have
$$\max_{1\leq l\leq k_e} \|D_{il} - f_i'\widehat \beta_{lPL}\|_{2,n} \lesssim_{\mathrm{P}} \frac{\mu}{\kappa_{\bar C}}\sqrt{\frac{s\log p}{n}} + c_s + \max_{1\leq l \leq k_e}\frac{\lambda\sqrt{s}}{n\kappa_{c_0}^l}.$$
Because $\max_{1\leq l\leq k_e}1/\kappa^l_{c_0} \leq \max_{1\leq l\leq k_e} \| \hat \Gamma_l^0\|_\infty/\kappa_{\bar C} \lesssim_P 1/\kappa_{\bar C}$ by Step 1 of Theorem \ref{Thm:RateLASSO}, we have
\begin{equation}\label{BoundPredicNormPL} \max_{1\leq l\leq k_e}\|D_{il} - f_i'\widehat \beta_{lPL}\|_{2,n} \lesssim_{\mathrm{P}} \frac{\mu}{\kappa_{\bar C}}\sqrt{\frac{s\log (k_ep/\gamma)}{n}}\end{equation}
since $k_e \lesssim p$ and $c_s \lesssim_P \sqrt{s/n}$. That establishes the first inequality of Theorem \ref{Thm:RatesPostLASSO}.

To establish the second inequality of Theorem \ref{Thm:RatesPostLASSO}, since
$\|\widehat \beta_{lPL}-\beta_{l0}\|_0 \leq \widehat m_l+s$, we have
$$\|\widehat \beta_{lPL}-\beta_{l0}\|_1\leq
\sqrt{\|\widehat \beta_{lPL}-\beta_{l0}\|_0}\|\widehat \beta_{lPL}-\beta_{l0}\|_2 \leq \sqrt{\widehat m_l+s}
\frac{\|f_i'(\widehat \beta_{lPL}-\beta_{l0})\|_{2,n}}{\sqrt{\semin{\widehat m_l + s}}}.$$
The sparsity bound (\ref{SparsityBound}), the prediction norm bound
(\ref{BoundPredicNormPL}), and the relation $\|D_{il} - f_i'\widehat \beta_{lPL}\|_{2,n} \leq c_s + \|f_i'(\widehat \beta_{lPL}-\beta_{l0})\|_{2,n}$ yield the result with the relation above. \qed

\begin{lemma}[Asymptotic Validity of the Data-Driven Penalty Loadings]\label{Lemma:ValidityDataDriven} Under the
conditions of Theorem \ref{Thm:RateLASSO} and Condition RF or the conditions of Theorem \ref{Thm:RatesPostLASSO} and Condition SE, the
penalty loadings $\hat \Upsilon $ constructed by the K-step
Algorithm A.1  are asymptotically valid. In particular, for $K\geq 2$ we have  $u' =1$.
\end{lemma}

For proof of Lemma \ref{Lemma:ValidityDataDriven} see Online Appendix.

\section{Proofs of Lemmas 1-4}

For proof of Lemma 1,  see \citen{BC-SparseQR}, Supplement.
For proof of Lemma 2, see \citen{BC-PostLASSO}. For proofs of Lemma \ref{lemma: RF} and \ref{lemma: SM}, see Online Appendix.

\section{Proofs of Theorems 3-7.}\label{AppThm4to6}
\subsection{Proof of Theorems \ref{theorem 1} and \ref{theorem 3}.} The proofs are original and they rely on the consistency of the sparsity-based estimators
both with respect to the $L^2(\Pn)$ norm $\|\cdot\|_{2,n}$ and the $\ell_1$-norm $\|\cdot\|_1$.  These proofs
also exploit the use of moderate deviation theory for self-normalized sums.

Step 0. Using data-driven penalty satisfying (\ref{the principle lambda}) and (\ref{Def:AsympValidPenaltyLoading}), we have by Theorem \ref{Thm:RateLASSO} and Condition RE that the Lasso estimator
and by Theorems \ref{Thm:RatesPostLASSO} and Condition SE that the Post-Lasso estimator obey:
\begin{eqnarray}\label{eq: Step 0}
&& \max_{1\leq {l}\leq k_e} \|\widehat D_{i{l}} -D_{i{l}}\|_{2,n}  \lesssim_{\mathrm{P}} \sqrt{\frac{s \log (p\vee n) }{n}} \to 0 \\
&& \sqrt{\log p}  \| \widehat \beta_{{l}} - \beta_{l0}\|_{1}  \lesssim_{\mathrm{P}} \sqrt{\frac{s^2 \log^2 (p\vee n) }{n}} \to 0.\label{eq: Step 01}
\end{eqnarray}
In order to prove Theorem \ref{theorem 1} we need also the condition
$$
\max_{1\leq {l}\leq k_e}\|\widehat D_{i{l}}-D_{i{l}}\|^2_{2,n} n^{2/q_{\epsilon}} \lesssim_{\mathrm{P}} \frac{s \log (p\vee n) }{n} n^{2/q_{\epsilon}}  \to 0,
$$
with the last statement holding by Condition SM.  Note that Theorem \ref{theorem 3} assumes
(\ref{eq: Step 0})-(\ref{eq: Step 01}) as high level conditions.

Step 1.    We have that by $\Ep[\epsilon_i|D_i]=0$
\begin{eqnarray*}
\sqrt{n}(\hat \alpha - \alpha_0) &= & \En [\widehat D_i d_i']^{-1} \sqrt{n} \En [\widehat D_i \epsilon_i]  =  \{\En [\widehat D_i d_i']\}^{-1} (\Gn [D_i \epsilon_i]+o_{\mathrm{P}}(1) )  \\
& = & \{\barEp [D_i d_i'] + o_{\mathrm{P}}(1)\}^{-1} \left( \Gn [D_i \epsilon_i] + o_{\mathrm{P}}(1) \right),
 \end{eqnarray*}
where by Steps 2 and 3 below:
\begin{eqnarray}
 \En [\widehat D_id_i'] = \barEp[D_id_i'] + o_{\mathrm{P}}(1) \label{eq: to show 1} \\
 \sqrt{n}\En [\widehat D_i \epsilon_i] = \Gn [D_i \epsilon_i] + o_{\mathrm{P}}(1)\label{eq: to show 2}
 \end{eqnarray}
where $\barEp [D_id_i'] = \barEp[D_i D_i'] = Q$ is bounded away from zero
and bounded from above in the matrix sense, uniformly in $n$. Moreover,
$\textrm{Var}( \Gn [D_i \epsilon_i]   ) = \Omega$ where $\Omega=\sigma^2 \barEp[D_i D_i']$ under homoscedasticity
and $\Omega= \barEp[\epsilon_i^2 D_i D_i'] $ under heteroscedasticity. In either case
we have that $\Omega$  is bounded away from zero and from above in the matrix sense, uniformly in $n$, by the assumptions
the theorems. (Note that matrices $\Omega$ and $Q$ are implicitly indexed by $n$, but we omit the index
to simplify notations.)
Therefore,
$$
\sqrt{n}(\hat \alpha - \alpha_0) = Q^{-1}  \Gn [D_i \epsilon_i] + o_{\mathrm{P}}(1),
$$
and  $
Z_n = (Q^{-1} \Omega Q^{-1})^{-1/2} \sqrt{n}(\hat \alpha - \alpha_0)  = \Gn[z_{i,n}] + o_{\mathrm{P}}(1),$
where $z_{i,n} = (Q^{-1} \Omega Q^{-1})^{-1/2} Q^{-1} D_i \epsilon_i$ are i.n.i.d. with mean zero and variance
$I$. We have
that for some small enough $\delta>0$
$
\barEp\|z_{i,n}\|_2^{2+\delta} \lesssim \barEp\left[ \| D_i \|_2^{2+\delta} |\epsilon_i|^{2 + \delta} \right]  \lesssim 1,
$
by Condition SM. This condition verifies the Lyapunov condition, and the application of the Lyapunov CLT for i.n.i.d. triangular arrays and the Cramer-Wold device implies that $Z_n \to_d N(0,I)$.

Step 2.   To show (\ref{eq: to show 1}), note that
{\small  \begin{eqnarray*}
\| \En[(\widehat D_i - D_i)d_i ']\| &\leq&  \En[\|\widehat D_i-D_i\|_2 \|d_i\|_2] \leq  \sqrt{\En[\|\widehat D_i-D_i\|_2^2] \En[\|d_i\|_2^2] } \\
& = & \sqrt{\En\[\sum_{{l}=1}^{k_e} |\widehat D_{i{l}}-D_{i{l}}|^2\] \En[\|d_i\|_2^2]} \leq   \sqrt{k_e}\max_{1\leq {l}\leq k_e} \|\widehat D_{i{l}} -D_{i{l}}\|_{2,n}  \sqrt{\En[\|d_i\|_2^2]} \\
& \lesssim_{\mathrm{P}} & \max_{1\leq{l}\leq k_e} \|\widehat D_{i{l}} -D_{i{l}}\|_{2,n} = o_{\mathrm{P}}(1).
 \end{eqnarray*}}
where $ \sqrt{\En[\|d_i\|_2^2]} \lesssim_{\mathrm{P}} 1$ by $\barEp\|d_i\|_2^2 \lesssim 1$ and Chebyshev, and the last assertion holds by Step 0.

Moreover, $\En [D_i D_i'] - \barEp [D_i D_i'] \to_{\mathrm{P}} 0$ by von Bahr-Essen inequality (\citename{vonbahr:esseen}, \citeyear*{vonbahr:esseen})
using that $\barEp[\|D_i\|_2^q] $ for a fixed $q >2$ is bounded uniformly in $n$ by Condition SM.

Step 3.   To show (\ref{eq: to show 2}), let $a_{il}:= a_l(x_i)$, note that $\Ep[f_{ij}\epsilon_i]=0$, $\Ep[\epsilon_i|D_{il}]=0$ and $\Ep[\epsilon_{il}|a_{il}]=0$, and
 \begin{eqnarray*}
 &  &  \max_{1\leq {l}\leq k_e}  |\sqrt{n}\En[(\widehat D_{il} - D_{il}) \epsilon_i]|  \\
 & & =  \max_{1\leq {l}\leq k_e} | \sqrt{n}\En \{f_i'(\widehat \beta_{l} - \beta_{{l0}})  \epsilon_i \}  - \Gn \{ a_{i{l}} \epsilon_i \} |  =   \max_{1\leq {l}\leq k_e} | \sum_{j=1}^p \Gn \{ f_{ij} \epsilon_i  \} '  (\widehat \beta_{{l} j} - \beta_{{l0} j})   - \Gn \{ a_{i{l}} \epsilon_i \} |  \\
 & & \leq
 \max_{1 \leq j \leq p}  \left |\frac{\Gn [ f_{ij}\epsilon_i]}
 {\sqrt{\En [f_{ij}^2 \epsilon_i^2]}}  \right |   \max_{1 \leq j \leq p} \sqrt{\En [f_{ij}^2 \epsilon_i^2]}
  \max_{1\leq {l}\leq k_e} \|\widehat \beta_{{l}} - \beta_{{l0}}\|_{1}   +   \max_{1 \leq {l} \leq k_e} |\Gn \{ a_{i{l}} \epsilon_i\} |.
 \end{eqnarray*}

Next we note that for each $l=1,\ldots,k_e$
$
 |\Gn \{ a_{i{l}} \epsilon_i \} |  \lesssim_{\mathrm{P}} [\En a^2_{i{l}}]^{1/2} \lesssim_{\mathrm{P}}  \sqrt{s/n} \to 0,
$
by  the Condition AS on $[\En a^2_{i{l}}]^{1/2}$ and by Chebyshev inequality, since in the homoscedastic case of Theorem \ref{theorem 1}:
$
\text{Var}\[\Gn \{ a_{i{l}} \epsilon_i \}| x_1,...,x_n \] \leq  \sigma^2 \En a^2_{i{l}},
$
and in the bounded heteroscedastic case of Theorem \ref{theorem 1}:
$
\text{Var}\[\Gn \{ a_{i{l}} \epsilon_i  \}| x_1,...,x_n \] \lesssim  \En a^2_{i{l}}.
$
Next we can bound
$
\max_{1 \leq j \leq p}  \left |{\Gn [ f_{ij}\epsilon_i]}/
 {\sqrt{\En [f_{ij}^2 \epsilon_i^2]}}  \right |  \lesssim_{\mathrm{P}} \sqrt{ \log p}
$
provided that $p$ obeys the growth condition $\log p = o(n^{1/3})$, and
\begin{equation}\label{eq: no bozos}
\min_{1 \leq j \leq p} M_{j0} := \frac{\barEp[f_{ij}^2 \epsilon_i^2]^{1/2}}{\barEp[|f_{ij}|^3 |\epsilon_i|^3]^{1/3}} \gtrsim 1.
\end{equation}
  This result follows by the bound on moderate deviations of a maximum of a self-normalized vector stated in Lemma \ref{Lemma:SNMD},  and by (\ref{eq: no bozos}) holding by Condition SM.
Finally,
$
 \max_{1 \leq j \leq p} \En [f_{ij}^2 \epsilon_i^2]  \lesssim_{\mathrm{P}} 1,
$
by Condition SM. Thus, combining bounds above with bounds in (\ref{eq: Step 0})-(\ref{eq: Step 01})
$$
\max_{1\leq {l}\leq k_e}  |\sqrt{n}\En[(\widehat D_{il} - D_{il}) \epsilon_i]|  \lesssim_{\mathrm{P}}   \sqrt{ \frac{s^2 \log^2 (p\vee n) }{n }}   +  \sqrt{\frac{s}{n}}   \to 0,
$$
where the conclusion by Condition SM (iii).

Step 4. This step establishes
consistency of the variance estimator in the homoscedastic case
of Theorem \ref{theorem 1}.

Since $\sigma^2$ and  $Q=\barEp [D_i D_i']$ are bounded away from zero
and from above uniformly in $n$, it suffices to show $\hat \sigma^2 - \sigma^2 \to_{\mathrm{P}} 0$
and  $\En [\widehat D_i\widehat D_i'] - \barEp [D_i D_i'] \to_{\mathrm{P}} 0$. Indeed,
$\hat\sigma^2 = \En[(\epsilon_i-d_i'(\widehat\alpha-\alpha_0))^2] =
\En[\epsilon_i^2]+2\En[\epsilon_id_i'(\alpha_0-\widehat\alpha)]
+\En[(d_i'(\alpha_0-\widehat\alpha))^2]$ so that
$\En[\epsilon_i^2]- \sigma^2 \to_{\mathrm{P}} 0$ by Chebyshev inequality since $\barEp[|\epsilon_i|^4]$ is bounded
uniformly in $n$, and the remaining
terms converge to zero in probability since $\widehat \alpha - \alpha_0 \to_{\mathrm{P}}
0$ by Step 3, $\|\En[d_i\epsilon_i]\|_2 \lesssim_{\mathrm{P}} 1$ by Markov and since
$\barEp\|d_i\epsilon_i\|_2 \leq \sqrt{\barEp\|d_i\|_2^2} \sqrt{\barEp|\epsilon_i|^2}$
is uniformly bounded in $n$ by Condition SM, and $ \En \|d_i\|_2^2 \lesssim_{\mathrm{P}} 1$
by Markov and $\barEp \|d_i\|_2^2$ bounded uniformly in $n$ by Condition SM.   Next, note that
$$ \|\En [\widehat D_i\widehat D_i'] - \En [D_i D_i'] \| = \|\En
[D_i (\widehat D_i-D_i)' + (\widehat D_i-D_i)D_i'] + \En
[(\widehat D_i-D_i) (\hat D_i-D_i)'] \|
$$
which is bounded up to a constant by
$$
\sqrt{k_e}\max_{1\leq {l}\leq k_e}\|\widehat D_{i{l}}-D_{i{l}}\|_{2,n} \| \ \|D_{i}\|_2\|_{2,n} + k_e\max_{1\leq {l}\leq k_e}\|\widehat D_{i{l}}-D_{i{l}}\|^2_{2,n} \to_{\mathrm{P}} 0
$$
by (\ref{eq: Step 0}) and by $\|\ \|D_i\|_2 \|_{2,n} \lesssim_{\mathrm{P}} 1$ holding by Markov inequality.
Moreover, $\En [D_i D_i'] - \barEp [D_i D_i'] \to_{\mathrm{P}} 0$ by Step 2.  

Step 5.   This step establishes
consistency of the variance estimator in the boundedly heteroscedastic case
of Theorem \ref{theorem 1}.

 Recall that
$\hat \Omega := \En [\hat \epsilon_i^2 \widehat D(\x_i) \widehat D(\x_i)'] $ and
$\Omega$ $:=$ $\barEp [ \epsilon_i^2  D(\x_i)  D(\x_i)'] $, where the latter
is bounded away from zero and from above uniformly in $n$.   Also, $Q=\barEp [D_i D_i']$ is bounded away from zero
and from above uniformly in $n$. Therefore, it suffices to show $\hat \Omega - \Omega \to_{\mathrm{P}} 0$
and  that $\En [\widehat D_i\widehat D_i'] - \barEp [D_i D_i'] \to_{\mathrm{P}} 0$. The latter has been shown in the previous step, and we only need to show the former.

In what follows, we shall repeatedly use the following elementary inequality: for arbitrary non-negative random variables $W_1,...,W_n$
and $q>1$:
\begin{equation} \label{eq: bound max}
\max_{i \leq n} W_i \lesssim n^{1/q} \text{ if }  \barEp[W_i^q] \lesssim 1,
\end{equation}
which follows by Markov inequality from $
\Ep[\max_{i \leq n} W_i ] \leq n^{1/q} \Ep \( \frac{1}{n} \sum_{i=1}^n W_i^q\)^{1/q} \leq  n^{1/q} (\barEp[W_i^q])^{1/q},$
which follows from the trivial bound $\max_{i \leq n} |w_i| \leq \sum_{i=1}^n |w_i|$ and Jensen's inequality.

First, we note
\begin{eqnarray*}
&  & \| \En [ (\hat \epsilon_i^2 - \epsilon_i^2) \hat D_i \hat D_i']\|  \leq   \| \En[ \{ d_i'(\hat \alpha - \alpha_0)\}^2 \hat D_i \hat D_i'] \|  + 2 \| \En[ \epsilon_i d_i'(\hat \alpha - \alpha_0) \hat D_i \hat D_i'] \|   \\ &  & \lesssim_{\mathrm{P}}   \max_{i \leq n} \|d_i\|_2^2 n^{-1} \| \En [\hat D_i \hat D_i'] \| +  \max_{i \leq n } |\epsilon_i| \|d_i\|_2 n^{-1/2}  \|  \En [\hat D_i \hat D_i'] \| \to_{\mathrm{P}} 0,
\end{eqnarray*}
since  $\|\hat \alpha - \alpha_0\|_2^2 \lesssim_P 1/n$, $\| \En \hat D_i \hat D_i' \| \lesssim_{\mathrm{P}} 1$  by Step 4, and
$\max_{i \leq n } \|d_i\|_2^2 n^{-1} \to_{\mathrm{P}} 0$ (by
$ \max_{i \leq n } \|d_i\|_2 \lesssim_P n^{1/q}$ for $q> 2$,  holding by  $\barEp[\|d_i\|_2^  q] \lesssim 1$ and inequality (\ref{eq: bound max})) and  $\max_{i \leq n} [\|d_i \|_2 |\epsilon_i|] n^{-1/2} \to_{\mathrm{P}} 0$
(by
$ \max_{i \leq n} [\|d_i \|_2 |\epsilon_i|] \lesssim_P n^{1/q}
$ for $q>2$ holding by $\barEp [  (\|d_i\|_2 |\epsilon_i|^2)^{q}] \lesssim 1$ and inequality (\ref{eq: bound max})).

 Next we note that
$$ \|\En [ \epsilon_i^2 \widehat D_i\widehat D_i'] - \En [  \epsilon_i^2 D_i D_i'] \| = \|  \En
[\epsilon_i^2 D_i (\widehat D_i-D_i)' +   \epsilon_i^2 (\widehat D_i-D_i)D_i'] + \En
[ \epsilon_i^2 (\widehat D_i-D_i) (\hat D_i-D_i)'] \|
$$
which is bounded up to a constant by
$$
\sqrt{k_e}\max_{1\leq {l}\leq k_e}\|\widehat D_{i{l}}-D_{i{l}}\|_{2,n} \|\epsilon_i^2 \|D_i\|_2\|_{2,n} + k_e\max_{1\leq {l}\leq k_e}\|\widehat D_{i{l}}-D_{i{l}}\|^2_{2,n} \max_{i \leq n} \epsilon_i^2\to_{\mathrm{P}} 0.
$$
The latter occurs because $\|\epsilon_i^2 \|D_i\|_2\|_{2,n} = \sqrt{ \En [\epsilon_i^4 \|D_i\|_2^2 ]} \lesssim_{\mathrm{P}} 1$ by  $\barEp [\epsilon_i^4 \|D_i\|_2^2 ]$ uniformly bounded in $n$ by Condition SM and by Markov inequality, and
$$
\max_{1\leq {l}\leq k_e}\|\widehat D_{i{l}}-D_{i{l}}\|^2_{2,n} \max_{i \leq n} \epsilon_i^2 \lesssim_{\mathrm{P}} \frac{s \log (p\vee n)  }{n} n^{2/q_{\epsilon}}\to 0,
$$
where the latter step holds by Step 0 and by $\max_{i \leq n} \epsilon_i^2 \lesssim_{\mathrm{P}} n^{2/q_{\epsilon}}$ holding
by $
\barEp[\epsilon_i^{q_{\epsilon}}] \lesssim 1$ and inequality (\ref{eq: bound max}).  Finally, $\En [\epsilon_i^2 D_i D_i'] - \barEp [\epsilon_i^2 D_i D_i'] \to_{\mathrm{P}} 0$ by the von Bahr-Essen inequality (\citename{vonbahr:esseen}, \citeyear*{vonbahr:esseen}) and by $\barEp [|\epsilon_i|^{2 + \mu} \|D_i\|_2^{2 + \mu}]$ bounded uniformly in $n$ for small enough $\mu>0$ by Condition SM.

We conclude that $\En [\hat \epsilon_i^2 \widehat D_i\widehat D_i'] -  \barEp [\epsilon_i^2 D_i D_i'] \to_{\mathrm{P}} 0$. \qed

\subsection{Proof of Theorem \ref{Thm:WIV}}

Step 1. To establish claim (1), using the properties of projection we note that
\begin{equation}\label{note a}
n \En[\tilde \epsilon_i \tilde f_{ij}] = n \En[\epsilon_i \tilde f_{ij}].
 \end{equation}
Since for $\hat \mu_\epsilon = (\En[w_i w_i'])^{-1} \En[w_i \epsilon_i] $  we have $\| \hat \mu_\epsilon \|_2 \leq \| \En[w_i w_i']^{-1} \| \|\En[w_i \epsilon_i]  \|_2 $, where $\| \En[w_i w_i']^{-1} \|$  is bounded by SM2(ii) and $ \|\En[w_i \epsilon_i]  \|_2$  is of stochastic order $\sqrt{k_w/n}$ by Chebyshev inequality and SM2(ii). Hence
$\|\hat \mu_\epsilon \|_2 \lesssim_\Pr  \sqrt{k_w/n}$. Since
$\|w_i\|_2 \leq \zeta_w $ by Condition SM2(i),  we conclude that $\max_{i \leq n } |w_i' \hat \mu_\epsilon| \lesssim_\Pr \zeta_w \sqrt{k_w}/\sqrt{n} \to 0$.  Hence, uniformly in $j \in \{1,...,p\}$,
\begin{equation}\label{note b}
\big |\sqrt{\En [\tilde \epsilon_i^2 \tilde f^2_{ij}]} -\sqrt{\En [\epsilon_i^2 \tilde f^2_{ij}]} \big | \overset{(a)}{\leq} \sqrt{ \En{ [(w_i'\hat \mu_\epsilon)^2 \tilde f^2_{ij} ]} } \overset{(b)}{=}  o_P(1) \sqrt{\En [\tilde f^2_{ij}]}  \overset{(c)}{=} o_P(1),
 \end{equation}
where (a) is by the triangular inequality and the decomposition $\tilde \epsilon_i = \epsilon_i - w_i' \hat \mu_\epsilon$,  (b) is by
the Holder inequality, and (c) is by the normalization $\sqrt{\En [\tilde f^2_{ij}]}=1$ for each $j$.  Hence, for $c > 1$, by (\ref{note a}) and (\ref{note b})  wp $\to 1$
$$
\Lambda_{\alpha_1} \leq  c \bar \Lambda_{\alpha_1},  \ \
  \bar \Lambda_{\alpha_1} := \max_{1 \leq j \leq p}  n|\En[ \epsilon_i \tilde f_{ij}] |/\sqrt{\En[\epsilon^2_i \tilde f^2_{ij}]}.
$$
Since $\bar \Lambda_{\alpha_1} $ is a maximum of self-normalized sum of i.n.i.d. terms conditional on $X$,  application of SM2(iii)-(iv) and the moderate deviation bound from Lemma \ref{Lemma:SNMD}  for the self-normalized sum with $U_{ij} = \epsilon_i \tilde f_{ij} $, conditional on $X$,  implies that  $\Pr( c\bar \Lambda_{\alpha_1}  \leq \Lambda(1-\gamma) ) \geq 1- \gamma -o(1)$. This verifies claim (i).

Step 2. To show claim (2) we  note that using triangular and other elementary inequalities:
{\small $$\Lambda_a   =   \max_{1 \leq j \leq p} \left | \frac{ n | \En[ (\tilde \epsilon_i - (a-\alpha_1)' \tilde d_{ei}) \tilde f_{ij}] }{ \sqrt{
\En[(\tilde \epsilon_i - (a-\alpha_1) '\tilde d_{ei})^2 \tilde f^2_{ij}] }} \right|
 \geq
\max_{1 \leq j \leq p} \left | \frac{  n | \En[(a-\alpha_1)' \tilde d_{ei}\tilde f_{ij}]| }{ \sqrt{
\En[\tilde \epsilon^2_i\tilde f^2_{ij}]} + \sqrt{
\En[ \{(a-\alpha_1)' \tilde d_{ei}\}^2 \tilde f^2_{ij}] }}   \right| - \Lambda_{\alpha_1}.
$$}
The first term on the right side is bounded below by, wp $\to$ 1,
$$
\max_{1 \leq j \leq p} \frac{n|\En[(a-\alpha_1)' \tilde d_{ei}\tilde f_{ij}]|}{ c\sqrt{\En[\epsilon^2_i \tilde f_{ij}^2]} + \sqrt{
\En[ \{(a-\alpha_1)' \tilde d_{ei}\}^2 \tilde f^2_{ij}] }},
$$
by Step 1 for some $c>1$,  and $\Lambda_{\alpha_1} \lesssim_P \sqrt{ n \log (p/\gamma)}$ also by Step 1.  Hence
for any constant $C$, by the last condition in the statement of the theorem, with probability converging to 1, $\Lambda_a - C\sqrt{n \log (p/\gamma)} \to + \infty,$
so that Claim (2) immediately follows, since
 $\Lambda(1-\gamma)  \lesssim \sqrt{n \log (p/\gamma)}$. \qed

\subsection{Proof of Theorems \ref{theorem: spec} and \ref{Thm:InferenceSplitSampleIV}.} See Online Appendix.

\bibliographystyle{econometrica}

\bibliography{mybib}

\pagebreak

\begin{figure}
    \includegraphics[width=\textwidth]{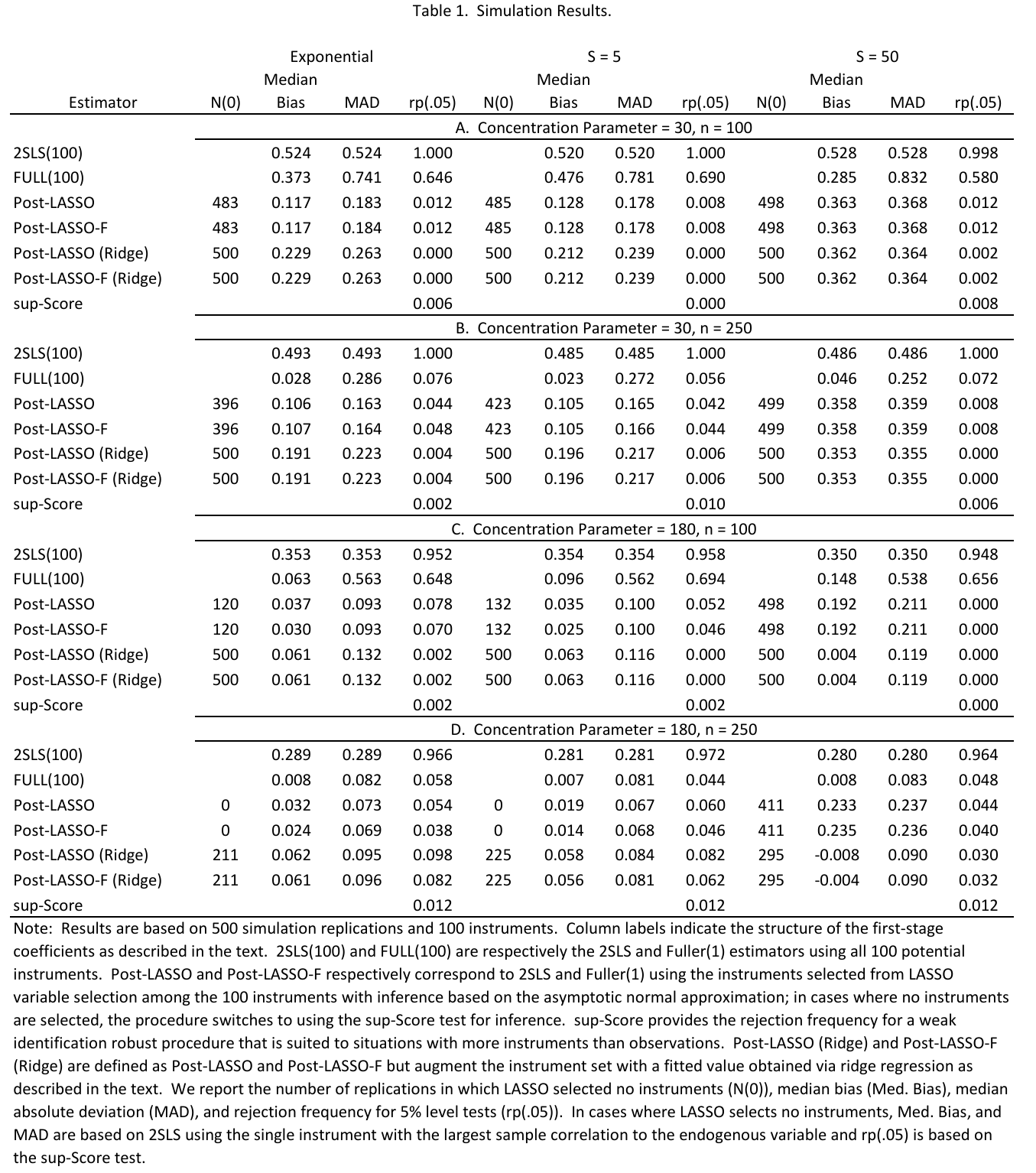}
    \label{fig:table1}
\end{figure}

\pagebreak

\begin{figure}
		\includegraphics[width=\textwidth]{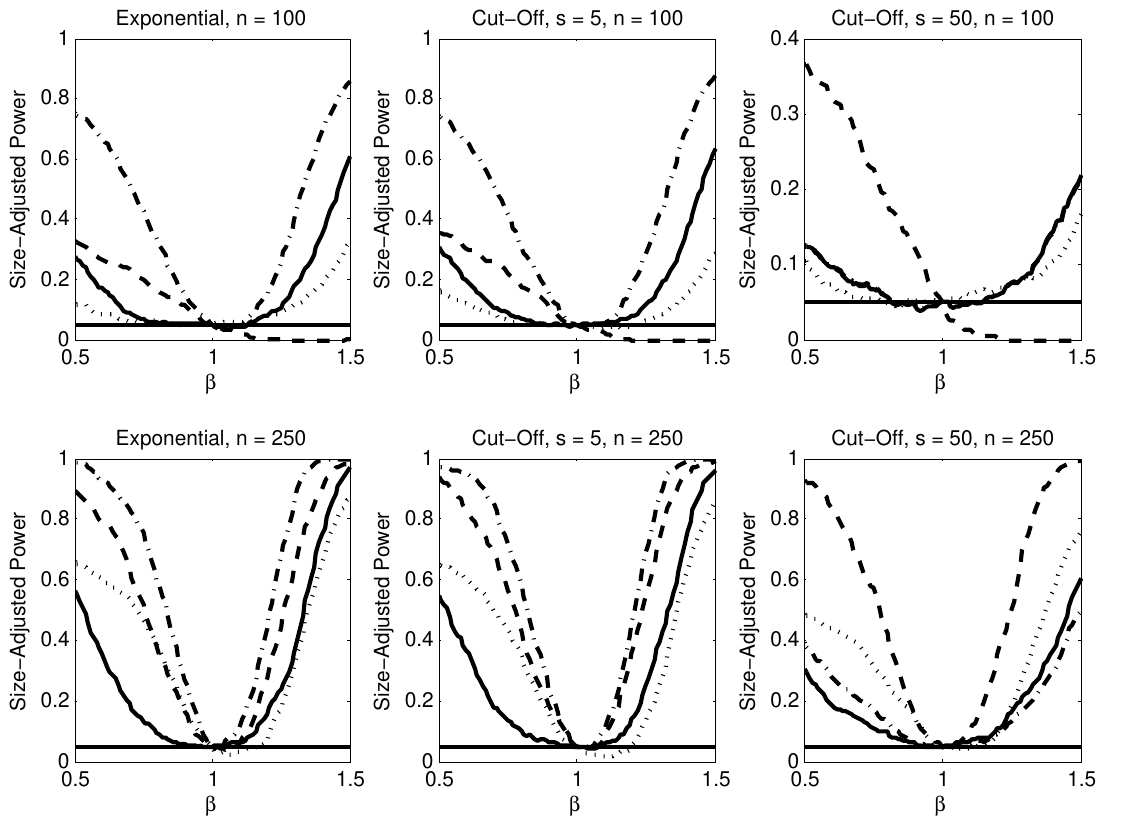}
		\label{fig:figure1}
		\caption{Size-adjusted power curves for Post-Lasso-F (dot-dash), Post-Lasso-F (Ridge) (dotted), FULL(100) (dashed), and sup-Score (solid) from the simulation example with concentration parameter of 180 for $n = 100$ and $n = 250$.}
\end{figure}

\pagebreak

%

\begin{figure}
    \includegraphics[width=\textwidth]{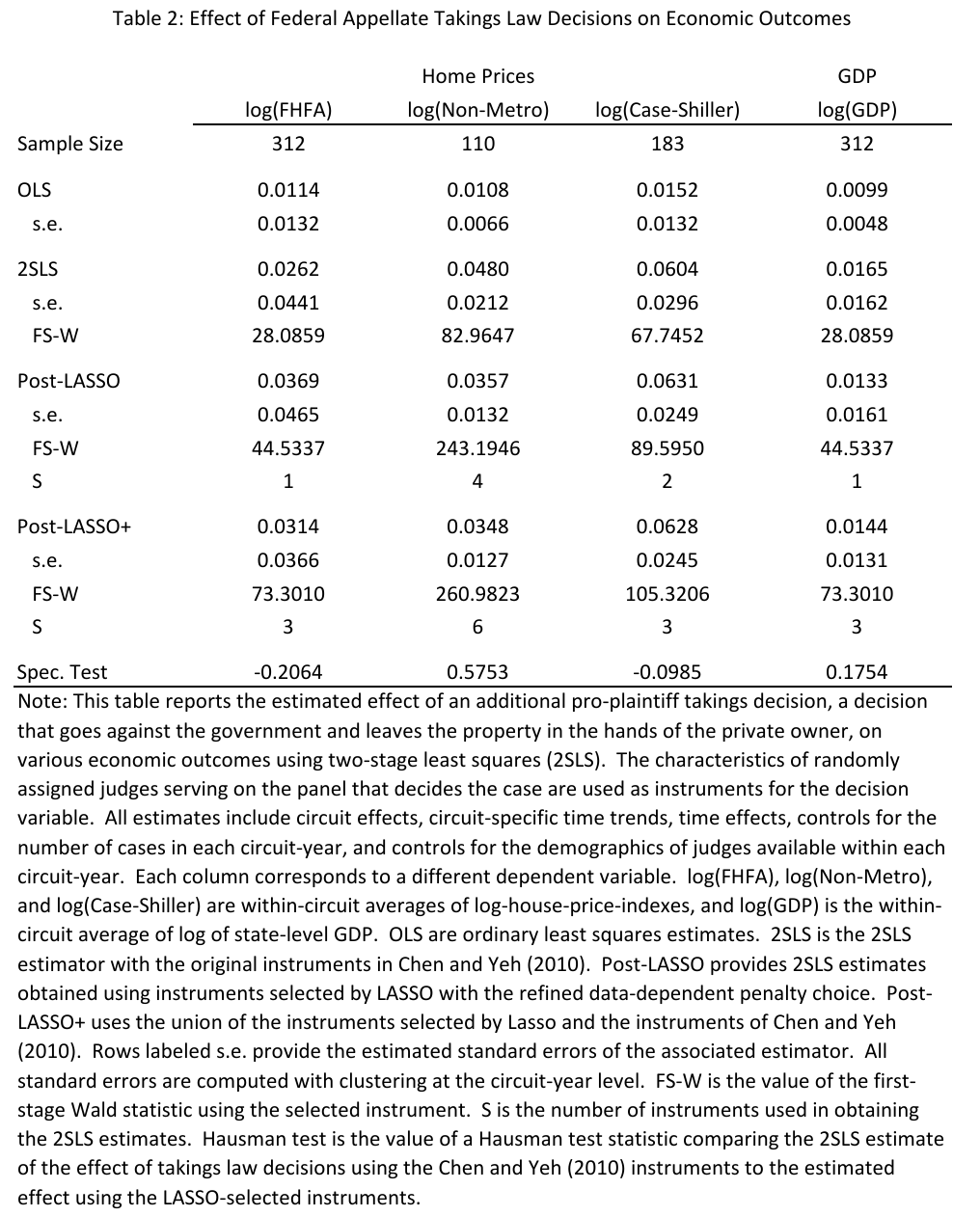}
    \label{fig:table3}
\end{figure}

\clearpage


\setcounter{page}{1}

\begin{center}
{\LARGE Online Appendix for ``Sparse Models and Methods for Optimal Instruments with an Application to Eminent Domain"}
\end{center}

\setcounter{section}{0}
\setcounter{lemma}{0}
\setcounter{theorem}{0}
\setcounter{corollary}{0}

\section{Tools}

\subsection{Lyapunov CLT, Rosenthal Inequality, and Von Bahr-Esseen Inequality}

\begin{lemma}[Lyapunov CLT] Let $\{X_{i,n}, i=1,...,n\}$ be independent zero-mean random variables with variance $s^2_{i,n}$,  $n=1,2,...$. Define $s_n^2 = \sum_{i=1}^n s^2_{i,n}$. If for some $\mu>0$, the Lyapunov's condition holds:
$$
     \lim_{n\to\infty} \frac{1}{s_{n}^{2+\mu}} \sum_{i=1}^{n} \operatorname{E}\big[\,|X_{i,n}|^{2+\mu}\,\big] = 0
$$
then as $n$ goes to infinity:
$$
    \frac{1}{s_n} \sum_{i=1}^{n} X_{i,n} \ \to_d \ \mathcal{N}(0,\;1).
$$
\end{lemma}

\begin{lemma}[Rosenthal Inequality] Let $X_1,...,X_n$ be independent zero-mean random variables, then for $r\geq 2$
$$
\Ep\left(  \left\vert \sum_{i=1}^{n}X_{i}\right\vert ^{r}\right) \leq C(r)  \max\left[ \sum_{t=1}^n \Ep(|X_i|^r),   \left\{\sum_{t=1}^n \Ep(X_i^2)\right\}^{r/2} \right].
$$
\end{lemma}

This is due to \citen{Rosenthal1970}.

\begin{corollary} Let $r \geq 2$, and consider the case of  independent zero-mean variables $X_i$ with $\Ep \En (X_i^2)=1$ and $\Ep\En(|X_i|^r)$ bounded by $C$.
Then for any $\ell_n \to \infty$
$$
Pr \left(\frac{ |\sum_{i=1}^n X_i|}{n} >  \ell_n n^{-1/2} \right) \leq  \frac{2C(r)C}{\ell_n^{r}} \to 0.
$$
\end{corollary}

To verify the corollary, we use Rosenthal's inequality $E\left(  \left\vert \sum_{i=1}^{n}X_{i}\right\vert ^{r}\right) \leq C n^{r/2}$,
and the result follows by Markov inequality, $$ P\left(\frac{|\sum_{i=1}^n X_i|}{n} >c\right) \leq  \frac{ C(r) C n^{r/2} }{c^rn^r} \leq  \frac{C(r) C}{ c^r n^{r/2} }.$$


\begin{lemma}[von Bahr-Essen Inequality] Let $X_1,...,X_n$ be independent zero-mean random variables. Then for $1 \leq r \leq 2$
$$
\Ep \left( \left|  \sum_{i=1}^n X_i\right|^r\right) \leq  (2- n^{-1}) \cdot \sum_{k=1}^n \Ep(|X_k|^r).
$$
\end{lemma}

This result is due to  \citen{vonbahr:esseen}.



\begin{corollary}  Let $r \in [1,2]$, and consider the case of independent zero-mean variables $X_i$ with $\Ep \En (|X_i|^r)$ bounded by $C$.
Then for any $\ell_n \to \infty$
$$
\Pr \left\{\frac{ \left|\sum_{i=1}^n X_i\right|}{n} >  \ell_n n^{-(1-1/r)} \right\}  \leq \frac{2C}{\ell_n^r} \to 0.
$$
\end{corollary}

The corollary follow by Markov and von Bahr-Esseen's inequalities,
$$
\Pr\left(\frac{|\sum_{i=1}^n X_i|}{n} >c\right) \leq  \frac{ C \Ep \left(|  \sum_{i=1}^n X_i|^r\right) }{c^rn^r} \leq  \frac{ \sum_{i=1}^n \Ep(|X_i|^r)}{ c^r n^r} \leq  C\frac{\barEp(|X_i|^r)}{ c^r n^{r-1}}.
$$

\subsection{A Symmetrization-based Probability Inequality}

Next we proceed to use symmetrization arguments to bound the empirical process. Let $\|f\|_{\mathbb{P}_n,2} = \sqrt{\En[ f(Z_i)^2 ]}$, $\mathbb{G}_n(f) = \sqrt{n}\En[f(Z_i) - \Ep[Z_i]]$, and for a random variable $Z$ let $q(Z,1-\tau)$ denote its $(1-\tau)$-quantile.
The proof follows standard symmetrization arguments.

\begin{lemma}[Maximal inequality via symmetrization]\label{Thm:masterSym}
Let $Z_1,\ldots, Z_n$ be arbitrary independent stochastic processes and $\mathcal{F}$ a finite set of measurable  functions. For any $\tau \in (0,1/2)$, and $\delta\in (0,1)$ with probability at least $1-4\tau-4\delta$ we have
$$
\sup_{f \in \mathcal{F}} | \mathbb{G}_n(f)| \leq   4\sqrt{2\log(2|\mathcal{F}|/\delta)}q(\sup_{f\in
\mathcal{F}}\|f\|_{\mathbb{P}_n,2},1-\tau) \vee 2 \sup_{f\in \mathcal{F}}q(|\Gn(f)|,1/2).
$$
\end{lemma}
\begin{proof}
Let
$$e_{1n} = \sqrt{2\log(2|\mathcal{F}|/\delta)} \  q\left(\max_{f\in\mathcal{F}}\sqrt{\En[ f(Z_i)^2 ]},1-\tau\right), \ \ \ e_{2n} = \max_{f\in\mathcal{F}} q\left(| \mathbb{G}_n(f(Z_i))|, \frac{1}{2}\right)$$
and the event $\mathcal{E}= \{\max_{f\in
\mathcal{F}}\sqrt{\En\[f^2(Z_i)\]} \leq q\left(\max_{f\in\mathcal{F}}\sqrt{\En[ f^2(Z_i) ]},1-\tau\right)\}$ which satisfies $P(\mathcal{E}) \geq 1-\tau$. By the symmetrization
Lemma 2.3.7 of \citen{vdV-W} (by definition of $e_{2n}$ we have $\beta_n(x)\geq 1/2$ in Lemma 2.3.7) we obtain
$$
\begin{array}{rl}
\mathbb{P}\left\{\max_{f \in \mathcal{F}} |\mathbb{G}_n(f(Z_i))|> 4 e_{1n} \vee 2e_{2n}  \right\} & \leq 4
\mathbb{P}\left\{\max_{f \in \mathcal{F}} | \mathbb{G}_n(\varepsilon_if(Z_i))| > e_{1n}  \right\}\\
 & \leq 4\mathbb{P}\left\{\max_{f \in \mathcal{F}} | \mathbb{G}_n(\varepsilon_i f(Z_i))| > e_{1n}  | \mathcal{E} \right\} + 4\tau
\end{array}$$ where $\varepsilon_i$ are independent Rademacher random variables, $P(\varepsilon_i=1) = P(\varepsilon_i=-1) = 1/2$.

Thus a union bound yields
\begin{equation}\label{Eq:afterUB} \mathbb{P}\left\{\max_{f \in \mathcal{F}}
|\mathbb{G}_n(f(Z_i))|> 4e_{1n}\vee 2e_{2n}\right\} \leq 4\tau +
4|\mathcal{F}| \max_{f \in \mathcal{F}}
\mathbb{P}\left\{ | \mathbb{G}_n(\varepsilon_if(Z_i))| > e_{1n}  | \mathcal{E}\right\} .
\end{equation}
 We then condition on the values of $Z_1,\ldots,Z_n$ and $\mathcal{E}$, denoting the conditional
 probability measure as $\mathbb{P}_{\varepsilon}$.
 Conditional on $Z_1,\ldots,Z_n$, by the Hoeffding inequality the symmetrized
 process $\mathbb{G}_n(\varepsilon_if(Z_i))$ is sub-Gaussian for the $L_2(\mathbb{P}_n)$ norm,
 namely, for $f \in \mathcal{F}$,
$ \mathbb{P}_{\varepsilon}\{|\mathbb{G}_n(\varepsilon_if(Z_i))| >x \} \leq 2 \exp(
- x^2/\{2\En[f^2(Z_i)]\}).$ Hence, under the event $\mathcal{E}$, we can bound
$$
\begin{array}{rcl}
\mathbb{P}_{\varepsilon}\left\{
|\mathbb{G}_n(\varepsilon_if(Z_i))| > e_{1n} | Z_1,\ldots,Z_n, \mathcal{E}\right\}
& \leq &  2\exp(-e_{1n}^2/[2\En[f^2(Z_i)]) \leq  2\exp(-\log (2|\mathcal{F}|/\delta)).\\
\end{array}
$$
Taking the expectation over $Z_1,\ldots,Z_n$ does not affect the
right hand side bound. Plugging in this bound  yields the result.
\end{proof}

\section{Proof of Theorem 6.}

\begin{proof}[Proof of Theorem 6]  To show part 1, note that by a standard argument
$$
\sqrt{n}(\tilde \alpha - \alpha)= M^{-1} \Gn[A_i \epsilon_i] + o_{\Pr}(1).
$$
From the proof of Theorem 4 we have
that
$$
\sqrt{n}(\widehat \alpha - \alpha_a)= Q^{-1} \Gn[D(x_i) \epsilon_i] + o_{\Pr}(1).
$$
The conclusion follows.  The consistency of $\widehat \Sigma$ for $\Sigma$
can be demonstrated similarly to the proof of consistency of $\widehat \Omega$ and $\widehat Q$
in the proof of Theorem 4.

To show part 2, let $\alpha$ denote the true value as before, which by assumption
coincides with the estimand of the baseline IV estimator by the standard argument,
$ \tilde \alpha - \alpha  = o_P(1). $  The baseline IV estimator is consistent for this quantity.
Under the alternative hypothesis, the estimand of $\hat \alpha$ is
$$
\alpha_{a} = \barEp[D(x_i)d_i']^{-1} \barEp[D(x_i) y_i] = \alpha + \barEp[D(x_i) D(x_i)']^{-1} \barEp[D(x_i) \epsilon_i].
$$
Under the alternative hypothesis, $\|\barEp[D(x_i) \epsilon_i] \|_2$ is bounded away from zero uniformly in $n$. Hence, since the eigenvalues of $Q$ are bounded away from zero uniformly in $n$,
$\|\alpha - \alpha_a\|_2$ is also bounded away from zero uniformly in $n$. Thus, it remains to show that
$\hat \alpha$ is consistent for $\alpha_a$.  We have that
$$
\widehat \alpha - \alpha_{a} =  \En[\hat D(x_i)d_i']^{-1} \En[\hat D(x_i) \epsilon_i ] - \barEp[D(x_i) D(x_i)']^{-1} \barEp[D(x_i) \epsilon_i]
$$
so that
\begin{align}
\|\widehat \alpha - \alpha_{a}\|_2  & \leq  \| \En[\hat D(x_i)d_i']^{-1} -\barEp[D(x_i) D(x_i)']^{-1}\| \|\En[\hat D(x_i) \epsilon_i ]\|_2 +\\
 & + \|\barEp[D(x_i) D(x_i)]^{-1}\| \|\En[\hat D(x_i) \epsilon_i ]- \barEp[D(x_i) \epsilon_i]\|_2 = o_P(1),
 \end{align}
provided that (i)
$\|\En[ \hat D(x_i)d_i']^{-1} - \barEp[D(x_i) D(x_i)']^{-1}\| = o_{P}(1),$
which is shown in the proof of Theorem 4; (ii) $\|\barEp[D(x_i) D(x_i)]^{-1}\|= \| Q^{-1}\|$ is bounded from above uniformly in $n$, which follows from the assumption on $Q$
in Theorem 4;  and (iii)
$$
\|\En[\hat D(x_i) \epsilon_i ]-\barEp[D(x_i) \epsilon_i]\|_2  = o_{P}(1), \ \  \| \barEp[D(x_i) \epsilon_i] \|_2 = O(1),
$$
where  $\| \barEp[D(x_i) \epsilon_i] \|_2 = O(1)$ is assumed.
To show the last claim note that
$$\begin{array}{rl}
\|\En[\hat D(x_i) \epsilon_i ]- \barEp[D(x_i) \epsilon_i]\|_2 & \leq \|\En[\hat D(x_i) \epsilon_i ] - \En[D(x_i)\epsilon_i] \|_2 + \|\En[D(x_i)\epsilon_i]-\barEp[D(x_i) \epsilon_i]\|_2\\
& \leq \sqrt{k_e}\max_{1\leq l\leq k_e}\| D_l(x_i) - \hat D_l(x_i) \|_{2, n} \| \epsilon_i\|_{2,n} + o_P(1)= o_P(1),
\end{array}$$
since $k_e$ is fixed, $\|\En[D(x_i)\epsilon_i]-\barEp[D(x_i) \epsilon_i]\|_2=o_P(1)$ by von Bahr-Essen inequality \citen{vonbahr:esseen} and SM, $\| \epsilon_i\|_{2,n} = O_P(1)$ follows by the Markov inequality and assumptions on the moments of $\epsilon_i$, and
$\max_{1\leq l \leq k_e}\| D_l(x_i) - \hat D_l(x_i)\|_{2,n} = o_P(1)$ follows from Theorems 1 and 2. \end{proof}

\section{Proof of Theorem 7}

We introduce additional superscripts $a$ and $b$ on all variables; and $n$ gets replaced by either $n_a$ or $n_b$. The proof is divided in steps. Assuming
\begin{eqnarray}\label{RateConditionSplitIV}
&&  \max_{1\leq l\leq k_e}\|\widehat D^k_{il} - D^k_{il}\|_{2,n_k} \lesssim_P \sqrt{\frac{s\log (p\vee n)}{n}} = o_P(1), \ \ \  k=a, b,
\end{eqnarray}
Steps 1 establishes bounds for the intermediary estimates on each subsample. Step 2 establishes the result for the final estimator $\widehat \alpha_{ab}$ and consistency of the matrices estimates. Finally, Step 3  establishes (\ref{RateConditionSplitIV}).

Step 1. We have that
\begin{eqnarray*}
\sqrt{n_k}(\widehat \alpha_k - \alpha_0) &= & \Enk [\widehat D^k_i d_i^k{}']^{-1} \sqrt{n} \Enk [\widehat D^k_i \epsilon^k_i]  \\
& = & \{ \Enk [\widehat D^k_i d^k_i{}']\}^{-1} \left( {\Gnk} [D_i^{k} \epsilon_i^k{}] + o_P(1) \right)  \\
& = & \{  \Enk [D_i D_i'] + o_P(1)\}^{-1} \left( {\Gnk} [D_i^{k} \epsilon_i^k{}] + o_P(1) \right)
 \end{eqnarray*}
since
\begin{eqnarray}
 \Enk [\widehat D^k_id^k_i{}'] =  \Enk[D_iD_i'] + o_P(1) \label{eq: to show a} \\
{ \sqrt{n}\Enk} [\widehat D^k_i \epsilon^k_i] = {\Gnk} [D^k_i \epsilon^k_i] + o_P(1)\label{eq: to show b}.
 \end{eqnarray}
Indeed,  (\ref{eq: to show a}) follows similarly to Step 2 in the proof of Theorems 4 and 5 and condition (\ref{RateConditionSplitIV}).
The relation (\ref{eq: to show b}) follows from $\Ep[\epsilon^k_i|x_i^k]=0$ for both $k=a$ and $k= b$, Chebyshev inequality and
$$
E[ \| \sqrt{n}  \Enk  [(\widehat D^k_i - D^k_i) \epsilon^k_i] \|_2^2| x_i^k, i=1,...,n, k^c]  \lesssim   \sqrt{k_e}\max_{1\leq l \leq k_e}\|(\widehat D^k_{il} - D^k_{il})  \|_{2,n_k}^2 \to_P 0, $$
where $E[\cdot |x_i^k, i=1,...,n, k^c]$ denotes the estimate computed conditional $x_i^k, i=1,...,n$ and on the sample $k^c$, where $k^c= \{a,b\}\setminus k$.  The bound follows from the fact that (a)
$$ \widehat D^k_{il} - D^k_{il}= f(x_i^k)'(\widehat \beta^{k^c}_{l} - \beta^{k^c}_{0l}) - a_l(x^i_i),  1 \leq i \leq n_k, $$
by Condition AS, where $(\widehat \beta^{k^c}_{l} - \beta^{k^c}_{0l})$  are independent of $\{\epsilon^k_i, 1 \leq i \leq n_k\}$,
by the independence of the two subsamples $k$ and $k^c$,  (b)  $\{\epsilon^k_i, x_i, 1 \leq i \leq n_k\} $ are independent across $i$ and independent from the sample $k^c$, (c)
$\{\epsilon^k_i, 1 \leq i \leq n_k\} $ have conditional mean equal to zero, conditional on $x_i^k, i=1,...,n$
and have conditional variance bounded from above, uniformly in $n$, conditional on $x_i^k, i=1,...,n$, by Condition SM,
and (d) that  $\max_{1\leq l\leq k_e} \|\widehat D^k_{il} - D^k_{il} \|_{2,n_k} \to_P 0$.

Using the same arguments as in  Step 1 in the proof of Theorems 4 and 5,
\begin{eqnarray*}
\sqrt{n_k}(\widehat \alpha_k - \alpha_0)
& = & (\Enk [{D^k_i} {D^{k'}_i} ])^{-1} {\Gnk} [D_i^{k} \epsilon_i^k{}]  +  o_P(1) = O_P(1).
 \end{eqnarray*}

Step 2.  Now putting together terms we get
\begin{eqnarray*}
\sqrt{n}(\widehat \alpha_{ab} - \alpha_0) & = &
( (n_a/n)  \Ena [\hat D_i^a  \hat D_i^{a}{}']  +  (n_b/n) \Enb [\hat D_i^b  \hat D_i^{b}{}']   )^{-1} \times \\
& \times &  ( (n_a/n)  \Ena [\hat D_i^a  \hat D_i^{a}{}']  \sqrt{n} (\widehat \alpha_a - \alpha_0) +   (n_b/n) \Enb  [\hat D_i^b  \hat D_i^{b}{}'] \sqrt{n} (\widehat \alpha_b - \alpha_0)   ) \\
& = &
( (n_a/n)  \Ena [ D_i^a   D_i^{a}{}']  +  (n_b/n) \Enb [ D_i^b   D_i^{b}{}']   )^{-1} \times \\
& \times &  ( (n_a/n)  \Ena [ D_i^a  D_i^{a}{}']  \sqrt{n} (\widehat \alpha_a - \alpha_0) +   (n_b/n) \Enb  [ D_i^b   D_i^{b}{}'] \sqrt{n} (\widehat \alpha_b - \alpha_0)   ) + o_P(1) \\
& = & \{ \En [D_i D_i{}']\}^{-1} \times \{(1/\sqrt{2}) \times  {{\Gn}_a} [D_i^{a} \epsilon_i^a{}]
+ (1/\sqrt{2}) {{\Gn}_b} [D_i^{b} \epsilon_i^b {}]\}+ o_P(1) \\
& = & \{ \En [D_i D_i{}']\}^{-1} \times {{\Gn}} [D_i \epsilon_i] + o_P(1) \\
 \end{eqnarray*}
 where we are also using the fact that
$$
\Enk [\hat D_i^k  \hat D_i^{k}{}'] -  \Enk [ D_i^k  D_i^{k}{}'] = o_P(1),
$$
which is shown similarly to the proofs given in Theorem 4 for showing that $
\En [ D_i   D_i^{}{}'] -  \En[ D_i  D_i^{}{}'] = o_P(1).
$
The conclusion of the theorem follows by an application of Liapunov CLT, similarly to the proof in of Theorem 4 of the main text.

Step 3. In this step we establish (\ref{RateConditionSplitIV}).
For every observation $i$ and $l=1,\ldots,k_e$, by Condition AS we have
\begin{equation}\label{StateCondAS} D_{il} = f_i'\beta_{l0} + a_l(x_i), \ \ \|\beta_{l0}\|_0 \leq s,  \  \max_{1\leq{l}\leq k_e} \|a_{{l}}(x_i)\|_{2,n} \leq c_{s} \lesssim_{\mathrm{P}} \sqrt{s/n}.\end{equation}

Under our conditions, the sparsity bound for LASSO by Lemma 9 implies that for all $\delta^k_l = \widehat \beta^k_l - \beta_{l0}$, $k=a,b$, and $l=1,\ldots,k_e$
$$ \| \delta^k_l\|_0 \lesssim_P s.$$
Therefore, by condition SE, we have for $M=\Enk[f_i^kf_i^k{}']$, $k=a,b$, that with probability going to $1$, for $n$ large enough  $$0 < \kappa' \leq \semin{\|\delta^k_l\|_0}(M) \leq \semax{\|\delta^k_l\|_0}(M) \leq \kappa''<\infty,$$
where $\kappa'$ and $\kappa''$ are some constants that do not depend on $n$.
Thus,
\begin{equation}\label{ExtraStep} \begin{array}{rl}
\|D^k_{il} - \widehat D_{il}^k\|_{2,n_k} & = \|f_i^k{}'\beta_{l0} + a_l(x_i^k) - f_i^k{}'\widehat \beta^{k^c}_l\|_{2,n_k} \\
& = \|f_i^k{}'(\beta_{l0} - \widehat \beta^{k^c}_l)\|_{2,n_k} + \|a_l(x_i^k)\|_{2,n_k} \\
& \leq \sqrt{ \kappa''/\kappa' } \|f_i^{k^c}{}'(\widehat \beta^{k^c}_l- \beta_{l0})\|_{2,n_{k^c}} + c_s\sqrt{n/n_k}\\
\end{array}\end{equation}
where the last inequality holds with probability going to 1 by Condition SE imposed on matrices $M = \Enk[f_i^kf_i^k{}']$, $k=a,b$, and
by
$$
 \|a_l(x_i^k)\|_{2,n_k} \leq   \sqrt{ n/n_k}  \|a_l(x_i)\|_{2,n} \leq  \sqrt{ n/n_k} c_s.
$$
Then, in view of (\ref{StateCondAS}), $\sqrt{n/n_k} = \sqrt{2} + o(1)$, and condition $s\log p = o(n)$, the result (\ref{RateConditionSplitIV}) holds by (\ref{ExtraStep}) combined with Theorem 1 for LASSO and Theorem 2 for Post-LASSO which imply
$$\max_{1\leq l\leq k_e} \|f_i^k{}'(\widehat \beta^{k}_l- \beta_{l0})\|_{2,n_{k}} \lesssim_P \sqrt{\frac{s\log (p\vee n)}{n}}= o_P(1), \ \ \ k=a,b.$$
\qed

\section{Proof of Lemma 3.}

Part 1. The first condition in RF(iv) is assumed, and we omit the proof for the third condition since it is analogous to the proof for the second condition.

Note that $\max_{1\leq j \leq p} \En[f_{ij}^4v_{il}^4] \leq (\En[v_{il}^8])^{1/2} \max_{1\leq j \leq p} (\En[f_{ij}^8])^{1/2} \lesssim_P 1$ since $\max_{ 1 \leq j \leq p } \sqrt{\En [f_{ij}^8]} \lesssim_{\mathrm{P}} 1$ by assumption and $\max_{1 \leq l \leq k_e} \sqrt{\En[v_{il}^8]} \lesssim_{\mathrm{P}} 1$ by the bounded $k_e$, Markov inequality, and the assumption that $\barEp[v_{il}^8]$ are uniformly bounded in $n$ and $l$.

Thus, applying  Lemma \ref{Thm:masterSym},
$$
\begin{array}{rl}
\displaystyle \max_{ 1 \leq l \leq k_e, 1 \leq j \leq p } | \En [f_{ij}^2 v_{il}^2] - \barEp[f_{ij}^2 v_{il}^2 ]  |  &\displaystyle  \lesssim_{\mathrm{P}} \sqrt{\frac{\log (pk_e)}{n}} \lesssim_{\mathrm{P}} \sqrt{\frac{\log p}{n}} \to_{\mathrm{P}} 0.\end{array}
$$

Part 2. To show (1), we note that by simple union bounds and
tail properties of Gaussian variable, we have that  $\max_{ij}f_{ij}^2 \lesssim_{\mathrm{P}} \log (p\vee n)$,
so we need $\log (p\vee n) \frac{s \log (p\vee n)}{n} \to 0$.

Therefore $\max_{1\leq j \leq p} \En[f_{ij}^4v_{il}^4] \leq \En[v_{il}^4] \max_{1\leq i\leq n, 1\leq j \leq p} f_{ij}^4\lesssim_P \log^2 (p\vee n)$. Thus, applying  Lemma \ref{Thm:masterSym},
$$
\begin{array}{rl}
\displaystyle \max_{ 1 \leq l \leq k_e, 1 \leq j \leq p } | \En [f_{ij}^2 v_{il}^2] - \barEp[f_{ij}^2 v_{il}^2 ]  |  &\displaystyle  \lesssim_{\mathrm{P}} \log (p\vee n) \ \sqrt{\frac{\log (pk_e)}{n}} \lesssim_{\mathrm{P}} \sqrt{\frac{\log^3 (p\vee n)}{n}} \to_{\mathrm{P}} 0\end{array}
$$ since $\log p = o(n^{1/3})$.
The remaining moment conditions of RF(ii) follows immediately from the definition of the conditionally bounded moments
since for any $m>0$, $\barEp[|f_{ij}|^m]$ is bounded, uniformly in $1 \leq j \leq p$, uniformly in $n$,
for the i.i.d. Gaussian regressors of Lemma 1 of \citen{BCCH-IV}. The proof of (2) for arbitrary bounded i.i.d. regressors of Lemma 2 of \citen{BCCH-IV} is similar.

\section{Proof of Lemma 4.}
The first two conditions of SM(iii) follows from the assumed rate $s^2\log^2(p\vee n) = o(n)$ since we have $q_\epsilon = 4$.
 To show part (1), we note that by simple union bounds and
tail properties of Gaussian variable, we have that  $\max_{1\leq i\leq n, 1\leq j\leq p}f_{ij}^2 \lesssim_{\mathrm{P}} \log (p\vee n)$,
so we need $\log (p\vee n) \frac{s \log (p\vee n)}{n} \to 0$. Applying union
bound, Gaussian concentration inequalities \citen{LedouxTalagrandBook}, and that $\log^2 p = o(n)$,
we have $\max_{1\leq j\leq p}\En[f^4_{ij}] \lesssim_{\mathrm{P}} 1$. Thus SM(iii)(c) holds by $\max_{j}\En[f^2_{ij}\epsilon_i^2]\leq (\En[\epsilon_i^4])^{1/2}\max_{1\leq j\leq p}(\En[f^4_{ij}])^{1/2}\lesssim_P 1$. Part (2) follows because regressors are bounded and the moment assumption on $\epsilon$. \qed

\comment{
Note that $P( \En[|f_{ij}^k|] > M ) = P( \|f_{\cdot j}\|_k^k > Mn ) = P(\|f_{\cdot j}\|_k > (Mn)^{1/k})$.

Since $| \|f\|_k - \|g\|_k | \leq \|f-g\|_k \leq \|f-g\|$, we have that $\|\cdot\|_k$ is 1-Lipschitz for $k\geq 2$. Moreover, $$\Ep[\|f_{\cdot j}\|_k] \leq (\Ep[\|f_{\cdot j}\|_k^k])^{1/k} = (\sum_{i=1}^n\Ep[|f_{ij}^k|])^{1/k}= n^{1/k}(\Ep[|f_{1j}^k|])^{1/k}$$
$$  = n^{1/k}\{\sigma^k 2^{k/2}\Gamma((k+1)/2)\}^{1/k} \leq n^{1/k}\sigma \sqrt{2} e$$

By \citen{LedouxTalagrandBook}, page 21 equation (1.6), we have
$$P( \|f_{\cdot j}\|_k  > (Mn)^{1/k} -  \Ep[\|f_{\cdot j}\|_k] ) \leq 2\exp(-\{(Mn)^{1/k} -  \Ep[\|f_{\cdot j}\|_k] \}^2/2).$$
Setting $M := \{\sigma \sqrt{2} e  + n^{-1/k}\sqrt{2\log (2p/\alpha)}  \}^k$, so that $(Mn)^{1/k} = n^{1/k}\sigma \sqrt{2} e + \sqrt{2\log (2p/\alpha)}$ we have
$$P( \max_{1\leq j\leq p} \En[|f_{ij}^k|] \geq M ) \leq \alpha.$$
Note that for fixed $k$, $M \lesssim 1$ provided that $n^{-1/k}\sqrt{2\log (2p/\alpha)}$ is bounded.

For $k=4$, it suffices $\log^{k/2}(p/\alpha) = o(n)$.

}

\section{Proof of Lemma 11.}

Step 1. Here we consider the initial option, in which
$\hat \gamma_{jl}^2 = \En [f_{ij}^2 (d_{il} - \En d_{il})^2]$. Let us define $\tilde d_{il} = d_{il} - \barEp [d_{il}] $, $\tilde \gamma_{jl}^2 = \En [f_{ij}^2 \tilde d_{il}^2]$ and  $\gamma_{jl}^2 = \barEp [f_{ij}^2 \tilde d_{il}^2].$
We want to show that
$$
\Delta_1 = \max_{ 1 \leq l \leq k_e, 1 \leq j \leq p }|\hat \gamma^2_{jl} - \tilde\gamma^2_{jl}| \to_{\mathrm{P}} 0, \ \Delta_2=\max_{ 1 \leq l \leq k_e, 1 \leq j \leq p } |\tilde \gamma^2_{jl} - \gamma^2_{jl}| \to _{\mathrm{P}} 0,
$$
which would imply that
$
\max_{1\leq j\leq p, 1 \leq l \leq k_e}|\hat \gamma^2_{jl} -  \gamma^2_{jl}| \to_{\mathrm{P}} 0$ and then since
$ \gamma_{jl}^2$'s are uniformly bounded from above by Condition RF and bounded below by $\gamma^{02}_{jl} = \barEp[f_{ij}^2 v^2_{il}]$, which are bounded away from zero. The asymptotic validity of the initial option then follows.

We have that $\Delta_2 \to_P 0$ by Condition RF, and, since $\En[\tilde d_{il}] = \En[d_{il}] - \barEp[d_{il}]$, we have
$$
\begin{array}{rl}
\Delta_1 & = \max_{ 1 \leq l \leq k_e, 1 \leq j \leq p }|\En [f_{ij}^2 \{(\tilde d_{il} - \En \tilde d_{il} )^2- \tilde d_{il}^2\}]| \\
& \leq  \max_{ 1 \leq l \leq k_e, 1 \leq j \leq p } 2 | \En[f_{ij}^2 \tilde d_{il}] \En [\tilde d_{il}] | + \max_{ 1 \leq l \leq k_e, 1 \leq j \leq p } | \En[f_{ij}^2]  (\En \tilde d_{il})^2| \to_{\mathrm{P}} 0.
\end{array}$$
Indeed, we have for the first term that,
$$\max_{ 1 \leq l \leq k_e, 1 \leq j \leq p }  | \En[f_{ij}^2 \tilde d_{il}]| \En [\tilde d_{il}] | \leq \max_{i \leq n ,j \leq p }|f_{ij}| \sqrt{ \En[f^2_{ij} \tilde d^2_{il}]} O_{\Pr}(1/\sqrt{n}) \to_{\Pr} 0 $$
where we first used Holder inequality and  $\max_{1\leq l\leq k_e} |\En [ \tilde d_{il}] | \lesssim_{\mathrm{P}} \sqrt{k_e}/\sqrt{n}$ by the Chebyshev inequality and by $\barEp[\tilde d^2_{il}]$ being uniformly bounded by Condition RF; then we invoked Condition RF to claim convergence to zero.  Likewise, by Condition RF,
$$
\max_{ 1 \leq l \leq k_e, 1 \leq j \leq p } | \En[f_{ij}^2]  (\En \tilde d_{il})^2| \leq  \max_{1 \leq j \leq p } |f_{ij}^2| O_{\Pr}(1/n) \to_{\Pr} 0.
$$
Step 2.  Here we consider the refined option, in which
$\hat \gamma_{jl}^2 = \En [f_{ij}^2 \widehat v^2_{il}].$
The residual here $\widehat v_{il} = d_{il} - \hat D_{il}$ can be based on any estimator that
obeys
\begin{equation}\label{FitRate}
\max_{1\leq l \leq k_e}\|\hat D_{il} - D_{il}\|_{2,n} \lesssim_{\mathrm{P}} \sqrt{\frac{s \log (p\vee n)}{n}}.
\end{equation}
Such estimators include the Lasso and Post-Lasso estimators based on the initial option.
Below we establish that the penalty levels, based on the refined option
using any estimator obeying (\ref{FitRate}), are asymptotically valid.
Thus by Theorems \ref{Thm:RateLASSO} and \ref{Thm:RatesPostLASSO}, the Lasso and Post-Lasso estimators based on the
refined option also obey (\ref{FitRate}). This, establishes
that we can iterate on the refined option a bounded number of times, without
affecting the validity of the approach.

Recall that  $\hat \gamma_{jl}^{02} = \En[f_{ij}^2 v^2_{il}]$ and define $\gamma_{jl}^{02} := \barEp[f_{ij}^2 v^2_{il}]$,
which is bounded away from zero and from above by assumption. Hence, by Condition RF, it suffices to show that
$
\max_{1\leq j \leq p, 1\leq l\leq k_e}|\hat \gamma^2_{jl} -  \gamma^{02}_{jl}| \to_{\mathrm{P}} 0$ which implies the loadings are asymptotic valid with $u'=1$. This in turn follows
from $$
\Delta_1 = \max_{ 1 \leq l \leq k_e, 1 \leq j \leq p }|\hat \gamma^2_{jl} - \hat \gamma^{02}_{jl}| \to_{\mathrm{P}} 0, \ \Delta_2=\max_{ 1 \leq l \leq k_e, 1 \leq j \leq p } |\hat \gamma^{02}_{jl} - \gamma^{02}_{jl}|^2 \to _{\mathrm{P}} 0,
$$
which we establish below.

Now note that we have proven $\Delta_2 \to _{\mathrm{P}} 0$ in the Step 3 of the proof of Theorem 1.  As for $\Delta_1$ we note
that
$$
\Delta_1 \leq 2 \max_{ 1 \leq l \leq k_e, 1 \leq j \leq p }|  \En[f^2_{ij} v_{jl} (\hat D_{il} - D_{il})]  | +  \max_{ 1 \leq l \leq k_e, 1 \leq j \leq p }\En[ f_{ij}^2 (\hat D_{il} - D_{il})^2].
$$
The first term is bounded by Holder and Liapunov inequalities,
$$\max_{ i \leq n, j \leq p} |f_{ij}|  ( \En [f_{ij}^2 v^2_{il}] )^{1/2}  \max_{ l \leq k_e} \| \hat D_{il} - D_{il}\|_{2,n} \lesssim_{\mathrm{P}} \max_{ i \leq n,  j \leq p} |f_{ij}|  ( \En [f_{ij}^2 v^2_{il}] )^{1/2}  \sqrt{ \frac{s \log (p\vee n)}{n}} \to_{\Pr} 0.
$$
where the conclusion is by Condition RF. The second term is of stochastic order
$$
\max_{ i \leq n, j \leq p }|f_{ij}^2| \frac{s \log (p\vee n)}{n} \to_{\mathrm{P}} 0,$$ which converges to zero by Condition RF.\qed

\section{Additional Simulation Results}

In this appendix, we present simulation results to complement the results given in the paper.  The simulations use the same model as the simulations in the paper:
\begin{align*}
\begin{array}{ll}
y_i &= \beta x_i + e_i \\
x_i &= z_i'\Pi + v_i
\end{array}  \ \ \ \ \ \
(e_i,v_i) &\sim N\(0,\(\begin{array}{cc} \sigma^2_e & \sigma_{ev} \\ \sigma_{ev} & \sigma^2_{v}\end{array}\)\) \ \text{i.i.d.}
\end{align*}
where $\beta=1$ is the parameter of interest, and $z_i = (z_{i1},z_{i2},...,z_{i100})' \sim N(0,\Sigma_Z)$ is a 100 x 1 vector with $E[z_{ih}^2] = \sigma^2_z$ and $Corr(z_{ih},z_{ij}) = .5^{|j-h|}$.  In all simulations, we set $\sigma^2_e = 2$ and $\sigma^2_z = 0.3$.

For the other parameters, we consider various settings.  We provide results for sample sizes, $n$, of 100, 250, and 500; and we consider three different values for $Corr(e,v)$: 0, .3, and .6.  We also consider four values of $\sigma^2_v$ which are chosen to benchmark four different strengths of instruments.  The four values of $\sigma^2_v$ are found as $\sigma^2_v = \frac{n \Pi'\Sigma_Z\Pi}{F^*\Pi'\Pi}$ for $F^*$: 2.5, 10, 40, and 160.  We use two different designs for the first-stage coefficients, $\Pi$.  The first sets the first $S$ elements of $\Pi$ equal to one and the remaining elements equal to zero.  We refer to this design as the ``cut-off'' design.  The second model sets the coefficient on $z_{ih} = .7^{h-1}$ for $h=1,...,100$.  We refer to this design as the ``exponential'' design.  In the cut-off case, we consider $S$ of 5, 25, 50, and 100 to cover different degrees of sparsity.

For each setting of the simulation parameter values, we report results from seven different estimation procedures.  A simple possibility when presented with many instrumental variables is to just estimate the model using 2SLS and all of the available instruments.  It is well-known that this will result in poor-finite sample properties unless there are many more observations than instruments; see, for example, \citen{bekker}.  The limited information maximum likelihood estimator (LIML) and its modification by \citen{fuller} (FULL)\footnote{\citen{fuller} requires a user-specified parameter.  We set this parameter equal to one which produces a higher-order unbiased estimator.} are both robust to many instruments as long as the presence of many instruments is accounted for when constructing standard errors for the estimators; see \citen{bekker} and \citen{hhn:weakiv} for example.  We report results for these estimators in rows labeled 2SLS(100), LIML(100), and FULL(100) respectively.\footnote{With $n = 100$, we randomly select 99 instruments for use in FULL(100) and LIML(100).}  For LASSO, we consider variable selection based on two different sets of instruments.  In the first scenario, we use LASSO to select among the base 100 instruments and report results for the IV estimator based on the LASSO (LASSO) and Post-LASSO (Post-LASSO) forecasts.  In the second, we use LASSO to select among 120 instruments formed by augmenting the base 100 instruments by the first 20 principle components constructed from the sampled instruments in each replication.  We then report results for the IV estimator based on the LASSO (LASSO-F) and Post-LASSO (Post-LASSO-F) forecasts.  In all cases, we use the refined data-dependent penalty loadings given in the paper.\footnote{Specifically, we start by finding a value of $\lambda$ for which LASSO selected only one instrument.  We used this instrument to construct an initial set of residuals for use in defining the refined penalty loadings.  Using these loadings, we run another LASSO step and select a new set of instruments.  We then compute another set of residuals and use these residuals to recompute the loadings.  We then run a final LASSO step.  We report the results for the IV estimator of $\beta$ based on the instruments chosen in this final LASSO-step.}  For each estimator, we report root-truncated-mean-squared-error (RMSE),\footnote{We truncate the squared error at 1e12.} median bias (Med. Bias), median absolute deviation (MAD), and rejection frequencies for 5\% level tests (rp(.05)).\footnote{In cases where LASSO selects no instruments, Med. Bias, and MAD use only the replications where LASSO selects a non-empty set of instruments, and we set the confidence interval eqaul to $(-\infty,\infty)$ and thus fail to reject.}  For computing rejection frequencies, we estimate conventional 2SLS standard errors for 2SLS(100), LASSO, and Post-LASSO, and the many instrument robust standard errors of \citen{hhn:weakiv} for LIML(100) and FULL(100).

\bibliographystyle{econometrica}

\bibliography{mybib}

\begin{figure}
    \includegraphics[width=\textwidth]{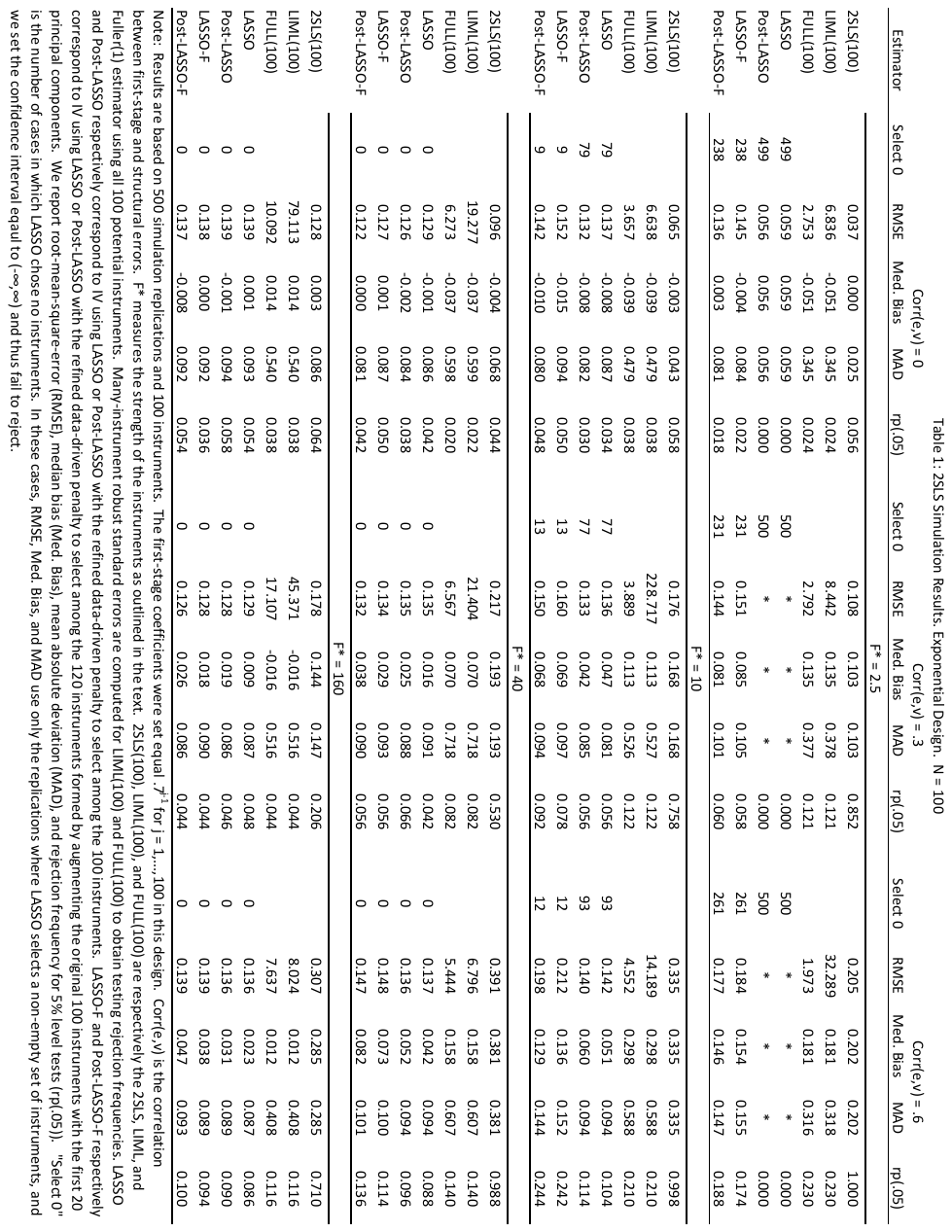}
    \label{fig:table1}
\end{figure}

\begin{figure}
    \includegraphics[width=\textwidth]{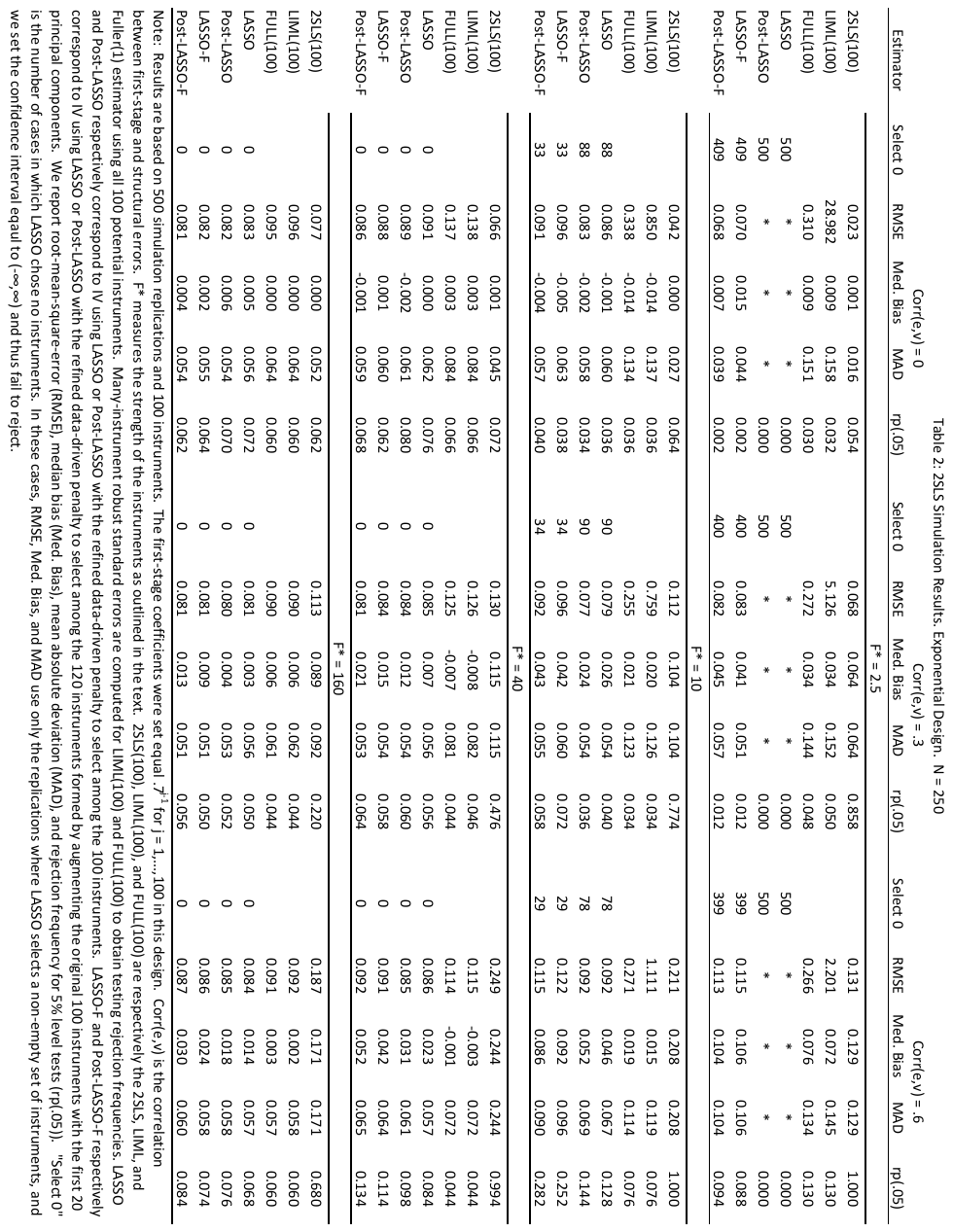}
    \label{fig:table2}
\end{figure}

\begin{figure}
    \includegraphics[width=\textwidth]{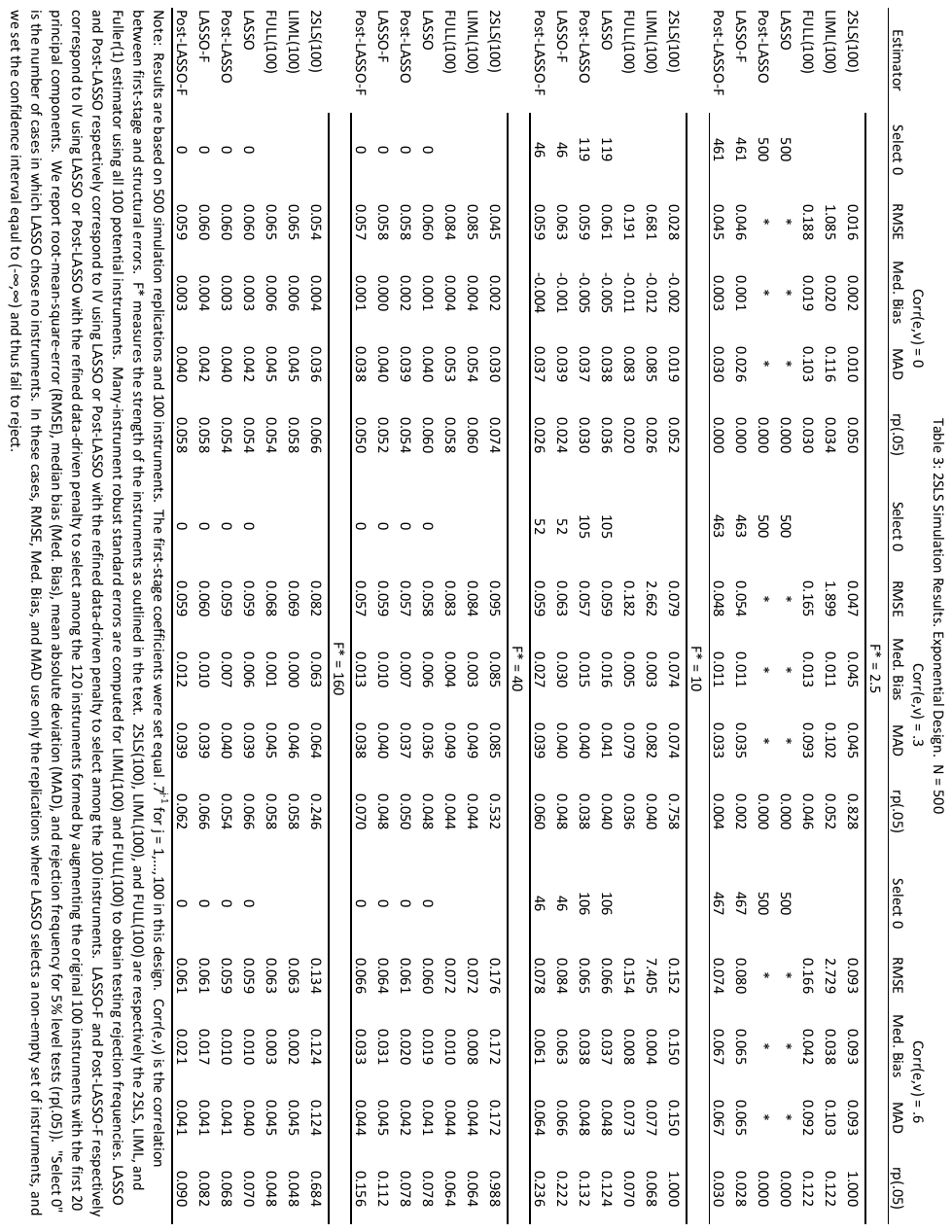}
    \label{fig:table3}
\end{figure}

\begin{figure}
    \includegraphics[width=\textwidth]{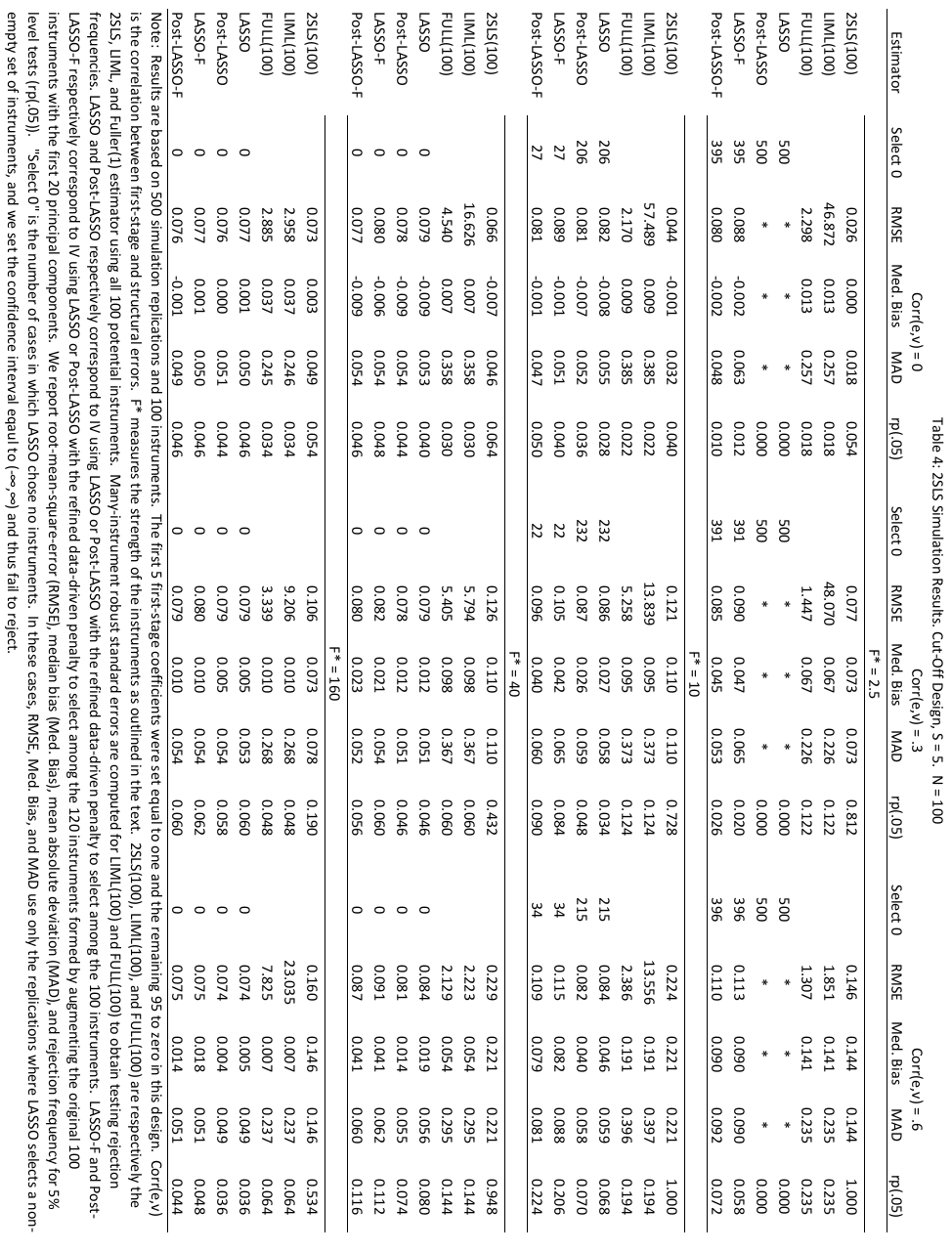}
    \label{fig:table4}
\end{figure}

\begin{figure}
    \includegraphics[width=\textwidth]{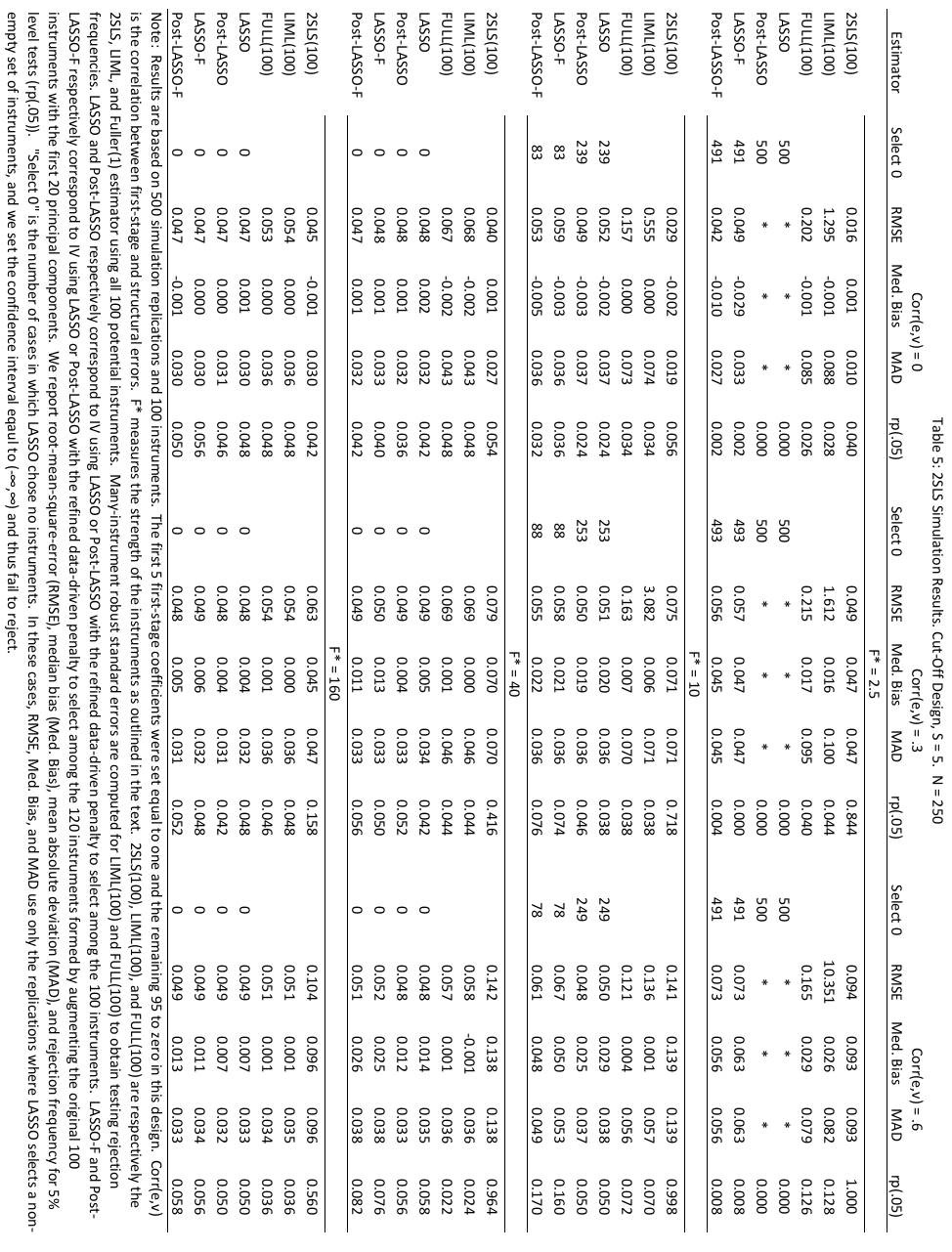}
    \label{fig:table5}
\end{figure}

\begin{figure}
    \includegraphics[width=\textwidth]{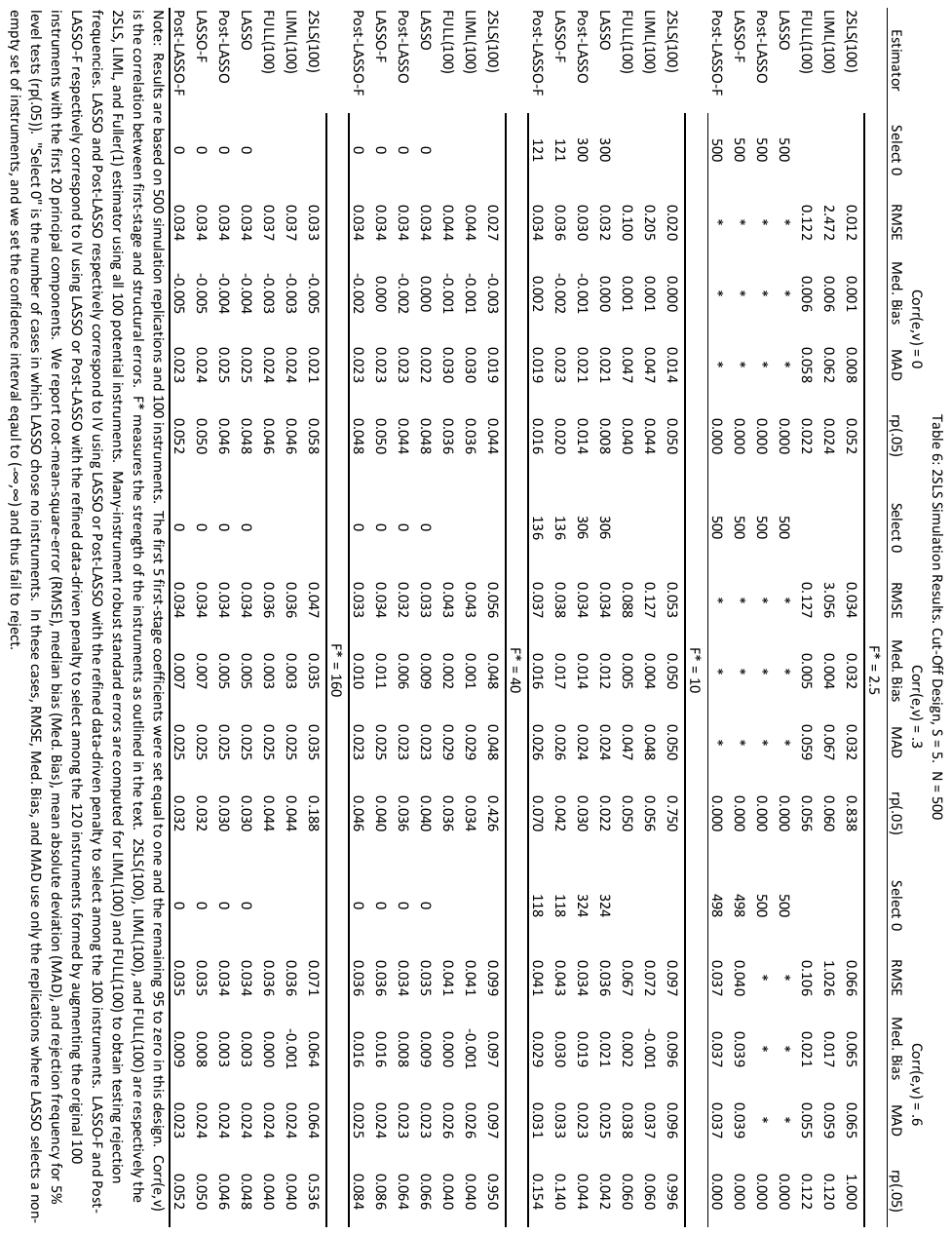}
    \label{fig:table6}
\end{figure}

\begin{figure}
    \includegraphics[width=\textwidth]{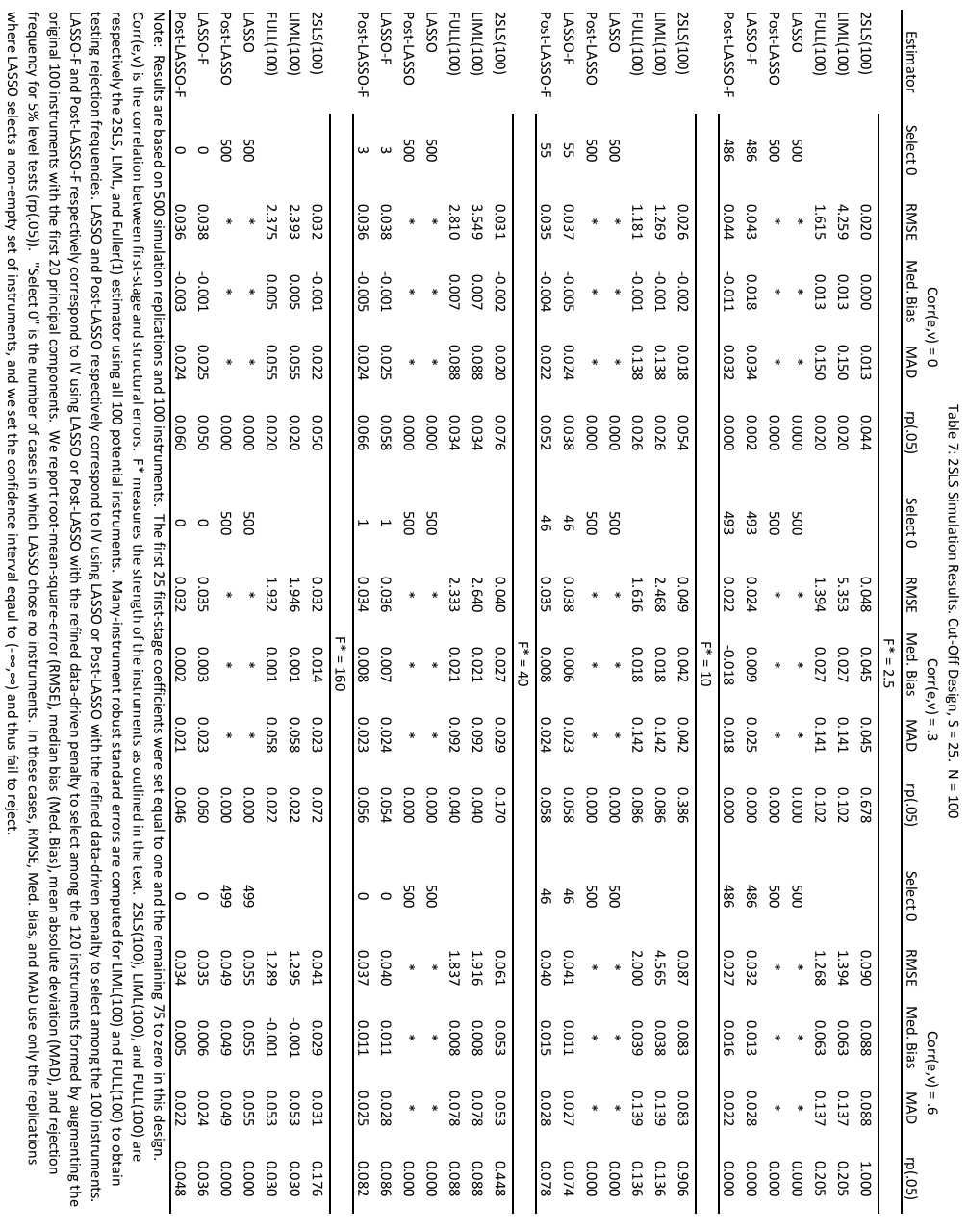}
    \label{fig:table7}
\end{figure}

\begin{figure}
    \includegraphics[width=\textwidth]{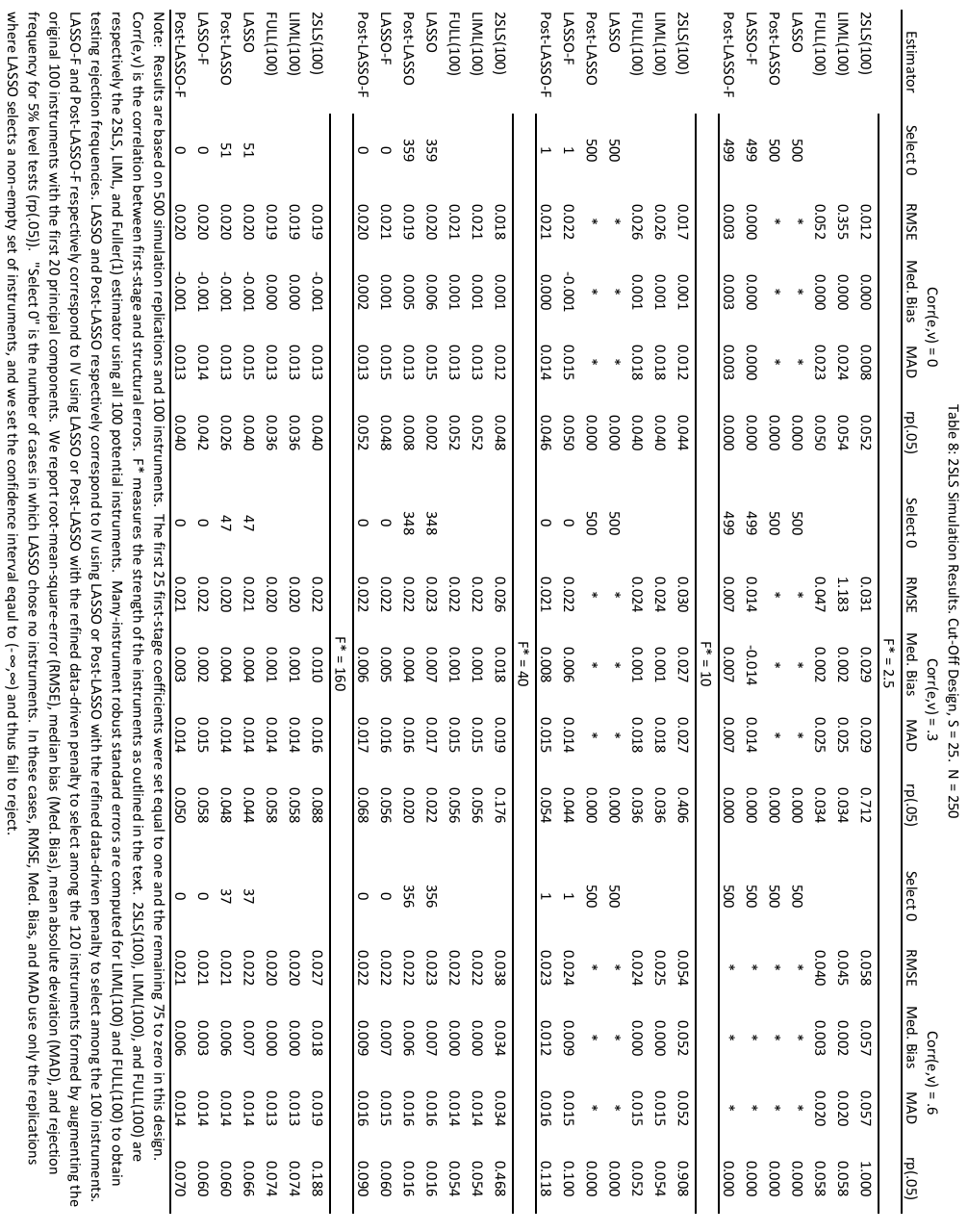}
    \label{fig:table8}
\end{figure}

\begin{figure}
    \includegraphics[width=\textwidth]{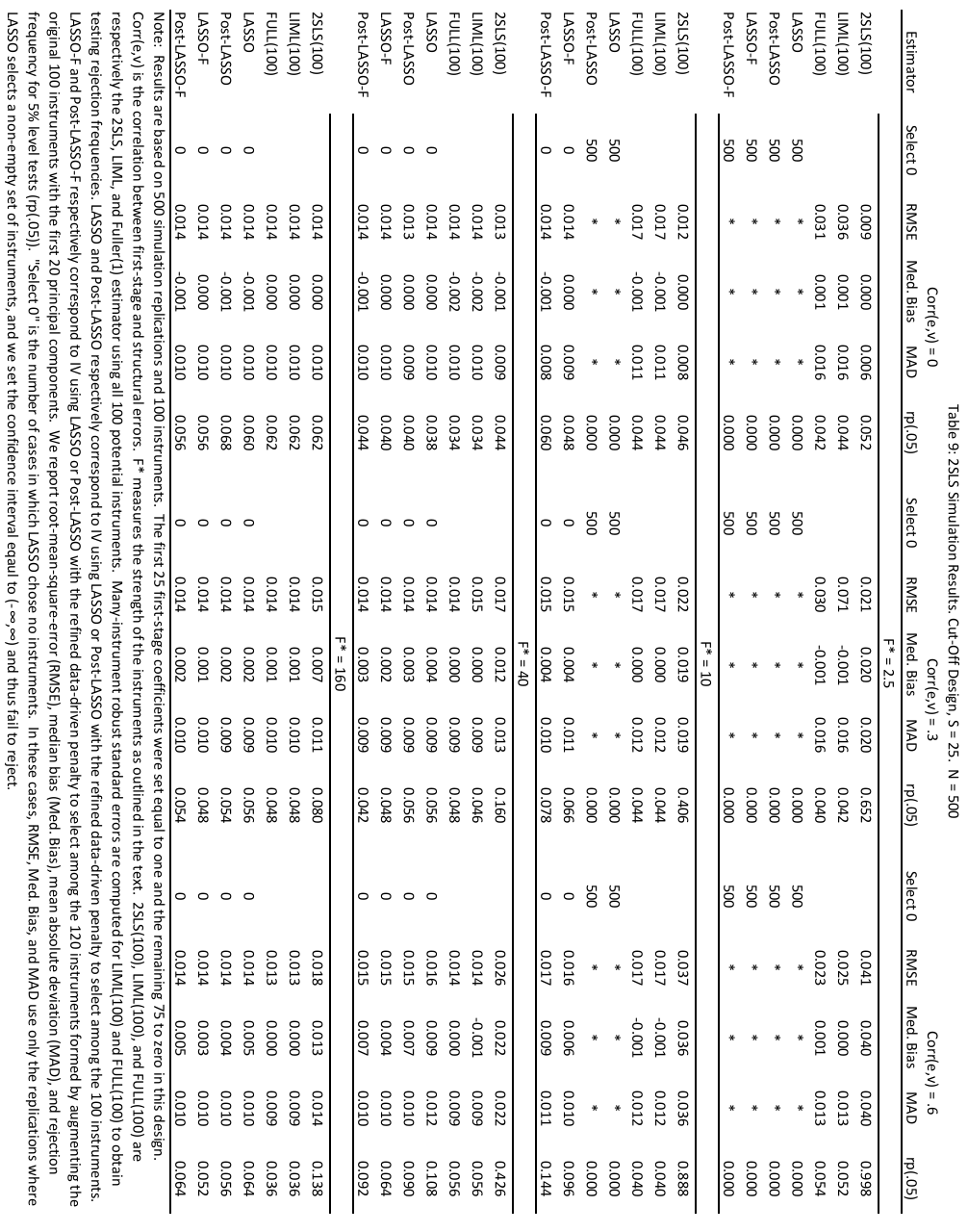}
    \label{fig:table9}
\end{figure}

\begin{figure}
    \includegraphics[width=\textwidth]{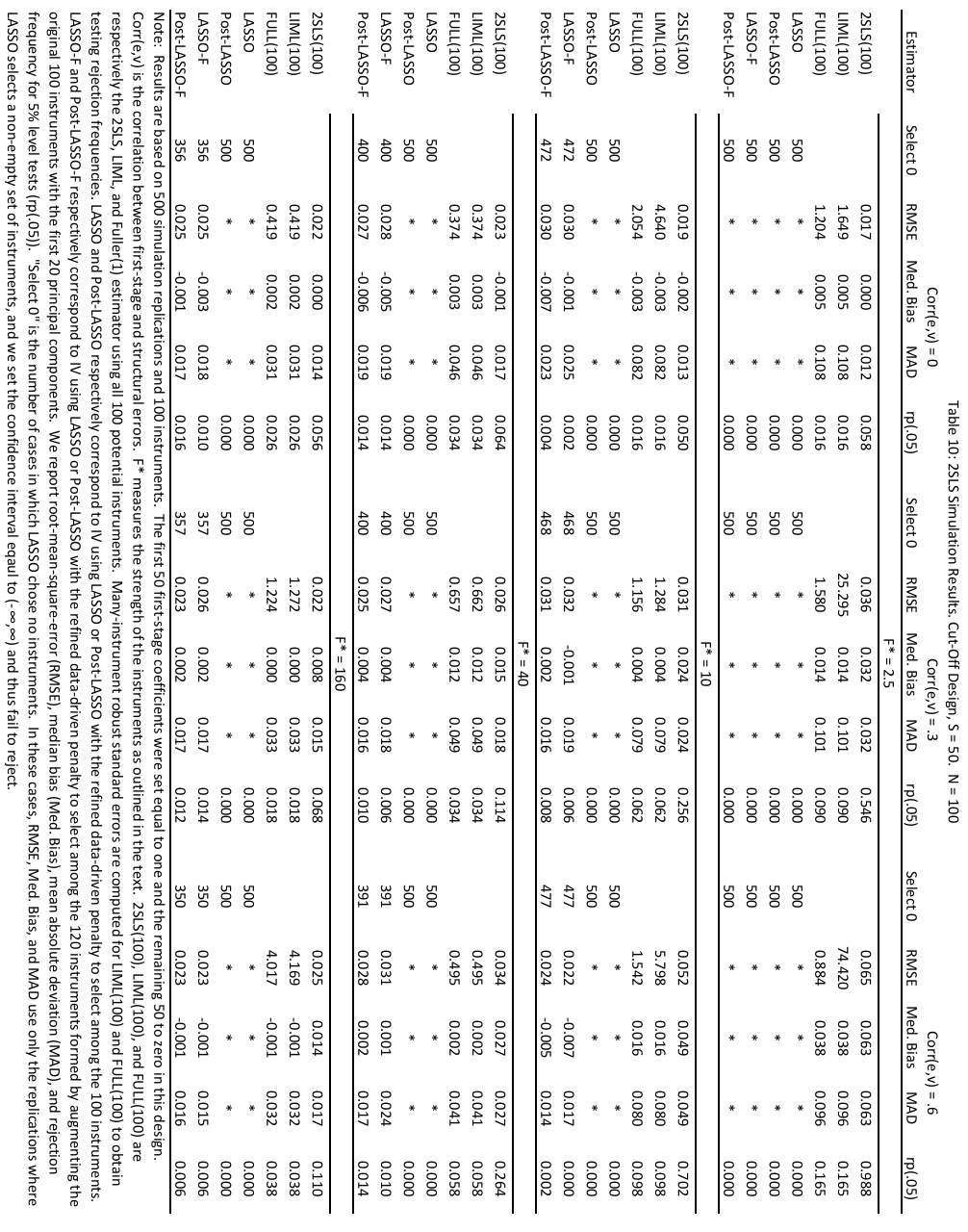}
    \label{fig:table10}
\end{figure}

\begin{figure}
    \includegraphics[width=\textwidth]{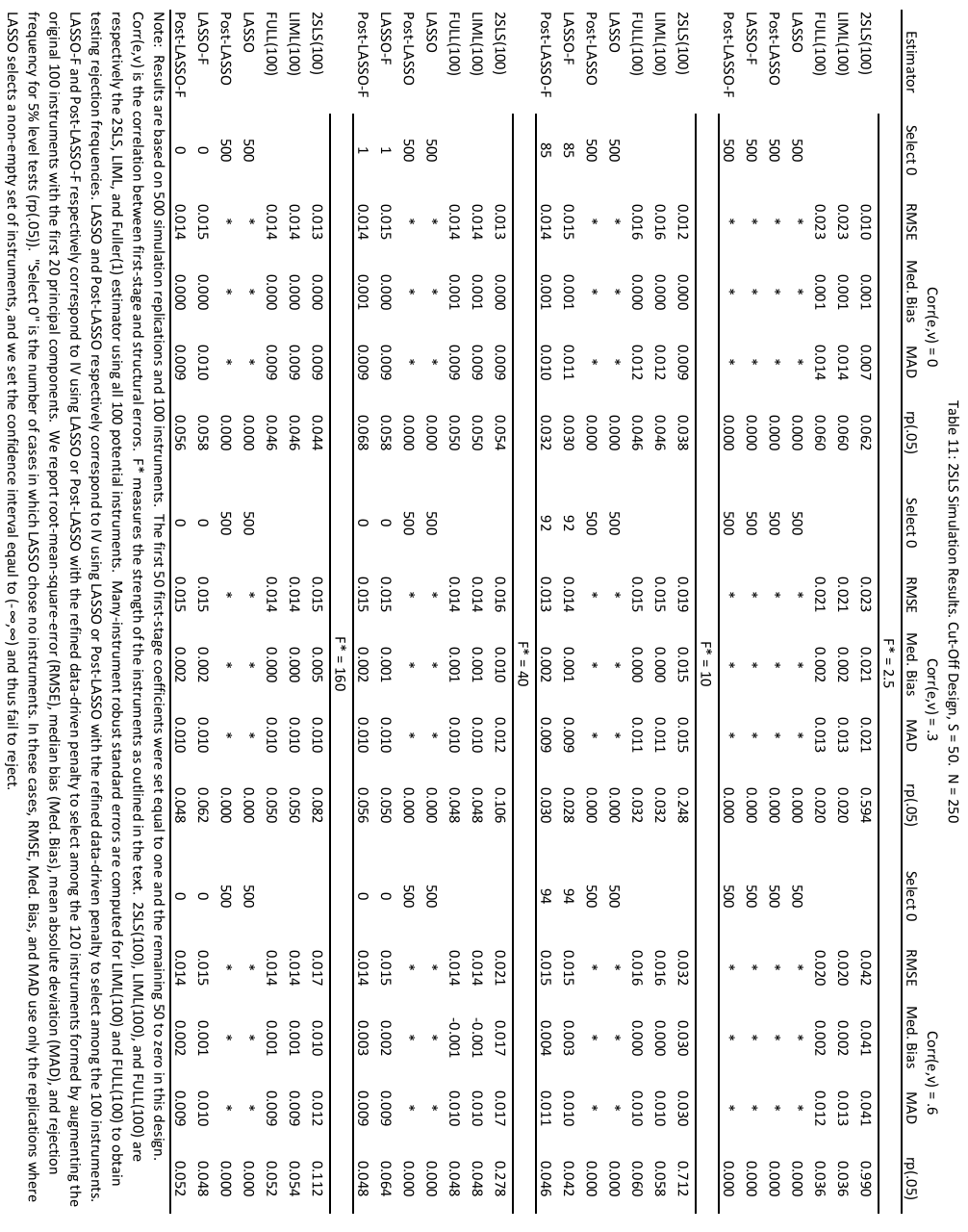}
    \label{fig:table11}
\end{figure}

\begin{figure}
    \includegraphics[width=\textwidth]{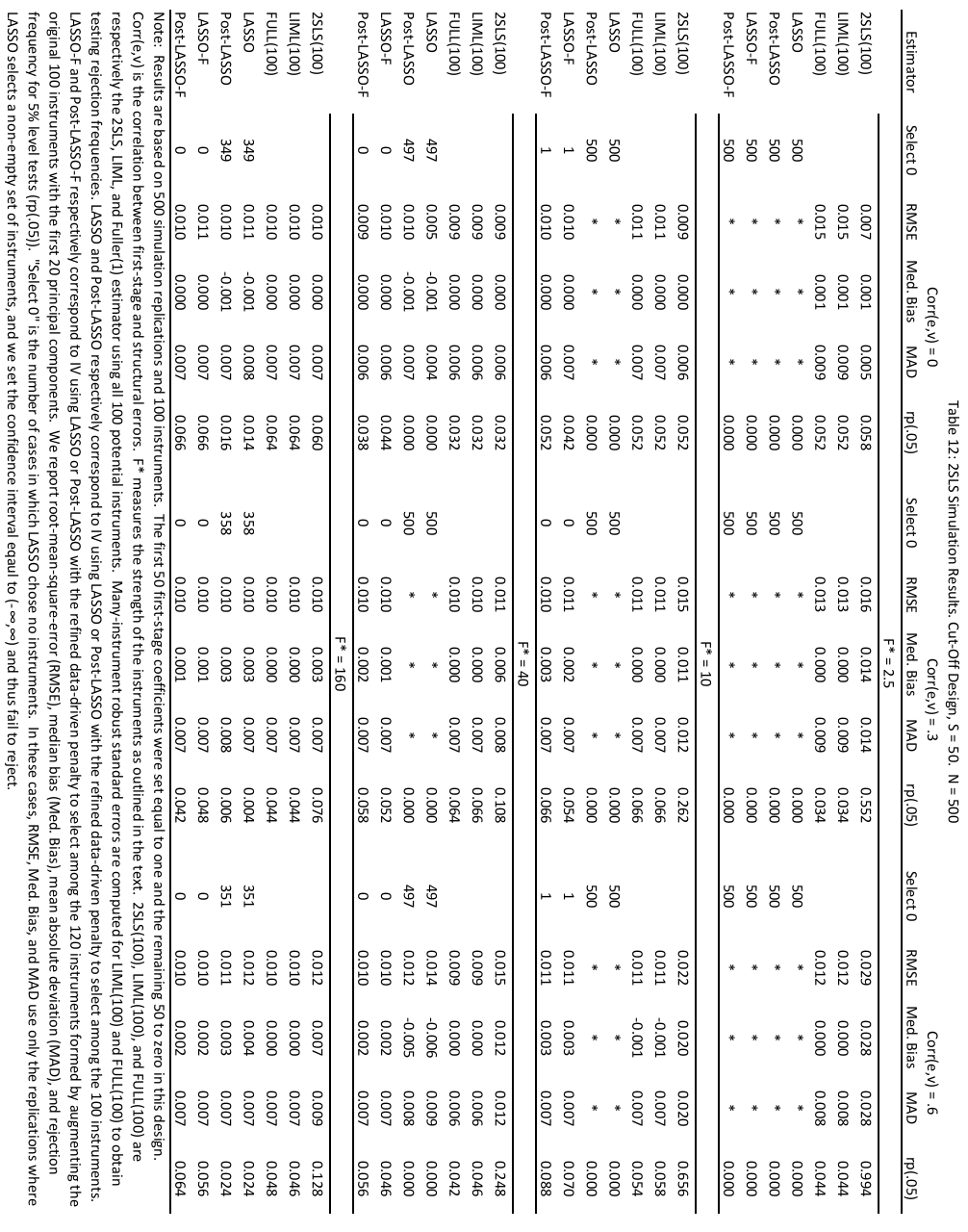}
    \label{fig:table12}
\end{figure}

\begin{figure}
    \includegraphics[width=\textwidth]{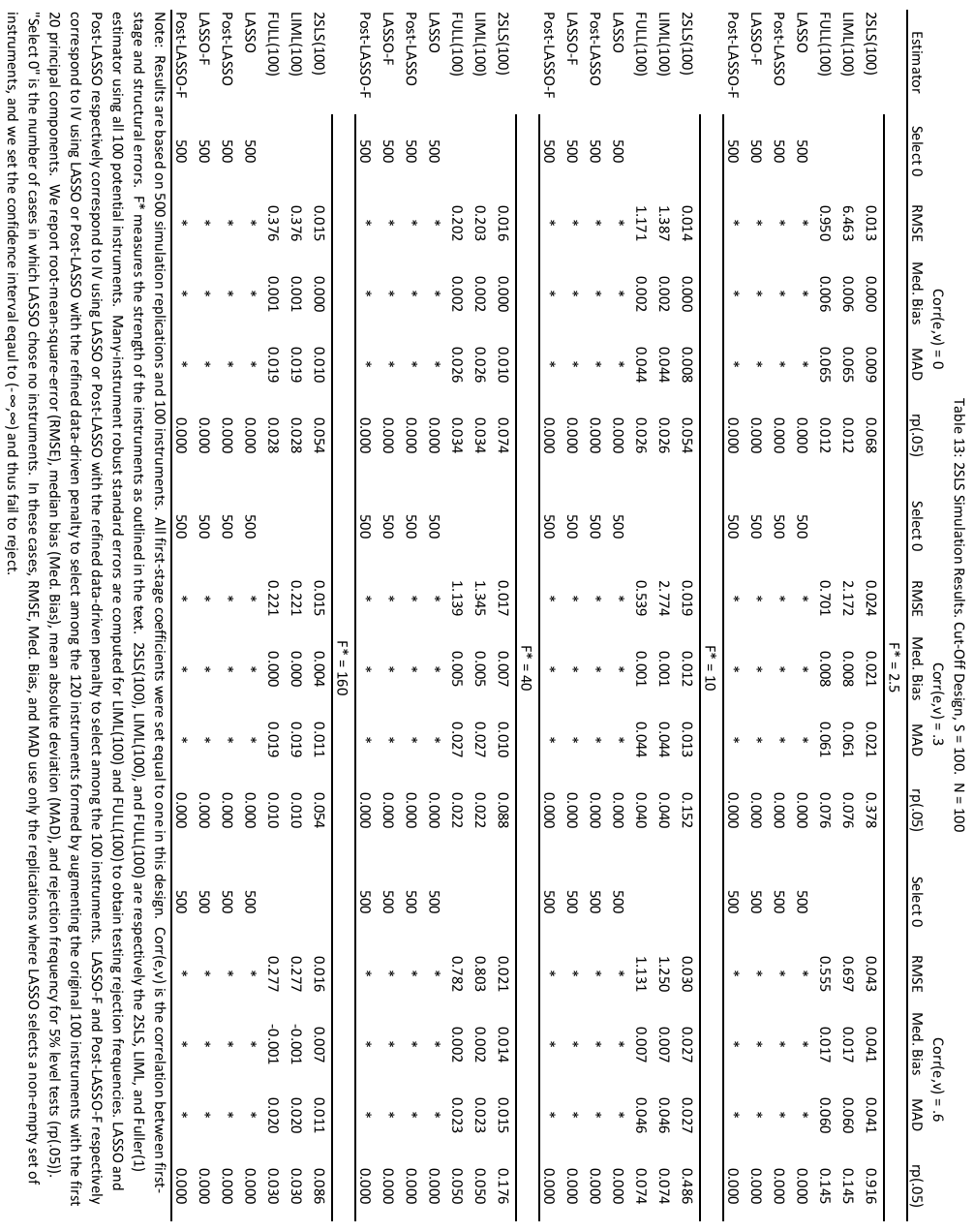}
    \label{fig:table13}
\end{figure}

\begin{figure}
    \includegraphics[width=\textwidth]{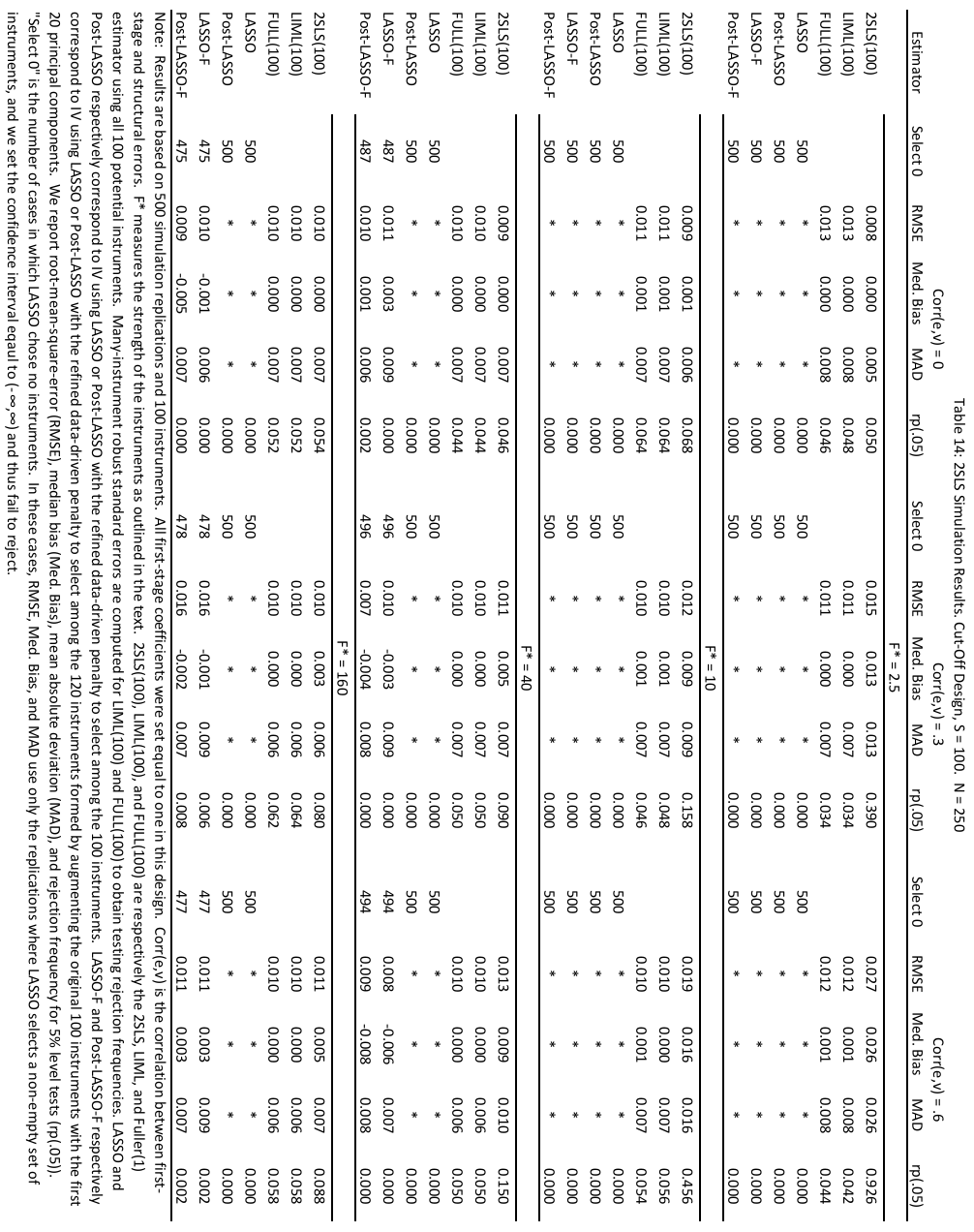}
    \label{fig:table14}
\end{figure}

\begin{figure}
    \includegraphics[width=\textwidth]{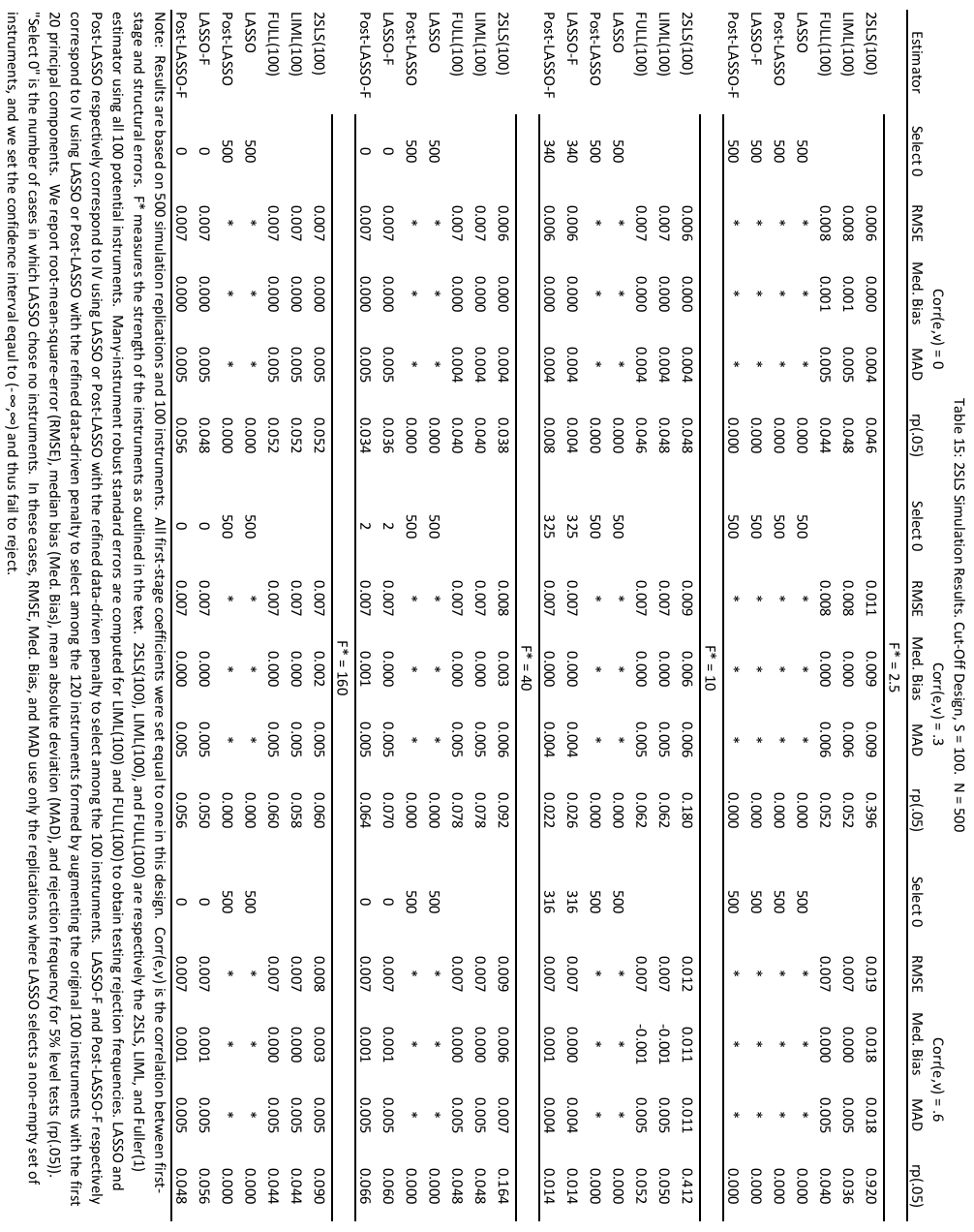}
    \label{fig:table15}
\end{figure}

\end{document}